\documentclass[12pt, a4paper]{article}
\usepackage{xprintlen}
\usepackage{float}

\usepackage[utf8]{inputenc}
\usepackage{amsmath}
\usepackage{amsfonts,dsfont} 
\usepackage{amssymb,amsthm}
\usepackage{mathtools}
\usepackage{mathrsfs}
\usepackage{graphicx}
\usepackage[margin=2cm]{geometry}

\usepackage{lipsum}

\usepackage[table,usenames,dvipsnames]{xcolor}
\usepackage[
    colorlinks=true,
    hypertexnames=false
    ]{hyperref}
\PassOptionsToPackage{linktocpage}{hyperref} 

\hypersetup{
  colorlinks   = true, %Colour links instead of ugly boxes
  urlcolor = magenta!80!black,
  linkcolor    = blue!60!black, %Colour of internal links
  citecolor=black!60 %Colour of citations
}
\usepackage[all]{hypcap}
\usepackage{url}
\newcommand{\ket}[1]{|#1\rangle}
\newcommand{\bra}[1]{\langle#1|}

\newcommand{\ketbra}[2]{\rangle#1|#2\langle}
\usepackage{textcomp, complexity}
\usepackage{caption,subcaption}
\usepackage{tikz}
\usetikzlibrary{quantikz}
\usepackage{physics}
\usepackage{tabularx,environ}
\usepackage{authblk}
\usepackage{booktabs,multirow}
\usepackage{mathpazo}
\usepackage{enumerate}

\let\oldnl\nl% Store \nl in \oldnl
\newcommand{\nonl}{\renewcommand{\nl}{\let\nl\oldnl}}% Remove line number for one line

\usepackage[capitalise,compress]{cleveref}
\newtheoremstyle{mystyle}%
{3pt}% Space above
{3pt}% Space below 
{\upshape}% Body font
{}% Indent amount
{\bfseries}% Theorem head font
{.}% Punctuation after theorem head
{.5em}% Space after theorem head
{}% Theorem head spec (can be left empty, meaning â€˜normalâ€™)

\crefname{thm}{Theorem}{Theorems}
\Crefname{thm}{Theorem}{Theorems}
\newtheorem{thm}{Theorem}
\newtheorem{propo}{Proposition}
\newtheorem*{thm*}{Theorem}
\newtheorem{lemma}[thm]{Lemma}

\newtheorem{corll}[thm]{Corollary}

\theoremstyle{remark}

\theoremstyle{definition}
\newtheorem{remark}{Remark}
\newtheorem{defn}{Definition}

\theoremstyle{mystyle}

\makeatother

\renewcommand{\exp}{\mathsf{exp}}
\renewcommand{\log}{{\llog}}
\renewcommand{\Tr}{\mathsf{Tr}}
\renewcommand{\Pr}{\mathsf{Pr}}
\newclass{\PQP}{PQP}
\newclass{\UQMA}{UQMA}
\newclass{\UMA}{UMA}
\newclass{\UQCMA}{UQCMA}
\newclass{\UNP}{UNP}
\newclass{\GQMA}{GQMA}
\newclass{\GQCMA}{GQCMA}
\newclass{\PGQMA}{PGQMA}
\newclass{\EGQMA}{EGQMA}
\newclass{\PGQCMA}{PGQCMA}
\newclass{\EGQCMA}{EGQCMA}
\newclass{\pQMA}{PreciseQMA}
\newclass{\pBQP}{PreciseBQP}
\newclass{\pQCMA}{PreciseQCMA}
\newclass{\postBQP}{postBQP}
\newclass{\postQMA}{postQMA}
\newclass{\postQCMA}{postQCMA}
\newclass{\pUQMA}{PreciseUQMA}
\newclass{\pGQMA}{PreciseGQMA}
\newclass{\pPGQMA}{PrecisePGQMA}
\newclass{\pPGQCMA}{PrecisePGQCMA}
\newclass{\pEGQMA}{PreciseEGQMA}
\newclass{\pEGQCMA}{PreciseEGQCMA}
\newclass{\pAPGQMA}{PreciseAPGQMA}
\newclass{\pAEGQMA}{PreciseAEGQMA}
\newclass{\SampBPP}{SampBPP}
\newclass{\X}{X}
\newclass{\Y}{Y}
\newclass{\Z}{Z}
\newclass{\sP}{\#P}

\usepackage[normalem]{ulem}

\DeclareMathOperator*{\argmin}{arg\,min}

\begin{document}

\newgeometry{margin=0.5in}

\title{Effect of non--unital noise on random circuit sampling}

\author[1]{\normalsize{Bill~Fefferman}\thanks{wjf@uchicago.edu}}%\thanks{abhinavd@caltech.edu} } 
\author[1]{Soumik~Ghosh\thanks{soumikghosh@uchicago.edu}} 
\author[2]{Michael~Gullans\thanks{mgullans@umd.edu}}
\author[3,4]{Kohdai~Kuroiwa\thanks{kkuroiwa@uwaterloo.ca}}
\author[5]{\normalsize{Kunal~Sharma}\thanks{kunals@ibm.com}} 

\affil[1]{Department of Computer Science, University of Chicago, Chicago, Illinois 60637, USA}
\affil[2]{\normalsize{Joint Center for Quantum Information and Computer Science and Joint Quantum Institute, University of Maryland \& NIST, College Park, Maryland 20742, USA}}
\affil[3]{\normalsize{Institute for Quantum Computing, University of Waterloo, Ontario N2L 3G1, Canada}}
\affil[4]{\normalsize{Perimeter Institute for Theoretical Physics, Ontario, Canada, N2L 2Y5}}
\affil[5]{\normalsize{IBM Thomas J. Watson Research Center, Yorktown Heights, New York 10598, USA}}
\date{\vspace{-2em}}

\maketitle

\begin{abstract}
In this work, drawing inspiration from the type of noise present in real hardware, we study the output distribution of random quantum circuits under practical non--unital noise sources with constant noise rates. We show that even in the presence of unital sources like the depolarizing channel, the distribution, under the combined noise channel, never resembles a maximally entropic distribution at any depth. To show this, we prove that the output distribution of such circuits never anticoncentrates --- meaning it is never too "flat" --- regardless of the depth of the circuit. This is in stark contrast to the behavior of noiseless random quantum circuits or those with only unital noise, both of which anticoncentrate at sufficiently large depths. As consequences, our results have interesting algorithmic implications on both the hardness and easiness of noisy random circuit sampling, since anticoncentration is a critical property exploited by both state-of-the-art classical hardness and easiness results.
\end{abstract}

\restoregeometry

\newpage
\tableofcontents
\newpage
\section{Introduction}
The defining feature of quantum systems today is noise \cite{Preskill_2018}. A fundamental question in this era of noisy intermediate-scale quantum computers (NISQ) is whether noise renders any demonstration of quantum advantage with these systems useless, or whether some advantage is still salvageable for specific tasks \cite{franca, Bouland_2022}. To study questions related to quantum advantage, a popular paradigm is the random quantum circuit model; for instance, see the works of \cite{Bouland_2018, movassagh_efficient_2018, Bouland_2022, Deshpande_2022}. This is because for large system sizes, sampling from the output distribution of these circuits is a task that is easy for quantum computers but provably
hard for classical computers, under plausible complexity theoretic assumptions \cite{aaronsonarkhipov, Bouland_2018}. Amongst other phenomena, these circuits can model quantum chaos \cite{Nahum_2018}, quantum pseudorandomness \cite{pseudorandomness}, and the ansatz for certain types of variational quantum algorithms used in optimization tasks \cite{kandala2017hardware, Boixo2018}. Understanding the behavior of these circuits under physically motivated noise models and limited system sizes is crucial to our understanding of quantum advantage. 

One central feature of the output distribution, found in random quantum circuits of sufficiently high depth, is anticoncentration: it is a "flatness property" or, more formally, it means that the distribution is not concentrated on a sufficiently small number of outcomes. Anticoncentration is believed to be a key ingredient in both easiness and hardness proofs of random circuit sampling; see, for example, \cite{Boixo2018, Harrow_2023, Deshpande_2022}. Importantly, anticoncentration is necessary for the final state of the system to have an output distribution that mimics the uniform distribution. If the system were to have that property, then the system would be simulable, because sampling from the uniform distribution is classically easy. This naturally prompts the following question:

 \begin{center}
\emph{Do random quantum circuits, under the influence of physically motivated noise models, anticoncentrate?}
 \end{center}

\subsection{Our contributions}
\label{contributions}
In this work, we answer this question in the negative: we show how random quantum circuits, under noise models inspired by real hardware, which is a mixture of both unital and non-unital noise of certain types, exhibit a lack of anticoncentration, regardless of the depth of the circuit. This shows that the output distribution does not resemble the uniform distribution, or close variants of the same. Interestingly, we find that this behaviour is exhibited \emph{even when} there are also unital noise sources, like the depolarizing noise, present in the system, whose property is to push the system towards the maximally mixed state --- which, when measured in the standard basis, gives the uniform distribution. In this sense, these non-unital sources "dominate" the other sources. 

We leave open whether there is a quantum signal in these distributions and sampling from them is indeed classically hard, or whether the system tends towards a more sophisticated classically simulable final state. Much is unknown about random circuit sampling under such realistic noise models, and our work provides one of the first rigorous theoretical analyses of this regime. 

\subsubsection{Our conventions}
Our circuit model is that of geometrically local random quantum circuits, with the output measured in the standard basis, as illustrated in \cref{fig:setup2}. Each single-qubit noise channel acts independently after each gate. 

We prove lack of anticoncentration with respect to either of the two popular definitions: a strong definition of anticoncentration with respect to the convergence of scaled collision probabilities, and a weak definition of anticoncentration with respect to high probability mass of typical probabilities. The definitions and connections between them are made explicit in \cref{section: anticoncentration_def}.
\begin{remark}
Note that the terms "strong" and "weak" definition are relics from studying these definitions with respect to noiseless random circuits or random circuits with the depolarizing noise, where the strong definition implies the weak definition. However, this is not the case for the non--unital noise channels we consider, as we elaborate in \cref{connection}. Hence, our proofs of lack of anticoncentration, with respect to these two definitions, are independent. Nonetheless, we stick with the existing nomenclature to refer to these definitions succinctly in different parts of the paper. 
\end{remark}
\subsubsection{Our noise model}
We first show our results for the case when the noise source is modelled as having a mixture of depolarizing noise and amplitude damping noise: these two are emblematic of  unital and non--unital noise, respectively. The model is discussed in detail in \cref{section: noise_modeling}. 

After our initial results, we discuss how the same techniques can be used in extending our results to a generic noise channel, in \cref{section: general noise1}.

\subsubsection{Our results}
\label{our results}
Broadly, we prove three categories of results.

\begin{itemize}
\item First, in \cref{lack_anticoncentration}, we show how the scaled collision probabilities for our noisy ensembles, where noise is modelled as a mixture of amplitude damping and depolarizing noise, diverge: this means anticoncentration fails with respect to the strong definition. This has a clean proof, which involves "removing" the last layer of noise by using the adjoint of the noise channel, and then using properties of the local Haar measure, like translational invariance and explicit formulae for second moments \cite{Collins_2006}, to prove a lower bound. 

\item Then, we show how typical output probabilities for strings which have high Hamming weights are small. This is done in \cref{sec1_typical}, \cref{section: low depth}, and \cref{sec:high depth}. This shows anticoncentration fails with respect to the weak definition. This has a more complicated proof involving lightcone arguments and the stat--mech model. 

\item Finally, we discuss how extensions of our proofs hold for a wide variety of generic noise models: this is done in \cref{section: general noise1} and \cref{measurement}. In particular, we prove how lack of anticoncentration, with respect to the strong definition, is exhibited whenever the noise channel, acting on the identity operator, puts non--zero constant weight on the Pauli-$\mathsf{Z}$ operator.
\end{itemize}

\subsection{Our techniques}
There are three main classes of techniques that we utilize in our proofs. For many of our calculations, like those involving putting bounds on the expected collision probability or those involving computing the first moment of output probabilities, we first "remove" the last layer of noise, compute relevant quantities for the modified circuit, and then generalize our calculations to the actual circuit.

This is usually done by considering the adjoint map of the last noise layer. For our calculations about typical probabilities, we use lightcone arguments and the stat mech model.

\subsubsection{Removing the last layer of noise}

The technique of removing the last layer of noise by considering the adjoint of the noise channel makes calculations convenient because if our circuit terminates with a last layer of single qubit Haar random gates, instead of a last layer of noise, then we can then use many properties of the Haar measure directly, like translational invariance or explicit expressions of higher moments. Additionally, this makes many of our results extremely general, especially the ones involving lack of anticoncentration by proving divergence of scaled collision probabilities, because this technique has \emph{no} dependence on the underlying architecture, unlike similar divergence results, for noiseless and unital noise models, in \cite{dalzell}, which involve the stat-mech model, and hence, is only known to be applicable for certain specific architectures. Moreover, because of this proof technique, our proof of divergence holds for \emph{any} circuit depth, as long as there is a last layer of noise. 

To be more formal, suppose that we have an ensemble $\mathcal{B}$ of noisy random quantum circuits with noise channel $\mathcal{N}$, and pick a quantum circuit $\mathcal{C} \in \mathcal{B}$. In our analysis, we usually ``remove'' the last layer of noise, and deal with the adjoint of the noise. 
More specifically, let $\mathcal{C}'$ be the quantum circuit obtained by removing the last layer of noise; that is, 
\begin{equation}
    \mathcal{C} = \mathcal{N}^{\otimes n} \circ \mathcal{C}'. 
\end{equation}
Let $\mathcal{B}'$ be the set of quantum channels obtained by removing the last layer of noise from the circuits in $\mathcal{B}$.
The expected probability of getting the result $x \in \{0,1\}^n$ is 
\begin{equation}
    \underset{\mathcal{B}}{\mathbb{E}} \left[\Tr\left(\ketbra{x}{x}~ \mathcal{C}(\ketbra{0}{0})\right) \right]. 
\end{equation}
By the definition of the adjoint map, we have 
\begin{equation}
    \underset{\mathcal{B}}{\mathbb{E}} \left[\Tr\left(\ketbra{x}{x}~ \mathcal{C}(\ketbra{0}{0})\right) \right] 
    = 
    \underset{\mathcal{B}'}{\mathbb{E}} \left[\Tr\left((\mathcal{N}^{\dagger})^{\otimes n}(\ketbra{x}{x})~ \mathcal{C}'(\ketbra{0}{0})\right) \right]. 
\end{equation}
By the linearity of $\Tr$ and $\mathbb{E}$, 
\begin{equation}
    \underset{\mathcal{B}'}{\mathbb{E}} \left[\Tr\left((\mathcal{N}^{\dagger})^{\otimes n}(\ketbra{x}{x})~ \mathcal{C}'(\ketbra{0}{0})\right) \right]
    = \Tr\left((\mathcal{N}^{\dagger})^{\otimes n}(\ketbra{x}{x})~ \underset{\mathcal{B}'}{\mathbb{E}} \left[\mathcal{C}'(\ketbra{0}{0})\right]\right). 
\end{equation}
Hence, to analyze this expected probability, we may evaluate 
\begin{equation}
(\mathcal{N}^{\dagger})^{\otimes n}(\ketbra{x}{x})~~\text{and}~~\underset{\mathcal{B}'}{\mathbb{E}} \left[\mathcal{C}'(\ketbra{0}{0})\right]
\end{equation}
separately. 
$(\mathcal{N}^{\dagger})^{\otimes n}(\ketbra{x}{x})$ can usually be evaluated straightforwardly when the description of noise $\mathcal{N}$ is given. Note that $\mathcal{C}'$ terminates with a last layer of two--qubit Haar random gates.

\begin{remark}
Just like in the computation of first moment quantities, the trick of "removing" the last layer of noise by taking its adjoint is useful even in bounding certain second moment quantities, like the collision probability, as we detail in \cref{lack_anticoncentration}.
\end{remark}
\newpage
\subsubsection{Lightcone arguments}

The second class of techniques that we utilize to study low depth circuits are variants of lightcone type arguments. These techniques are popular in the study of low depth random circuits in different settings; see for example \cite{napp, Deshpande_2022}. 

The qualitative intuition behind lightcone arguments is that by looking at the size of the lightcone for each  qubit marginal at low depth, one could argue that the circuit doesn't "scramble" the distribution too well for sufficiently many strings to have high probability mass. Then, by studying particular noise channels, one could argue that specific noises do not assist in the "scrambling" either. Additionally, because of small lightcone sizes, instead of directly looking at the random variable $p_x$ --- the output probability of a string $x \in \{0, 1\}^n$ --- it suffices to look at the random variable
\begin{equation}
- \frac{1}{n} \log p_x,
\end{equation}
and prove its concentration around the mean by Markov's inequality; stronger second moment bounds are not needed. Let us emphasize that we succeed in developing techniques that are more general than variants that came before, which may be of independent interest. 
Previously, the analysis and application of these methods only extended to the noiseless or the unital case, as discussed in papers like \cite{dalzell, Deshpande_2022}. 

To use lightcone arguments to show how certain output strings, with high Hamming weight, have very low probability mass at low depth, the key ingredient is to show that the lower bound of the total variation distance between the noisy distribution and the noiseless distribution is exponentially decaying in the depth of the circuit. In \cref{lowdepth}, we prove how this holds true whenever the noise channel $\mathcal{N}$ satisfies 
\begin{equation}
\left\langle\ketbra{0}{0},\mathcal{N}^d(\ketbra{0}{0})\right\rangle = \kappa + \tau\lambda^d,
\end{equation}
where $d$ is the depth of the circuit, with some $\tfrac{1}{2} \leq \kappa \leq 1$, $0 \leq \tau,\lambda\leq 1$. More specifically, this statement is true when the noise channel under consideration is a mixture of amplitude damping and depolarizing noise. 
\begin{remark}
This analysis extends that in \cite{Deshpande_2022}, where a similar technique was used for the noiseless case and the case with only Pauli noise.
\end{remark}

\subsubsection{Mapping to classical partition functions}
The third type of techniques we utilize, to study sufficiently deep circuits, are mappings to classical partition functions. These techniques were developed to study the output distribution of random quantum circuits in different settings; see for example \cite{dalzell, dalzell2021random}. 

Since lightcone sizes blow up for superlogarithmic depths, lightcone arguments are no longer tight, and we need stronger second--moment inequalities to study the output distribution. To use these inequalities, we explicitly upper bound the second moment of the distribution. This is done using mappings to classical partition functions.

To show that the same strings have low probability mass at sufficiently high depths by applying second--moment inequalities, we need to upper bound the second moment of their output probabilities using the stat--mech model. 
\newpage
To do this, one standard way is to iteratively replace each two--qubit random gate in the circuit with two copies of a single qubit random gate, show, using the stat--mech model, that the collision probability of this modified circuit upper bounds the collision probability of the actual circuit, and then directly upper bound the collision probability of the modified circuit. The steps are explained in detail in \cref{sec:high depth}, where we also show how all the steps go through whenever the noise channel satisfies 
\begin{align}
    \tilde{M}_{U_1, \mathsf{N}}(I_4) &= (1 - a)I_4 + 2a S, \\ 
    \tilde{M}_{U_1, \mathsf{N}}(S) &= bI_4 + (1 - 2b) S,
\end{align}
with $a > 0$, and $b > 0$, where $I_4$ is the 2-qubit identity operator and $S$ is the $2$-qubit SWAP gate,
\begin{equation}
M_{U_1}[\rho] = \underset{U_1 \sim \mathcal{U}_{\textup{Haar}}}{\mathbb{E}}\bigg[U_1^{\otimes 2} \rho U_1^{\dagger \otimes 2} \bigg],
\end{equation}
$\mathsf{N} = \mathcal{N} \otimes \mathcal{N}$, and the operator 
\begin{equation}
    \tilde{M}_{U_1, \mathsf{N}} = M_{U_1} \circ \mathsf{N} \circ M_{U_1}.
\end{equation}
\begin{remark}
Our analysis is inspired by the techniques in \cite{dalzell, Deshpande_2022}, whose analysis only covered the noiseless and the unital cases. 
Their analysis does not work for general noise models, as it is non--trivial to prove that when we iteratively replace each two-qubit gate with two copies of a single qubit gate in a noisy circuit, the stat mech representation is still a valid one --- that is, that there are no negative path weights for $I$ and $S$ --- and the original collision probability is upper bounded by that of the modified circuits. We prove that this is indeed the case in \cref{lem:modified-circuit}.
\end{remark}

\subsection{Putting our results in context}
Depending on how strong the noise is, we can divide random circuit sampling ($\mathsf{RCS}$) into two different complexity theoretic regimes. 

\subsubsection{High noise regime}
The first regime is that of \emph{high noise}, when the noise rate is a constant that is independent of the system size, even when the number of qubits grows asymptotically. Near term devices are susceptible to constant noise rates \cite{bravyi2022future}. It is an equally reasonable model for scaled--up fault--tolerant systems, because to achieve fault--tolerance, suppressing the noise below a certain constant threshold suffices --- one does not need noise to go down with system size \cite{aharonov2, Knill1998, Kitaev2003}. 

\subsubsection{Modeling the high noise regime}
If we model the noise as only depolarizing noise in the high noise regime, then after sufficient depth, random circuit sampling becomes an easy task classically. It was known that a trivial classical algorithm which just samples from the uniform distribution achieves a total variation distance error of $2^{-\Theta(d)}$ from the target noisy $\mathsf{RCS}$ distribution \cite{aharonov1996polynomial}: the upper bound is due to \cite{aharonov1996polynomial} and, more recently, the lower bound is due to \cite{Deshpande_2022}. As is evident, for depth strictly greater than logarithmic, this trivial sampler achieves a total variation distance error that is smaller than any inverse polynomial. 
\newpage
At logarithmic depth, in a recent work, \cite{aharonov2022polynomialtime} proposed another classical sampler from the output distribution of random circuits with depolarizing noise, which instead of achieving a total variation distance error that is inversely proportional to a fixed polynomial, provides a way to fine tune the total variation distance error to \emph{any} polynomial function of our choice. In \cref{Easiness_sampling}, it is elaborated how, because of the special property of the depolarizing noise, the guarantee that the noiseless ensemble of \cite{aharonov2022polynomialtime} satisfies anticoncentration is necessary for their current analysis to work. This is why their sampler works well only for logarithmic depth and beyond because anticoncentration needs at least logarithmic depth to kick in, and before that, anticoncentration fails \cite{dalzell, Deshpande_2022}. 

Although depolarizing noise, along with anticoncentration, implies a classical sampler from random circuits, such noise is not the only type of noise source present in real hardware. There are non--unital effects in all known experimental hardware, for example those in \cite{Debnath2016, Arute2019, Pino2021, Zhu2022, Carroll2022}. These sources are fundamentally different from unital sources, like the depolarizing channel, in the following sense: the depolarizing channel --- \emph{increases} the entropy of the system by pushing it towards the maximally mixed state; however, non--unital noise channels can \emph{decrease} the entropy of the system and actually push it towards a pure state. For the low noise regime, there is some evidence that depolarizing noise remains a good approximation to all the noise sources present in the system \cite{dalzell}, if the system only has unital noise. However, apart from some very special cases which we discuss in \cref{sec: twirling}, this approximation is not valid in the high noise regime, especially when the system also has non--unital noise sources.

\subsubsection{Low noise regime}
The second regime is that of \emph{low noise}. Here, the strength of the noise is inversely proportional to the number of qubits. This can be thought to be a good approximation of the noise present in relatively small, fixed--sized systems, like those used in sampling hardness demonstrations, for example \cite{Arute2019, morvan2023phase}, where the number of qubits is fixed, and the noise is a fixed constant that is inversely related to the number of qubits. However, without further investigation, it is unclear if asymptotic analysis of the low noise regime is relevant to studying the properties of finite system sizes: we cannot ensure by current technology that the noise rate continues to go down with the number of qubits, when the latter is increased. 

The low noise regime provides advantages to tasks like benchmarking, where fidelity of the output state is a figure of merit: below a certain noise threshold, the linear cross--entropy test ($\mathsf{XEB}$) \cite{aaronson2016complexitytheoretic, Arute2019, aaronson2020classical} corresponds to the fidelity of the output state of the noisy circuit and gives us a sample efficient way of estimating the fidelity of that state using only standard basis measurement. Originally, linear cross--entropy was proposed as a heuristic proxy for fidelity by \cite{Boixo2018, Bouland_2018, liu2022benchmarking}, which was recently confirmed by \cite{morvan2023phase, ware2023sharp}. There is also some evidence from complexity theory that classically sampling from these circuits is hard when the depth is sufficiently large~\cite{dalzell}. 

\subsubsection{Implications of our results}

Our results show that the phenomenon of anticoncentration is a function of the noise present in the circuit: it \emph{does not} hold for reasonable, physically--motivated, non--unital noise models in the high noise regime. Many of our techniques can be generalized to work for a wide variety of noise models and setups, including when we do not have the last layer of noise and only have noise in the middle layers, which we discuss in Section~\ref{measurement}. 
\clearpage
The failure to anticoncentrate implies that the final state of the system does not resemble a maximally mixed state and the uniform distribution is not a good proxy for the output distribution of the system, as elaborated on in \cref{closeness to uniform}. So, sampling from the uniform distribution, or close variants of the same, no longer works as an effective strategy to classically spoof the output distribution. Additionally, the failure to anticoncentrate also implies that no known techniques can be harnessed to show that more sophisticated samplers, like the one in \cite{aharonov2022polynomialtime}, succeed in spoofing the output distribution. 

As mentioned in \cref{contributions}, our work leaves open whether sampling from this distribution, in the asymptotic limit, is classically hard under plausible complexity theoretic hardness conjectures, or whether the final state converges to a classically simulable fixed point that is much more sophisticated than just a maximally mixed state. It could also be interesting to investigate whether the lack of anticoncentration implies any advantage in the computation of the expectation value of certain cost functions in variational setups. We discuss many more open problems in \cref{open problems}.

\section{Notation and setting}
In this section, we set up our problem, introduce the notation used in this paper, and define our architecture. Throughout this paper, we use $I_{N}$ to denote the $N\times N$ identity matrix. For example, $I_2$ represents the single-qubit identity operator. 
The single-qubit Pauli operators are defined as follows: 
\begin{equation}
    \sigma_x = \left(
    \begin{array}{cc}
        0 & 1 \\
        1 & 0
    \end{array}
    \right)\,\,\,\, 
     \sigma_y = \left(
    \begin{array}{cc}
        0 & -i \\
        i & 0
    \end{array}
    \right)\,\,\,\, 
     \sigma_z = \left(
    \begin{array}{cc}
        1 & 0 \\
        0 & -1
    \end{array}
    \right). 
\end{equation}
We use $\mathcal{I}$ to denote the single-qubit identity channel. 
For a string $p \in \{0,1,2,3\}^n$, $\omega_p$ is defined as the Hamming weight of $p$, that is, the number of nonzero symbols in $p$.  

We define the circuit architecture we use in the rest of the paper. 
We assume familiarity with basic terminologies in quantum computing, such as, qubits, quantum circuits, and circuit depth. 
Unless otherwise stated, a quantum circuit, usually denoted by $\mathcal{C}$, is taken to be a CPTP map, and the final measurements, after applying a quantum circuit, are always performed in the standard basis.

\begin{defn}[Parallel quantum circuit]
An $n$-qubit \emph{parallel} quantum circuit, where $n$ is an even number, 
is a quantum circuit where every qubit is involved in a one or two qubit gate, at every depth instance.
\end{defn}

\begin{defn}[Geometrically local quantum circuit]
An $n$-qubit \emph{geometrically local} quantum circuit is a quantum circuit where every quantum gate acts on nearest-neighbour qubits.
\end{defn}

\begin{defn}[Noisy quantum circuit]
An $n$-qubit \emph{noisy} quantum circuit is one where every quantum gate is followed by a single qubit noise channel $\mathcal{N}$ on every qubit involved in the gate.
\end{defn}

\begin{defn}[Random quantum circuit]
An $n$-qubit \emph{random} quantum circuit is one where every quantum gate, acting on $k$ qubits, for a constant $k$, is drawn from the Haar measure on $\mathsf{U}(2^k)$: this is the set of all unitary matrices, of dimension $2^k \times 2^k$. 
\end{defn}
Alternately, a random quantum circuit can be interpreted as a quantum circuit picked uniformly at random from an \emph{ensemble} of quantum circuits. In this paper, unless otherwise stated, we consider parallel, geometrically local, noisy, and random circuits, as depicted in \cref{fig:setup2}.
We usually use $\mathcal{B}$ to denote an ensemble of noisy random circuits. 

\subsection{Modeling the noise}
\label{section: noise_modeling}

Quantum devices are affected by two sources of noise: unital and non-unital quantum noise channels. A quantum channel $\mathcal{N}$ is a unital channel if $\mathcal{N}(I) = I$, where $I$ denotes the identity operator. Dephasing, bit-flip, and depolarizing noise channels are examples of unital quantum channels. On the other hand, an amplitude damping channel is an example of a non-unital quantum channel, which models the $T_1$ noise in superconducting quantum devices.  Unital and non-unital noise sources have the opposite behavior. While amplitude damping noise "bias" the system towards a particular fixed state ($\ket{0}$ state) and  unital sources push the system towards the maximally mixed state.

In the next few sections, we will use a combination of the depolarizing channel and the amplitude damping channel to model the noise after each single qubit gate. Amplitude damping noise is emblematic of the $\mathsf{T_1}$ noise \cite{Carroll2022,kubica2022erasure}, and the depolarizing channel is emblematic of the type of unital noise just described \cite{Arute2019, aharonov2022polynomialtime}. In later sections, we discuss how our analysis and techniques are general enough to apply to a wide variety of noise channels.

\subsubsection{Amplitude damping noise}
This type of noise pushes a qubit towards the state $|0\rangle \langle 0|$. It is represented by two Kraus operators, as given by \cref{defn: amp-damping}. The first Kraus operator "dampens" the $|1\rangle \langle 1|$ term, and the second Kraus operator takes the state $\ket{1}\bra{1}$ to $\ket{0}\bra{0}$ state, with a prefactor. Both the operators serve to make the contribution of $\ket{0}\bra{0}$ dominate in the final state that we get after the channel is applied.

\begin{defn}[Amplitude Damping Noise Channel]
\label{defn: amp-damping}

    Let $0 \leq q \leq 1$ be a real parameter. 
    A single-qubit amplitude damping noise $\mathcal{N}^{(\textup{amp})}_q$ with noise strength $q$ is a quantum channel with the following Kraus operators: 
    \begin{equation}
    K_0 = \left(
    \begin{array}{cc}
        1 & 0 \\
        0 & \sqrt{1-q}
    \end{array}
    \right),\,\,
        K_1 = \left(
    \begin{array}{cc}
        0 & \sqrt{q} \\
        0 & 0
    \end{array}
    \right).
\end{equation}
\end{defn}
\noindent Therefore, for a single-qubit linear operator 
\begin{equation}
    X = \left(\begin{array}{cc}
        x_{00} & x_{01} \\
        x_{10} & x_{11}
    \end{array}\right), 
\end{equation}
the action of the amplitude damping channel $\mathcal{N}^{(\textup{amp})}_q(X)$ is given by 
\begin{equation}
\mathcal{N}^{(\textup{amp})}_q(X) = \left(\begin{array}{cc}
    x_{00} + qx_{11} & \sqrt{1-q}x_{01} \\
    \sqrt{1-q}x_{10} & (1-q)x_{11}
\end{array}\right). 
\end{equation}

\begin{figure}
 \begin{centering}
    \includegraphics{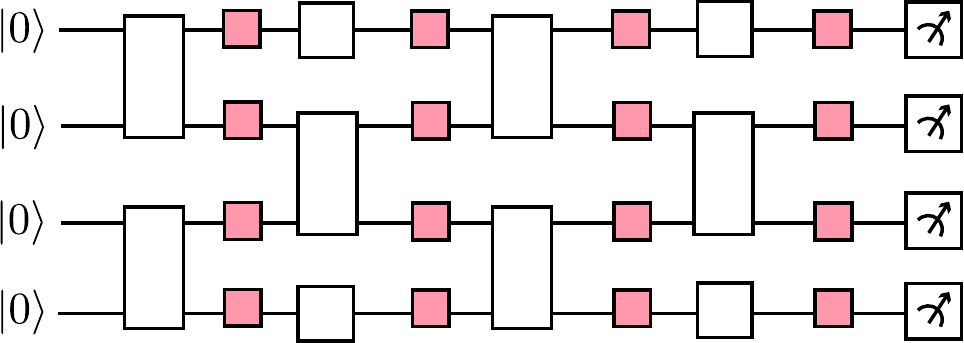}
    \caption{The circuit model that we use for our analysis. Every pink rectangle is a single qubit noise channel. Every white rectangle is either a single qubit or a two--qubit Haar random gate. In the end, the final state is measured in the standard basis.
    }
    \label{fig:setup2}
     \end{centering}
\end{figure} 

\subsubsection{Depolarizing noise}

\begin{defn}[Depolarizing Noise Channel]
    Let $0 \leq p \leq 1$ be a real parameter. 
    A single-qubit depolarizing noise $\mathcal{N}^{(\textup{dep})}_p$ with noise strength $p$ is a quantum channel with the following Kraus operators:
    \begin{equation}
        \begin{aligned}
    &K_0 = \sqrt{1 - \frac{3p}{4}}\left(
    \begin{array}{cc}
        1 & 0 \\
        0 & 1
    \end{array}
    \right),\,\,
        K_1 = \sqrt{\frac{p}{4}}\left(
    \begin{array}{cc}
        0 & 1 \\
        1 & 0
    \end{array}
    \right),\,\, \\
    &K_2 = \sqrt{\frac{p}{4}}\left(
    \begin{array}{cc}
        0 & -i \\
        i & 0
    \end{array}
    \right), \,\,
        K_3 = \sqrt{\frac{p}{4}}\left(
    \begin{array}{cc}
        1 & 0 \\
        0 & -1
    \end{array}
    \right).
        \end{aligned}
\end{equation}
\end{defn} 
\noindent Therefore, for any single-qubit linear operator $X$
\begin{equation}
    X = \left(\begin{array}{cc}
        x_{00} & x_{01} \\
        x_{10} & x_{11}
    \end{array}\right), 
\end{equation}
the action of the depolarizing channel $\mathcal{N}^{(\textup{dep})}_p(X)$ is given by 
\begin{align}
\mathcal{N}^{(\textup{dep})}_p(X) &= \left(\begin{array}{cc}
    \left(1 - p\right)x_{00} + \frac{p}{2} x_{11} & (1-p)x_{01} \\
    (1-p)x_{10} &  \left(1 - p\right)x_{11} + \frac{p}{2} x_{00}
\end{array}\right) \\ 
&= (1-p)X + \frac{p}{2}\Tr(X)I_2,  
\end{align}
with the single-qubit identity operator $I_2$.

\begin{remark}
Usually, when the type of noise is clear from the definition, we drop the superscript and refer to the noise channel as $\mathcal{N}_c$, for some choice of the strength $c$.
Furthermore, we will often consider properties of noisy random quantum circuits \emph{in expectation}: unless otherwise stated, the expectation will always be taken only over the choice of random gates, and \emph{not} over the noise channels.
\end{remark}

\begin{remark}
The notion that the amplitude damping channel is fundamentally different from the depolarizing channel is elaborated on in works like \cite{benor2013quantum}, where it is shown how quantum circuits with uncorrected amplitude damping channel can be used to do exponential time quantum computation, in the worst case, by using the noise as a resource to generate fresh ancilla qubits. However, this cannot be done for quantum circuits with the depolarizing noise channel.
\end{remark}

\section{Anticoncentration and the lack thereof}
\label{section: anticoncentration_def}
In this section, we introduce the notion of anticoncentration and discuss what it means to \emph{not} anticoncentrate. First, we will use the definition of anticoncentration in \cite{dalzell}, where it is defined with respect to the collision probability. Then, we will talk about another definition of anticoncentration, which, for example, is found in references like \cite{aaronsonarkhipov, Bremner_2016, Deshpande_2022}, and discuss connections between the two definitions.

\subsection{Strong definition: with respect to the scaled collision probability}

\begin{defn}
The \emph{output probability} $p_x$, of a string $x \in \{0, 1\}^n$, for a quantum circuit $\mathcal{C}$ is given by
\begin{equation}
p_x = \Tr\big(|x\rangle \langle x|~ \mathcal{C}(|0^n\rangle \langle 0^n |)\big).
\end{equation}
\end{defn}

\begin{defn}
For an ensemble $\mathcal{B}$, the scaled collision probability is defined as
\begin{equation}
\label{eq_collision}
\mathcal{Z} = 2^n \cdot \underset{\mathcal{B}}{\mathbb{E}} \left[ \sum_{x \in \{0, 1\}^n} p_x^2 \right]-1.
\end{equation}
\end{defn}
\noindent In other words,
\begin{equation}
\mathcal{Z} = 2^n \cdot \underset{\mathcal{B}}{\mathbb{E}} \left[ \sum_{x \in \{0, 1\}^n}  \Tr(|x\rangle \langle x| ~\mathcal{C}(\ketbra{0}{0}))^2 \right]-1. 
\end{equation}
\begin{defn}
\label{Definiton: anticoncentration}
An ensemble $\mathcal{B}$ of $n$-qubit  quantum circuits is defined to be \emph{anticoncentrated} if 
\begin{equation}
\label{equation: anticoncentration}
\mathcal{Z} = \mathcal{O}(1).
\end{equation}
\end{defn}
Intuitively, \cref{Definiton: anticoncentration} says that the probability of seeing a collision, after sampling twice from the output distribution of a circuit in $\mathcal{B}$, is extremely small in expectation. In other words, "most" $n$--bit strings have sufficiently high probability weight in the output distribution. \cref{Definiton: anticoncentration} is a "stronger" definition of anticoncentration because, for certain types of circuits, it implies another well--studied "weaker" definition of anticoncentration, as we will see in \cref{secondsec}.

\begin{defn}
\label{Definiton: concentration}
An ensemble $\mathcal{B}$ of $n$-qubit  quantum circuits is said to exhibit \emph{lack of anticoncentration} if 
\begin{equation}
\label{equation: concentration}
\mathcal{Z} = \omega(1).
\end{equation}
\end{defn}

The $\mathsf{RHS}$ means that the quantity is strictly more than a constant: it grows as some function of $n$. Intuitively, a distribution which satisfies \cref{Definiton: anticoncentration} is not "evenly" supported on all $n$--bit strings: some strings have much more probability weight than others. So, there is a much larger probability of seeing a collision when sampling multiple times from this distribution.
\begin{remark}
The connection of \cref{Definiton: anticoncentration} to easiness proofs is elaborated on in \cref{Easiness_sampling}.
In short, the guarantee that ensembles satisfy \cref{Definiton: anticoncentration} is necessary for the current techniques to analyze the accuracy of the classical sampler, proposed by \cite{aharonov2022polynomialtime}, to work. This is why the samplers only works after logarithmic depth; logarithmic depth is the threshold at which these ensembles are known to start anticoncentrating \cite{dalzell}. 
\end{remark}
\subsection{Weak definition: with respect to individual probabilities}
\label{secondsec}
There is another definition of anticoncentration, which defines it in terms of how large individual probabilities are. This can be found in works like \cite{fefferman, Bouland_2018, aaronsonarkhipov, Deshpande_2022}.

\begin{defn}
\label{newdef}
An ensemble $\mathcal{B}$ of $n$-qubit  quantum circuits is defined to be \emph{anticoncentrated} if for every $x \in \{0, 1\}^n$, there exists a choice of $\alpha, \beta \in (0, 1]$ such that
\begin{equation}
\label{equation: anticoncentration2}
\underset{\mathcal{B}}{\Pr}\bigg[p_x \geq \frac{\alpha}{2^n}\bigg] \geq \beta. 
\end{equation}
\end{defn}

\noindent We define the lack of anticoncentration in a similar way.

\begin{defn}
\label{newdef: concentration}
An ensemble $\mathcal{B}$ of $n$-qubit  quantum circuits is defined to \emph{lack anticoncentration} if there exists an $x \in \{0, 1\}^n$, such that for any $\alpha \in (0, 1]$,
\begin{equation}
\label{equation: lack anticoncentration2}
\underset{n \rightarrow \infty}{\lim} \underset{\mathcal{B}}{\Pr}\bigg[p_x < \frac{\alpha}{2^n}\bigg] = 1. 
\end{equation}
\end{defn}

The form of anticoncentration in \cref{newdef} is believed to be useful in proving that the output distribution of some random circuit ensembles are hard to classically sample from \cite{aaronsonarkhipov, Bouland_2018}.  \cref{newdef} is called the weaker definition of anticoncentration because \cref{Definiton: anticoncentration} implies \cref{newdef}, as we will now establish.

\subsection{Connection between two definitions}
\label{connection}
Before establishing the connection between \cref{Definiton: anticoncentration} and \cref{newdef}, let us state a version of the "hiding" property.

\begin{defn}
Let $\mathcal{B}$ be an ensemble of random quantum circuits. Then, $\mathcal{B}$ is said to satisfy hiding if
\begin{equation}
\underset{\mathcal{B}}{\mathbb{E}}\left[p_x^{k}\right] = \underset{\mathcal{B}}{\mathbb{E}}\left[p_y^{k}\right],
\end{equation}
for $x, y \in \{0, 1\}^n$, $x \neq y$, and for any $k$.
\end{defn}

\noindent Note that for ensembles that satisfy hiding,
\begin{equation}
\label{col_prob}
\underset{\mathcal{B}}{\mathbb{E}} \left[ \sum_{x \in \{0, 1\}^n} p_x^2 \right] = 2^n \cdot \underset{\mathcal{B}}{\mathbb{E}} \left[ p_y^2 \right],
\end{equation}
for any $y \in \{0, 1\}^n.$ Note that the hiding property follows form the left and right invariance of the Haar measure under unitary transformations. 
Furthermore, from \cref{col_prob}, for any ensemble $\mathcal{B}$ satisfying both hiding and anticoncentration with respect to \cref{Definiton: anticoncentration},
\begin{equation}
\label{cp2}
\underset{\mathcal{B}}{\mathbb{E}} \left[ p_x \right] = \frac{1}{2^n};~~~~~~~~\underset{\mathcal{B}}{\mathbb{E}} \left[ p_x^2 \right] = \frac{\mathcal{O}(1)}{4^n},
\end{equation}
for any $x \in \{0, 1\}^n$. A number of ensembles satisfy hiding, including noiseless random quantum circuits and circuits with Pauli noise \cite{Deshpande_2022}. Let us state a proposition that follows from our discussion so far.

\begin{propo}
    Let $\mathcal{B}$ be an ensemble of random quantum circuits which satisfy hiding and anticoncentration with respect to \cref{Definiton: anticoncentration}. Then, it also satisfies anticoncentration with respect to \cref{newdef}.
\end{propo}
\begin{proof}
For any $x \in \{0, 1\}^n$, by the Paley-Zygmund inequality,
\begin{equation}
\Pr[p_x \geq \alpha \cdot \underset{\mathcal{B}}{\mathbb{E}} \left[ p_x \right]] \geq \frac{(1-\alpha)^2}{4^n \cdot \underset{\mathcal{B}}{\mathbb{E}} \left[ p_x^2 \right]}.
\end{equation}
Then, the proof follows from \cref{cp2}.
\end{proof}
\noindent In this sense, \cref{Definiton: anticoncentration} 
is stronger than \cref{newdef}. 

\begin{remark}
\label{first remark}
When an ensemble $\mathcal{B}$ does not satisfy hiding, the relation between \cref{Definiton: anticoncentration} and \cref{newdef} is not clear. Note that hiding is not satisfied in most of the setups that we study in this paper, because of how we model our noise. The details of our noise model are given in \cref{section: noise_modeling}.
\end{remark}

\subsection{Fine graining the weak definition}
\noindent One can consider a fine-grained version of \cref{newdef}. The fine-graining is important in setups where hiding is not satisfied and analyzing the typical probability weight for one particular bitstring does not tell us about the typical behavior of other bitstrings..

\begin{defn}
\label{newdef2}
An ensemble $\mathcal{B}$ of $n$-qubit  quantum circuits is defined to be $k$-\emph{anticoncentrated} if there exists a set $S = \{x : x \in \{0, 1\}^n \}$, with $|S| = k$, such that for every $x \in S$, there exists a choice of $\alpha, \beta \in (0, 1]$ satisfying
\begin{equation}
\label{equation: anticoncentration3}
\underset{\mathcal{B}}{\Pr}\bigg[p_x \geq \frac{\alpha}{2^n}\bigg] \geq \beta. 
\end{equation}
\end{defn}

\begin{remark}
Note that after a sufficiently large depth, noiseless random circuits, or random circuits with Pauli noise, are $2^n$-anticoncentrated \cite{dalzell}.
\end{remark}

\subsection{Our work}
Our work shows that, for the setups we discuss in the paper, anticoncentration fails both with respect to \cref{Definiton: anticoncentration} and \cref{newdef}. Our analysis is fine grained, in the spirit of \cref{newdef2}.

\section{Overview of proofs}
In this section, we summarize our proofs. The proofs are detailed in the subsequent sections. Just like \cref{our results}, this section gives a bird's eye view of the rest of the paper; but it is more formal than \cref{our results}.

\begin{itemize}
\item First, we will prove that the distribution exhibits a lack of anticoncentration, according to \cref{Definiton: concentration}, at any depth. Particularly, we will show that for a noisy ensemble $\mathcal{B}$, the scaled collision probability
\begin{equation}
\label{anticoncentration}
\mathcal{Z} \geq (1+s)^n - 1,
\end{equation}
for a non-negative constant $s$ depends on the strength of the noises present.

\item Then, we show how the distribution lacks anticoncentration according to \cref{newdef: concentration}, at any depth.

Moreover, our results show that the distribution is $\emph{never}$ $2^{n-1}$--anticoncentrated, according to \cref{newdef2}.

\begin{itemize}

\item For this, first we calculate the first moment of output probabilities to show that, in expectation, the probability weight on a string is exponentially suppressed with respect to the Hamming weight of the string: strings with lower Hamming weight have exponentially more weight than strings with higher Hamming weight.

Intuitively, this behavior comes from the fact that the fixed point of an $n$--fold tensor product of a single qubit amplitude damping channel is $|0^n\rangle \langle 0^n |$ and the fact that $0^n$ is a string with a Hamming weight of $0$. So, the distribution is biased towards strings that are closer in Hamming distance to $0^n$.

\item We then show that for any string with Hamming weight at least $\frac{n}{2}$, the probability weight on that string is negligible, for most circuits in the ensemble $\mathcal{B}$. 

This calculation is divided into two parts: the low depth and the high depth regime, with different techniques for each regime. We use a variant of a lightcone argument in the low depth regime, whereas the high depth regime is analyzed using mappings to a stat--mech model. 
Note that
mapping random circuits to stat--mech models, to study various quantities of interest, is a useful analysis tool and have been studied in \cite{hunterjones2019unitary, dalzell, Deshpande_2022}.
\end{itemize}

\item Then, we generalize our techniques to an arbitrary noise channel. First, we show that for an arbitrary noise channel with a non--unital component, the distribution exhibits a lack of anticoncentration according to \cref{Definiton: concentration}, for a wide range of parameters.

\item Thereafter, we derive a condition for which the noisy distribution shows a lack of anticoncentration according to \cref{newdef: concentration} and is never $2^{n-1}$--anticoncentrated. This results holds for a general noise model, when the noise is modelled as any CPTP map. 

\item Note that a layer of random gates can, intuitively, be thought to "scramble" the output distribution. On the other hand, a layer of amplitude damping noise can be thought to "unscramble" the distribution and push it back to a pure state. So, one might suspect that there is a "see--saw" effect and whether the distribution exhibits a lack of anticoncentration is dependent on whether we terminate our circuit with a layer of noise or a layer of noiseless gates. However, we do not think this is the case: we argue that amplitude damping noise in the middle layers is sufficient to cause lack of anticoncentration. 

To elaborate on this conceptual point that we want to make, at the end of our paper, we discuss a setup where we do not have the last layer of noise, and instead terminate with a last layer of noiseless gates. 

\noindent We argue how such setups also appear to lack anticoncentration, according to \cref{Definiton: concentration} and \cref{newdef: concentration}. Our results in this regime are not as general as they were before, but, nonetheless, they strongly suggest that the phenomenon of lack of anticoncentration holds true even when we have non--unital noise in only the middle layers.

\end{itemize}

\section{Lack of anticoncentration using scaled collision probability}
\label{lack_anticoncentration}
In this section, we will show how our random circuit ensemble exhibits a lack of anticoncentration according to \cref{Definiton: concentration}. The noise is modeled by a mixture of amplitude damping and depolarizing noise. That is, the combined noise channel could either be $\mathcal{N}^{(\textup{amp})}_q\circ \mathcal{N}^{(\textup{dep})}_p$ or be $\mathcal{N}^{(\textup{dep})}_p\circ \mathcal{N}^{(\textup{amp})}_q$. 
 \clearpage
\noindent Let us observe the following identities:
\begin{equation}
|0\rangle \langle 0| = \frac{I_2 + \sigma_z}{2},
\end{equation}
\begin{equation}
|1\rangle \langle 1| = \frac{I_2 - \sigma_z}{2},
\end{equation}
where $I_2$ is the single-qubit identity operator and $\sigma_z$ is the single-qubit Pauli-$\mathsf{Z}$ operator. 
Using these identities, for an ensemble $\mathcal{B}$, one could rewrite \cref{eq_collision} as:
\begin{equation}
\mathcal{Z} = \underset{\mathcal{B}}{\mathbb{E}} \left[ \sum_{p \in \{0, 3\}^n, p \neq 0^n}  \Tr(\sigma_p ~\mathcal{C}(\ketbra{0}{0}))^2 \right],
\end{equation}
where $\sigma_p$ is an $n$--qubit Pauli operator. Note that $\sigma_0 = I_2$ and $\sigma_3 = \sigma_z$. Additionally, note that
\begin{equation}
\underset{\mathcal{B}}{\mathbb{E}} \left[ \Tr(\sigma_p ~\mathcal{C}(\ketbra{0}{0}))^2 \right] = \underset{\mathcal{B}}{\mathbb{E}} \left[ \Tr(\sigma_p \otimes \sigma_p ~\mathcal{C}(\ketbra{0}{0}) \otimes \mathcal{C}(\ketbra{0}{0})) \right]. 
\end{equation}

\noindent We prove the following theorem, which shows a lack of anticoncentration for our noise model. 
\begin{thm}
\label{first_theorem}
Let $\mathcal{B}$ be an ensemble of noisy random quantum circuits with noise channel $\mathcal{N}$, where $\mathcal{N}$ is either $\mathcal{N}^{(\textup{amp})}_q\circ \mathcal{N}^{(\textup{dep})}_p$ or be $ \mathcal{N}^{(\textup{dep})}_p\circ \mathcal{N}^{(\textup{amp})}_q$. Then,
\begin{equation}~\label{eq:collision_lower_bound_statement}
\mathcal{Z} \geq \left(1+r^2\right)^n - 1, 
\end{equation}
where 
\begin{align}\label{eq:rvalue}
    r &\coloneqq \Bigg\{ 
    \begin{array}{ll}
        q, & \mathcal{N} = 
        \mathcal{N}^{(\textup{amp})}_q \circ \mathcal{N}^{(\textup{dep})}_p, \\
        q(1-p), & \mathcal{N} = \mathcal{N}^{(\textup{dep})}_p \circ \mathcal{N}^{(\textup{amp})}_q. 
    \end{array} 
\end{align}
\end{thm}
\begin{proof}

\noindent For pedagogical reasons, before we generalize to $n$ qubits, let us first consider just the single qubit case. In the single-qubit case, we have 
\begin{align}
    \mathcal{Z} 
    &= \underset{\mathcal{B}}{\mathbb{E}} \left[\Tr(\sigma_z ~\mathcal{C}(\ketbra{0}{0}))^2 \right] \\
    &= \underset{\mathcal{B}}{\mathbb{E}} \left[ \Tr(\sigma_z \otimes \sigma_z ~\mathcal{C}(\ketbra{0}{0}) \otimes \mathcal{C}(\ketbra{0}{0})) \right]
\end{align}
by definition. 
Let $\mathcal{N}$ be either $\mathcal{N}^{(\textup{amp})}_q\circ \mathcal{N}^{(\textup{dep})}_p$ or $ \mathcal{N}^{(\textup{dep})}_p\circ \mathcal{N}^{(\textup{amp})}_q$. 
Let 
\begin{equation}
\rho = \mathcal{N} \left(U_1 (\tilde{\rho}) U_1^\dagger \right), 
\end{equation}
where $\tilde{\rho}$ is the state just before the last block. 
Additionally, let
\begin{equation}
\rho' = U_1 (\tilde{\rho}) U_1^{\dagger}.
\end{equation}
\noindent By definition of the adjoint map, we have  
    \begin{align}
        \mathcal{Z} &= \underset{U_1}{\mathbb{E}} \left[ \Tr(\sigma_z \otimes \sigma_z ~\rho \otimes \rho) \right] \\
        & = \underset{U_1}{\mathbb{E}} \left[ \Tr(\sigma_z \otimes \sigma_z ~\mathcal{N}(\rho') \otimes \mathcal{N}(\rho') \right] \\
        \label{eqlast}
        &= \underset{U_1}{\mathbb{E}} \left[ \Tr(\mathcal{N}^{\dagger}(\sigma_z) \otimes \mathcal{N}^{\dagger}(\sigma_z) ~ \rho' \otimes \rho') \right]. 
    \end{align}

\noindent In \eqref{eqlast}, we have used the fact that for a CPTP map, $\mathcal{N}^{\dagger} \circ \mathcal{N}$ is equal to the identity map. We can explicitly compute $\mathcal{N}^{\dagger}(\sigma_z)$ as 
\begin{equation}~\label{eq:expand_N_dagger}
    \mathcal{N}^{\dagger}(\sigma_z) = r I_2 + (1-q)(1-p)\sigma_z, 
\end{equation}
where $r$ is defined in Eq.\eqref{eq:rvalue}.
Using this expansion, we have 
\begin{equation}~\label{eq:collision_lower_bound}
\begin{aligned}
\underset{U_1}{\mathbb{E}} \left[ \Tr(\mathcal{N}^{\dagger}(\sigma_z) \otimes \mathcal{N}^{\dagger}(\sigma_z) ~ \rho' \otimes \rho') \right]
&= r^2 \underset{U_1}{\mathbb{E}} \left[ \Tr(I_2 \otimes I_2 ~\rho' \otimes \rho') \right] + c_1 \underset{U_1}{\mathbb{E}} \left[ \Tr(I_2 \otimes \sigma_z ~\rho' \otimes \rho') \right]  \\ &+ c_1 \underset{U_1}{\mathbb{E}} \left[ \Tr(\sigma_z \otimes I_2 ~\rho' \otimes \rho') \right] + c_2^2 \underset{U_1}{\mathbb{E}} \left[ \Tr(\sigma_z \otimes \sigma_z ~\rho' \otimes \rho') \right], 
\end{aligned}
\end{equation}
where 
\begin{align}
    c_1  &\coloneqq \Bigg\{ 
    \begin{array}{ll}
        q(1-q)(1-p), & \mathcal{N} = 
        \mathcal{N}^{(\textup{amp})}_q \circ \mathcal{N}^{(\textup{dep})}_p, \\
        q(1-p)(1-p)^2, & \mathcal{N} = \mathcal{N}^{(\textup{dep})}_p \circ \mathcal{N}^{(\textup{amp})}_q,
    \end{array} \\ 
    c_2  &\coloneqq \Bigg\{ 
    \begin{array}{ll}
        (1-q), & \mathcal{N} = 
        \mathcal{N}^{(\textup{amp})}_q \circ \mathcal{N}^{(\textup{dep})}_p, \\
        (1-q)(1-p), & \mathcal{N} = \mathcal{N}^{(\textup{dep})}_p \circ \mathcal{N}^{(\textup{amp})}_q~.
    \end{array}
\end{align}

\noindent Now, we evaluate the terms one by one.
\begin{itemize}
\item The first term: 
\begin{equation}
\underset{U_1}{\mathbb{E}} \left[ \Tr(I_2 \otimes I_2 ~\rho' \otimes \rho') \right] = 1
\end{equation}
by definition. 
\item The second term:
\begin{align}
\underset{U_1}{\mathbb{E}} \left[ \Tr(I_2 \otimes \sigma_z ~\rho' \otimes \rho') \right] &= \underset{U_1}{\mathbb{E}} \Tr\bigg( \big[ U_1^{\dagger \otimes 2} (I_2 \otimes \sigma_z) U_1^{\otimes 2}\big] ~\tilde{\rho} \otimes \tilde{\rho}\bigg) \\
&=  \Tr\bigg(\underset{U_1}{\mathbb{E}}\big[ U_1^{\dagger \otimes 2} (I_2 \otimes \sigma_z) U_1^{\otimes 2}\big] ~ \tilde{\rho} \otimes \tilde{\rho} \bigg)\\
\label{eq:cross_term_1}
&= 0.
\end{align}
The second line follows because the linearity of trace and expectation to interchange the two operations suitably.
In the third line, we have used the fact that 
\begin{align}
\label{eq:vanish_cross_term}
    M_{U_{1}}[\sigma_i \otimes \sigma_j] = \underset{U_1}{\mathbb{E}}\left[U_1^{\dagger \otimes 2} \sigma_i \otimes \sigma_j U_1^{\otimes 2}\right] = 
    \underset{U_1}{\mathbb{E}}\left[U_1^{\otimes 2} \sigma_i \otimes \sigma_j U_1^{\dagger \otimes 2}\right] = 0
\end{align}
for $i \neq j$.
This equation follows because the Hermitian conjugate of a unitary is also a unitary and because the expectation is only over the unitary $U_1$.

\item By a similar reasoning as the second term, the third term:
\begin{equation}
\label{eq:cross_term_2}
\underset{U_1}{\mathbb{E}} \left[ \Tr(\sigma_z \otimes I_2 ~\rho' \otimes \rho') \right] = 0.
\end{equation}

\item The fourth term:
\begin{equation}
\label{eq:positive_redidual}
\underset{U_1}{\mathbb{E}} \left[ \Tr(\sigma_z \otimes \sigma_z ~\rho' \otimes \rho') \right] = \underset{U_1}{\mathbb{E}} \left[ \Tr(\sigma_z \rho')^2 \right] \geq 0.
\end{equation}
This is because $\Tr(\sigma_z \rho')^2$ is always a positive--valued random variable. 
\end{itemize}
Putting these analyses together, Eq.~\eqref{eq:collision_lower_bound} can be lower-bounded by $r^2$; that is, 
\begin{equation}
    \mathcal{Z} \geq r^2. 
\end{equation}

\noindent With this observation for the single-qubit case, we consider the general $n$-qubit case: 
\begin{equation}
    \mathcal{Z} 
    = \underset{\mathcal{B}}{\mathbb{E}} \left[ \sum_{p \in \{0, 3\}^n, p \neq 0^n} \Tr(\sigma_p \otimes \sigma_p ~\mathcal{C}(\ketbra{0}{0}) \otimes \mathcal{C}(\ketbra{0}{0})) \right]. 
\end{equation}
Let us fix ${p \in \{0, 3\}^n\backslash\{0^n\}}$ and consider 
\begin{equation}
    \underset{\mathcal{B}}{\mathbb{E}} \left[ \Tr(\sigma_p \otimes \sigma_p ~\mathcal{C}(\ketbra{0}{0}) \otimes \mathcal{C}(\ketbra{0}{0})) \right]. 
\end{equation}
Letting $\mathcal{C}$' denote the circuit obtained by removing the last layer of noise, we have 
\begin{equation}
\begin{aligned}
    &\underset{\mathcal{B}}{\mathbb{E}} \left[ \Tr(\sigma_p \otimes \sigma_p ~\mathcal{C}(\ketbra{0}{0}) \otimes \mathcal{C}(\ketbra{0}{0})) \right]
    \\&= \underset{\mathcal{B}}{\mathbb{E}} \left[ \Tr(\sigma_p \otimes \sigma_p ~ \mathcal{N}^{\otimes n} \circ \mathcal{C'}(\ketbra{0}{0}) \otimes \mathcal{N}^{\otimes n} \circ \mathcal{C'}(\ketbra{0}{0})) \right]
    \\&= 
    \underset{\mathcal{B}'}{\mathbb{E}} \left[ \Tr((\mathcal{N}^{\dagger})^{\otimes n}(\sigma_p) \otimes(\mathcal{N}^{\dagger})^{\otimes n}(\sigma_p) ~\mathcal{C}'(\ketbra{0}{0}) \otimes \mathcal{C}'(\ketbra{0}{0})) \right], 
    \end{aligned}
\end{equation}
where $\mathcal{B}'$ is the set of quantum channels obtained by removing the last layer of noise from the circuits in $\mathcal{B}$.
Now, we consider the expansion of $(\mathcal{N}^{\dagger})^{\otimes n}(\sigma_p)$. 
Let us write $p = p_1p_2\cdots p_n$, where $p_i \in \{0,3\}$ for $1\leq i \leq n$. 
Then, 
\begin{equation}~\label{eq:expand_N_dagger_n}
    (\mathcal{N}^{\dagger})^{\otimes n}(\sigma_p) = \bigotimes_{i = 1}^n \mathcal{N}^\dagger (\sigma_{p_i}). 
\end{equation}
Recalling \cref{eq:expand_N_dagger}, 
\begin{align}
    \mathcal{N}^\dagger (\sigma_{p_i}) = 
    \begin{cases}
        I_2 & p_i = 0 \\ 
        rI_2 + (1-q)(1-p)\sigma_z & p_i = 3~.
    \end{cases}
\end{align}
Note that $\mathcal{N}^\dagger(I_2) = I_2$ since the adjoint map of a quantum channel is unital. 
Substituting this relation to \cref{eq:expand_N_dagger_n}, we have 
\begin{equation}
    (\mathcal{N}^{\dagger})^{\otimes n}(\sigma_p) = r^{w_p}I_{2^n} + \sum_{\substack{q \in \{0,3\}^n\backslash\{0^n\} \\ w_q\leq w_p}} c_{q} \sigma_q, 
\end{equation}
where $c_q$ are nonnegative coefficients and $w_i$ is the number of non-identity Pauli operators in $\sigma_i$. 
Using this expression, 
\begin{align}
    &\underset{\mathcal{B}}{\mathbb{E}} \left[ \Tr(\sigma_p \otimes \sigma_p ~\mathcal{C}(\ketbra{0}{0}) \otimes \mathcal{C}(\ketbra{0}{0})) \right] 
    \\&= 
    r^{2w_p}\underset{\mathcal{B}'}{\mathbb{E}} \left[ \Tr(I_{2^n}\otimes I_{2^n} ~\mathcal{C}'(\ketbra{0}{0}) \otimes \mathcal{C}'(\ketbra{0}{0})) \right]\\
    &+ \sum_{\substack{q,q' \in \{0,3\}^n\backslash\{0^n\} \\ w_q\leq w_p}}c_qc_{q'} \underset{\mathcal{B}'}{\mathbb{E}} \left[ \Tr(\sigma_{q}\otimes \sigma_{q'} ~\mathcal{C}'(\ketbra{0}{0}) \otimes \mathcal{C}'(\ketbra{0}{0})) \right] \notag. 
\end{align}
The first term is 
\begin{equation}
    r^{2w_p}\underset{\mathcal{B}'}{\mathbb{E}} \left[ \Tr(I_{2^n}\otimes I_{2^n} ~\mathcal{C}'(\ketbra{0}{0}) \otimes \mathcal{C}'(\ketbra{0}{0})) \right] = r^{2w_p}.
\end{equation}
If $q \neq q'$, then with a similar argument as in \cref{eq:cross_term_1} and \cref{eq:cross_term_2}, we have 
\begin{equation}
    \underset{\mathcal{B}'}{\mathbb{E}} \left[ \Tr(\sigma_{q}\otimes \sigma_{q'} ~\mathcal{C}'(\ketbra{0}{0}) \otimes \mathcal{C}'(\ketbra{0}{0})) \right] = 0. 
\end{equation}
Here, we used the fact that the average statistics do not change even if one appends single-qubit Haar random unitary gates after a two-qubit Haar random unitary gate.
In addition, 
\begin{equation}
    \underset{\mathcal{B}'}{\mathbb{E}} \left[ \Tr(\sigma_{q}\otimes \sigma_{q} ~\mathcal{C}'(\ketbra{0}{0}) \otimes \mathcal{C}'(\ketbra{0}{0})) \right] = \underset{\mathcal{B}'}{\mathbb{E}} \left[ \Tr(\sigma_{q}\mathcal{C}'(\ketbra{0}{0}))^2 \right] \geq 0
\end{equation}
by a similar discussion in \cref{eq:positive_redidual}. 
By combining these observations, 
\begin{equation}
\label{cp}
    \mathcal{Z}\geq \sum_{p \in \{0, 3\}^n, p \neq 0^n} r^{2w_p}. 
\end{equation} 
Hence, by computing RHS with the binomial theorem, 
\begin{align}
        \mathcal{Z} 
        &\geq (1 + r^2)^{n} - 1, 
\end{align}
which completes the proof. 
\end{proof}

The lower bound~\eqref{eq:collision_lower_bound_statement} established in \cref{first_theorem} indicates that the scaled collision probability diverges in the limit of $n\to\infty$, \textit{i.e.}, the given circuit is not anticoncentrated, if the noise parameter $r \neq 0$. 
For example, recall that $r = q$ when $\mathcal{N} = \mathcal{N}^{(\textup{amp})}_q \circ \mathcal{N}^{(\textup{dep})}_p$. 
Hence, the circuit is not anticoncentrated as long as the circuit is affected by the amplitude damping noise, no matter how strong the depolarizing noise is. 
On the other hand, when $\mathcal{N} = \mathcal{N}^{(\textup{dep})}_p \circ \mathcal{N}^{(\textup{amp})}_q$, the noise parameter is $r = q(1-p)$.  
In this case, the circuit is not anticoncentrated for any positive $q$ unless $p = 1$. 
This slight difference originates from the order of the two kinds of noise. 
When the completely depolarizing noise comes after the amplitude damping noise, the completely depolarizing noise nullifies the effects of the amplitude damping noise. 

\begin{remark}
If we only have the depolarizing channel, then $q = 0$. Then, from \eqref{eq:rvalue}, $r = 0$. Hence, from \eqref{cp},
\begin{equation}
\mathcal{Z} \geq 0.
\end{equation}
This gives us a vacuous bound and our techniques fail to prove the lack of anticoncentration. This is consistent with the observation that random quantum circuits with the depolarizing noise indeed anticoncentrate \cite{Deshpande_2022}.
\end{remark}
\clearpage
\begin{remark}
While \cref{eq:vanish_cross_term} is obtained by a straightforward calculation with \cref{eq:werner_twirl_calc} in \cref{sec:Werner_twirl}, let us give qualitative reasoning for this relation. 
Observe that $M_{U_{1}}$ projects an input onto the symmetric subspace and the anti-symmetric subspace (See \cref{sec:Werner_twirl}). 
The projector $\Pi_{\mathrm{sym}}$ onto the symmetric subspace and the projector $\Pi_{\mathrm{antisym}}$ onto the anti-symmetric subspace can be expressed as 
\begin{align}
    \Pi_{\mathrm{sym}} 
    &= \ketbra{\Phi^{+}}{\Phi^{+}}+ \ketbra{\Phi^{-}}{\Phi^{-}} + \ketbra{\Psi^{+}}{\Psi^{+}} \\ 
    \Pi_{\mathrm{antisym}} 
    &= \ketbra{\Phi^{-}}{\Phi^{-}}, 
\end{align}
where 
\begin{equation}
    \ket{\Phi^{+}} \coloneqq \dfrac{\ket{00} + \ket{11}}{\sqrt{2}}, 
    \ket{\Phi^{-}} \coloneqq \dfrac{\ket{00} - \ket{11}}{\sqrt{2}},
    \ket{\Psi^{+}} \coloneqq \dfrac{\ket{01} + \ket{10}}{\sqrt{2}},
    \ket{\Psi^{-}} \coloneqq \dfrac{\ket{01} - \ket{10}}{\sqrt{2}}
\end{equation}
are the Bell states. 
In fact, for any Bell state $\ket{\Xi}$, we have $\bra{\Xi}|\sigma_i \otimes \sigma_j|\ket{\Xi} = 0$ when $i\neq j$. 
An intuitive justification can be given with the following simultaneous measurement scenario. 
Suppose that two parties---Alice and Bob--- share a Bell state $\ket{\Xi}$. 
Alice has an observable $\sigma_i$, and Bob has an observable $\sigma_j$. 
They measure their observables with the state $\ket{\Xi}$. 
In this case, the expectation value of this measurement is zero. 
Indeed, Bob obtains $1$ with probability $\tfrac{1}{2}$ and obtains $-1$ with probability $\tfrac{1}{2}$ regardless of Alice's measurement result because they share a maximally entangled state and they locally have different Pauli operators. 
Hence, from this observation, we have 
\begin{equation}
    \langle\sigma_i\otimes \sigma_j,\Pi_{\mathrm{sym}}\rangle = \langle\sigma_i\otimes \sigma_j,\Pi_{\mathrm{antisym}}\rangle = 0.  
\end{equation}
Hence, $\sigma_i\otimes \sigma_j$ has no component on the symmetric and anti-symmetric subspaces. 
Thus,  
\begin{equation}
    M_{U_{1}}[\sigma_i \otimes \sigma_j] = 0. 
\end{equation}
Note that this does not contradict that $M_{U_{1}}$ is a CPTP map because $\Tr[\sigma_i \otimes \sigma_j] = 0$ when $i\neq j$. 
\end{remark}

\section{Lack of anticoncentration using typical probabilities: preliminaries}
\label{sec1_typical}

In the next sections, we will prove that our setup exhibits a lack of anticoncentration according to \cref{newdef: concentration}. Before doing that, it will be helpful to prove some useful results.

To that effect, in this section, we calculate the first moment of two different noise models. We divide our proof into many subparts. First, let us calculate the exact expression for the first moment of the output probabilities of our distribution. We will then argue about typical probabilities, first in the low depth regime using a lightcone argument, and then in the high depth regime, using a second--moment inequality.

\subsection{First moment of output probabilities}

\begin{thm}
\label{first moment}
Let $\mathcal{B}$ be an ensemble of noisy random quantum circuits with noise channel $\mathcal{N}$. Furthermore, for a particular $x \in \{0, 1\}^n$, let $p_x$ be the corresponding outcome probability.
Then, the following hold. 
\begin{enumerate}
    \item If $\mathcal{N} = \mathcal{N}^{(\textup{amp})}_q\circ \mathcal{N}^{(\textup{dep})}_p$, 
    \begin{equation}
        \underset{\mathcal{B}}{{\mathbb{E}}}[p_{x}] = \frac{(1-q)^{w_x} (1+q)^{n-w_x}}{2^{n}}. 
    \end{equation}
    \item If $\mathcal{N} = \mathcal{N}^{(\textup{dep})}_p\circ \mathcal{N}^{(\textup{amp})}_q$, 
    \begin{equation}
        \underset{\mathcal{B}}{{\mathbb{E}}}[p_{x}] = \frac{(1-(1-p)q)^{w_x} (1+(1-p)q)^{n-w_x}}{2^{n}}. 
    \end{equation}
\end{enumerate}
\end{thm}
\begin{proof}
We will prove the two cases separately. 
\begin{enumerate} \item Observe that 
\begin{align}
\label{relation_1}
\left(\mathcal{N}^{(\textup{dep})}_p\right)^{\dagger}\circ \left(\mathcal{N}^{(\textup{amp})}_q\right)^{\dagger}\left(\ketbra{0}{0}\right) &= \left(1 - \frac{p}{2} + \frac{pq}{2} \right)\ketbra{0}{0} + \left(q + \frac{p}{2} - \frac{pq}{2} \right)\ketbra{1}{1}, \\
\label{relation_2}
\left(\mathcal{N}^{(\textup{dep})}_p\right)^{\dagger}\circ \left(\mathcal{N}^{(\textup{amp})}_q\right)^{\dagger}\left(\ketbra{1}{1}\right) &= \left(\frac{p}{2} - \frac{pq}{2} \right)\ketbra{0}{0} + \left(1 - q - \frac{p}{2} + \frac{pq}{2} \right)\ketbra{1}{1},
\end{align}
where for a quantum channel $\mathcal{N}$, $\mathcal{N}^{\dagger}$ represents its adjoint map. Now, construct a new ensemble $\mathcal{B}'$ by removing the last layer of the noise channel from the circuits in $\mathcal{B}$. 
Let $\mathcal{C}$ be a quantum circuit in $\mathcal{B}$, and let $\mathcal{C}'$ be the circuit obtained by removing the last layer of noise, namely, 
\begin{equation}
    \mathcal{C} = \mathcal{N}^{\otimes n}\circ \mathcal{C}'. 
\end{equation}
By the definition of the adjoint map, 
\begin{equation}
    \underset{\mathcal{B}}{\mathbb{E}}\left[p_{x} \right] 
    = 
    \underset{\mathcal{B}}{\mathbb{E}}\Tr\left[
    \ketbra{x}{x} \mathcal{C}\left(\ketbra{0^n}{0^n}\right)
    \right] 
    = \underset{\mathcal{B}'}{\mathbb{E}}\Tr\left[\left(\mathcal{N}^\dagger\right)^{\otimes n} \left(\ketbra{x}{x}\right) \mathcal{C}'\left(\ketbra{0^n}{0^n}\right)\right].
\end{equation}
By computing $\left(\mathcal{N}^\dagger\right)^{\otimes n} \left(\ketbra{x}{x}\right)$ using \cref{relation_1} and \cref{relation_2}, we have 
\begin{equation}
    \left(\mathcal{N}^\dagger\right)^{\otimes n} \left(\ketbra{x}{x}\right) =  \sum_{y \in \{0,1\}^n} \left(\prod_{k=1}^n c^{(x,y)}_{k}\right)\ketbra{y}{y},  
\end{equation}
where $c^{(x,y)}_{k}$ is defined by 
\begin{equation}
    c^{(x,y)}_{k} = \left\{ 
    \begin{array}{cc}
         \left(1 - \dfrac{p}{2} + \dfrac{pq}{2} \right) & \left(x_k,y_k\right) = (0,0) \\ \\
         \left(q + \dfrac{p}{2} - \dfrac{pq}{2} \right)& \left(x_k,y_k\right) = (0,1)\\ \\
        \left(\dfrac{p}{2} - \dfrac{pq}{2} \right) 
        & \left(x_k,y_k\right) = (1,0)\\ \\
        \left(1 - q - \dfrac{p}{2} + \dfrac{pq}{2} \right) & \left(x_k,y_k\right) = (1,1).
    \end{array}
    \right.
\end{equation}
On the other hand, we have 
\begin{equation}
    \underset{\mathcal{B}'}{\mathbb{E}}\left[\mathcal{C}'\left(\ketbra{0^n}{0^n}\right)\right] = \frac{I_{2^n}}{2^n}
\end{equation}
by considering the expectation over the last layer of random unitaries. 
Therefore, by combining these equations, 
\begin{align}
\label{hiding}
    \underset{\mathcal{B}}{\mathbb{E}}\left[p_{x} \right] 
    = \sum_{y \in \{0,1\}^n} \left(\prod_{k=1}^n c^{(x,y)}_{k}\right)\Tr\left[
    \ketbra{y}{y} \frac{I_{2^n}}{2^n}
    \right]
    =  \frac{1}{2^n} \sum_{y \in \{0,1\}^n} \left(\prod_{k=1}^n c^{(x,y)}_{k}\right) 
\end{align}
We can compute \cref{hiding} as 
\begin{align}
     \underset{\mathcal{B}}{\mathbb{E}}[p_{x}] 
     &= \frac{1}{2^n} \sum_{y \in \{0,1\}^n}
     \left(\prod_{k=1}^n c^{(x,y)}_{k}\right) \\ 
     \label{eq:px_intermediate}
     &= \frac{1}{2^n} \sum_{l=0}^{n-w_x} \binom{n-w_x}{l} \left(1 - \dfrac{p}{2} + \dfrac{pq}{2} \right)^l \left(q + \dfrac{p}{2} - \dfrac{pq}{2} \right)^{(n-w_x) - l} \\ &\quad \times \sum_{l'=0}^{w_x} \binom{w_x}{l'}  \left(\dfrac{p}{2} - \dfrac{pq}{2} \right)^{l'} 
        \left(1 - q - \dfrac{p}{2} + \dfrac{pq}{2} \right)^{w_x-l'} \\ 
    \label{eq:px_intermediate_2}
    &= \frac{1}{2^n}\left(1 - \dfrac{p}{2} + \dfrac{pq}{2} + q + \dfrac{p}{2} - \dfrac{pq}{2}\right)^{n-w_x}\left(\dfrac{p}{2} - \dfrac{pq}{2} + 1 - q - \dfrac{p}{2} + \dfrac{pq}{2}\right)^{w_x} \\ 
    &= \frac{(1+q)^{n-w_x}(1-q)^{w_x}}{2^n}, 
\end{align}
where to obtain \cref{eq:px_intermediate}, we divided the cases based on the hamming weight of string $x$. 
Then, \cref{eq:px_intermediate_2} follows from the binomial theorem. 

\item In this case, we have 
\begin{align}  \left(\mathcal{N}^{(\textup{amp})}_q\right)^{\dagger}\circ \left(\mathcal{N}^{(\textup{dep})}_p\right)^{\dagger}\left(\ketbra{0}{0}\right) &= \left(1 - \frac{p}{2} \right)\ketbra{0}{0} + \left(q(1-p) + \frac{p}{2} \right)\ketbra{1}{1}, \\
\left(\mathcal{N}^{(\textup{amp})}_q\right)^{\dagger}\circ \left(\mathcal{N}^{(\textup{dep})}_p\right)^{\dagger}\left(\ketbra{1}{1}\right) &= \left(\frac{p}{2} \right)\ketbra{0}{0} + \left(1 - q(1-p) - \frac{p}{2}  \right)\ketbra{1}{1}.      
\end{align}
With the same argument from case 1, we have 
\begin{equation}
    \underset{\mathcal{B}}{\mathbb{E}}[p_{x}] 
     = \frac{(1+q(1-p))^{n-w_x}(1-q(1-p))^{w_x}}{2^n}. 
\end{equation}
\end{enumerate}
\end{proof}

From \cref{first moment}, it is evident that the expected output probabilities of strings with higher Hamming weights are exponentially suppressed with respect to the Hamming weight.

\subsection{A discussion on marginal probabilities}
\label{marginal}
Our calculations to prove lack of anticoncentration, with respect to \cref{newdef: concentration}, in the low depth regime requires an argument based on lightcones. As we will soon see, this necessitates that we compute marginal probabilities, where the order of the marginal is determined by the size of the lightcone. The following corollary is immediate from \cref{first moment}.
\begin{corll}
\label{marginals}
Let $\mathcal{B}$ be an ensemble of noisy random quantum circuits with noise channel $\mathcal{N}$. 
For a binary string $x\in\{0,1\}^n$, consider a substring $y$ of $x$ with length $|y|$ and Hamming weight $w_y$.
Let $p_y$ be the corresponding marginal probability. 
Then, the following hold. 
\begin{enumerate}
    \item If $\mathcal{N} = \mathcal{N}^{(\textup{amp})}_q\circ \mathcal{N}^{(\textup{dep})}_p$, 
    \begin{equation}
        \underset{\mathcal{B}}{{\mathbb{E}}}[p_{y}] = \frac{(1-q)^{w_y} (1+q)^{|y|-w_y}}{2^{|y|}}. 
    \end{equation}
    \item If $\mathcal{N} = \mathcal{N}^{(\textup{dep})}_p\circ \mathcal{N}^{(\textup{amp})}_q$, 
    \begin{equation}
        \underset{\mathcal{B}}{{\mathbb{E}}}[p_{y}] = \frac{(1-(1-p)q)^{w_y} (1+(1-p)q)^{|y|-w_y}}{2^{|y|}}. 
    \end{equation}
\end{enumerate}
\end{corll}
\begin{proof}
    We only show the proof for $|y| = n-1$. 
    The other cases also follow similarly.
    Without loss of generality, we may assume $y = x_1x_2\cdots x_{n-1}$, where $x_i$ denotes the $i^{\text{th}}$ bit of $x$; that is, $y$ is the substring obtained by discarding the last bit of $x$. 
    Consider the case where $\mathcal{N} = \mathcal{N}^{(\textup{amp})}_q\circ \mathcal{N}^{(\textup{dep})}_p$.  
    The other case similarly holds. 

    Since the marginal probability $p_{x_1x_2\cdots x_{n-1}}$ is the sum of $p_{x_1x_2\cdots x_{n-1}0}$ and $p_{x_1x_2\cdots x_{n-1}1}$, 
    we have 
    \begin{align}
           \underset{\mathcal{B}}{{\mathbb{E}}}\left[p_{y}\right] 
           &= \underset{\mathcal{B}}{{\mathbb{E}}}\left[p_{x_1x_2\cdots x_{n-1}}\right] \\ 
           &= \underset{\mathcal{B}}{{\mathbb{E}}}\left[p_{x_1x_2\cdots x_{n-1}0} + p_{x_1x_2\cdots x_{n-1}1}\right] \\ 
           & =\underset{\mathcal{B}}{{\mathbb{E}}}\left[p_{x_1x_2\cdots x_{n-1}0} \right] + \underset{\mathcal{B}}{{\mathbb{E}}}\left[ p_{x_1x_2\cdots x_{n-1}1}\right]
    \end{align}
    Now, from \cref{first moment}, 
    \begin{align}
        \underset{\mathcal{B}}{{\mathbb{E}}}\left[p_{x_1x_2\cdots x_{n-1}0} \right] &= \dfrac{(1+q)^{n - w_y}(1-q)^{w_y}}{2^n}, \\ 
        \underset{\mathcal{B}}{{\mathbb{E}}}\left[ p_{x_1x_2\cdots x_{n-1}1}\right] &= \dfrac{(1+q)^{n - (w_y+1)}(1-q)^{w_y+1}}{2^n}. 
    \end{align}
    Therefore, 
    \begin{align}
        \underset{\mathcal{B}}{{\mathbb{E}}}\left[p_{y}\right] 
        & = \dfrac{(1+q)^{n - w_y}(1-q)^{w_y}}{2^n} + \dfrac{(1+q)^{n - (w_y+1)}(1-q)^{w_y+1}}{2^n} \\ 
        &= \dfrac{(1+q)^{(n-1) - w_y}(1-q)^{w_y}}{2^n}\left((1 + q) + (1 -q )\right)\\ 
        &= \dfrac{(1+q)^{(n-1) - w_y}(1-q)^{w_y}}{2^{n-1}} \\ 
        &= \dfrac{(1+q)^{|y| - w_y}(1-q)^{w_y}}{2^{|y|}}, 
    \end{align}
    where the last equality follows because $|y| = |x_1x_2\cdots x_{n-1}| = n-1$. 
\end{proof}

\section{Lack of anticoncentration at low depth}
\label{section: low depth}
In this section, we prove that for a sublogarithmic depth noisy random circuit ensemble, the probability weight on strings with Hamming weight at least $\frac{n}{2}$ is negligible in most circuits of the ensemble. This means that these circuits exhibit a lack of anticoncentration, as defined in \cref{newdef: concentration}. Our analysis is in spirit of the fine graining in the spirit of \cref{newdef2}.

\subsection{Notations and useful facts}
Note that introducing random variables $X_i$, corresponding to the $i^{\text{th}}$ bit of $n$-bit string $x$, we may write 
\begin{equation}
p_x = \Pr[X_1 = x_1] \Pr[X_2 = x_2 | X_1 = x_1] \cdots \Pr[X_n = x_n | X_1 = x_1, X_2 = x_2, \ldots, X_{n-1} = x_{n-1}].  
\end{equation}
Observe that $p_x$ is also a random variable, by definition.
\clearpage
\noindent Hence,
\begin{equation}
-\log p_x = -\frac{1}{n!} \underset{\sigma \in S_n}{\sum} ~\underset{i = 1}{\overset{n}{\sum}} \log  \Pr\bigg(X_{\sigma(i)} = x_{\sigma(i)} | X_{\sigma(1)} = x_{\sigma(1)}, \ldots, X_{\sigma(i-1)} = x_{\sigma(i-1)}\bigg),
\end{equation}
where $S_n$ is the permutation group on $n$ qubits. 
For a set of non-negative integers $J_i$, such that $i \notin J_i$, define
\begin{equation}
\langle Z_i \rangle_{J_i} = 2 \mathsf{
Pr}\big(X_i = x_i | \{X_j = x_j\}_{j \in J_i}\big) - 1.
\end{equation}
If there is no conditional dependence, then we drop the subscript $J_i$. Finally, let
\begin{equation}
A_{\sigma} = -\sum_{i=1}^{n}\langle Z_{\sigma(i)} \rangle_{\sigma(1, 2, \ldots, i-1)}.
\end{equation}
Using a lightcone type argument, the stated lemma follows. 
\begin{lemma}[\cite{Deshpande_2022}, Eq.~$58$]
\label{gullans paper}
For a quantum circuit $\mathcal{C}$, and for any $x \in \{0, 1\}^n$, let $p_x$ be the probability of getting $x$ in the output. Then,
\begin{equation}
-\log p_x \geq n \log 2 + \frac{1}{n!} \underset{\sigma \in S_n}{\sum} A_{\sigma} + \frac{1}{4 \cdot 4^d} \sum_{i=1}^n \langle Z_i \rangle^2.
\end{equation}
\end{lemma}
\subsection{Our results}
\noindent Now, we are ready to state our main theorem.

\begin{thm}
\label{concentration}
Let $\mathcal{B}$ be an ensemble of noisy random quantum circuits of depth $d$, with noise channel $\mathcal{N} = \mathcal{N}^{(\textup{amp})}_q\circ \mathcal{N}^{(\textup{dep})}_p$ or $\mathcal{N} = \mathcal{N}^{(\textup{dep})}_p\circ \mathcal{N}^{(\textup{amp})}_q$. 
Let $x \in \{0, 1\}^n$ and $w_x$ be the Hamming weight of $x$. Then, there exists a constant $t$ for every $x$ with $w_x \geq \frac{n}{2}$, and $\alpha \in (0, 1]$, such that,
\begin{equation}~\label{eq:low_depth_concentration}
\underset{n \rightarrow \infty}{\lim} \underset{\mathcal{B}}{\Pr}\left[p_x < \frac{\alpha}{2^n} \right] = 1,
\end{equation}
when $d < t \log (n)$.
\end{thm}
\begin{proof}
From \cref{conditional}, we have 
\begin{equation}
\underset{\mathcal{B}}{\mathbb{E}}\left[\Pr\big(X_i = x_i | \{X_j = x_j\}_{j \in J_i}\big)\right] = \dfrac{1}{2}\left[\bra{x_i}\mathcal{N}\left(I_2\right)\ket{x_i}\right]. 
\end{equation} 
Here, $\bra{x_i}\mathcal{N}\left(I_2\right)\ket{x_i}$ can be evaluated as follows. 
\begin{enumerate}
    \item If $\mathcal{N} = \mathcal{N}^{(\textup{amp})}_q\circ \mathcal{N}^{(\textup{dep})}_p$, 
    \begin{equation}
    \bra{x_i}\mathcal{N}\left(I_2\right)\ket{x_i}= \left\{
    \begin{array}{cc}
       1 + q  & (x_i = 0) \\
       1 - q  & (x_i = 1)
    \end{array}. 
    \right. 
    \end{equation}
    \item If $\mathcal{N} = \mathcal{N}^{(\textup{dep})}_p\circ \mathcal{N}^{(\textup{amp})}_q$, 
    \begin{equation}
    \bra{x_i}\mathcal{N}\left(I_2\right)\ket{x_i}= \left\{
    \begin{array}{cc}
       1 + (1-p)q  & (x_i = 0) \\
       1 - (1-p)q  & (x_i = 1)
    \end{array}. 
    \right. 
    \end{equation}
\end{enumerate}
Let us write 
\begin{equation}
    \bra{x_i}\mathcal{N}\left(I_2\right)\ket{x_i}= \left\{
    \begin{array}{cc}
       1 + r  & (x_i = 0) \\
       1 - r  & (x_i = 1)
    \end{array}
    \right. 
    \end{equation}
with $0\leq r \leq 1$, so that we can cover both cases. 
The expectation value of $A_{\sigma}$ over $\mathcal{B}$ is computed as 
\begin{equation}
\label{eq:A_sigma}
\begin{aligned}
\underset{\mathcal{B}}{\mathbb{E}}[A_\sigma] 
&= \sum_{i=1}^n \left(1 - \left[\bra{x_i}\mathcal{N}\left(I_2\right)\ket{x_i}\right]\right) \\ 
&= 
(n-w_x)\left(1 - (1 + r)\right) + w_x \left(1 - (1 - r)\right)\\
&= 2 w_x r - nr.
\end{aligned}
\end{equation}
Additionally, from \cref{appendixA}, 
\begin{equation}
\label{eq2}
\underset{\mathcal{B}_d}{\mathbb{E}}[\langle Z_i \rangle^2] \geq b e^{-a d}, \\
\end{equation}
for some positive constant $a,b$, for any $i \in [n]$. 
Moreover, letting
\begin{equation}
\label{random variable}
X = - \log p_x, 
\end{equation}
from \cite{Deshpande_2022}, it follows that 
\begin{equation}
\underset{\mathcal{B}}{\mathsf{Var}}(X) \leq 2n.
\end{equation}
From \cref{gullans paper} and \cref{eq:A_sigma}, 
\begin{equation}
\label{expectation}
\underset{\mathcal{B}}{\mathbb{E}}[X] \geq n \log 2 + (2 w_x r - n r) +  \frac{b}{4}ne^{-c d},
\end{equation}
for $c = \log 4 + a$. 
By Chebyshev's inequality,
\begin{equation}
\Pr[|X - \mathbb{E}(X)| \geq k] \leq \frac{\underset{\mathcal{B}}{\mathsf{Var}}(X)}{k^2}.
\end{equation}
Taking, say, $k = n^{0.01} \sqrt{\underset{\mathcal{B}}{\mathsf{Var}}(X)}$, we have
\begin{equation}
\Pr[|X - \mathbb{E}(X)| \leq \mathcal{O}(n^{0.51})]  \geq 1 - \frac{1}{n^{0.0001}}.
\end{equation}
Hence,
\begin{equation}
\lim_{n \rightarrow \infty} \Pr[-X \leq -\mathbb{E}(X) + \mathcal{O}(n^{0.51})] = 1.
\end{equation}
Putting back the values, from \cref{random variable} and \cref{expectation},
\begin{equation}
\lim_{n \rightarrow \infty} \underset{\mathcal{B}}{\Pr}\left[p_x \leq 2^{-n} \mathsf{exp}\left( -2w_x r + n r -  \frac{b}{4}ne^{-c d} + \mathcal{O}(n^{0.51})\right)\right] = 1.
\end{equation}
Now, if  
\begin{equation}~\label{eq:low_depth_condition}
2r\left(w_x  - \frac{n}{2}\right) +  \frac{b}{4}ne^{-c d} = \omega(n^{0.51}), 
\end{equation}
then 
\begin{equation}
    \mathsf{exp}\left( -2w_x r + n r -  \frac{b}{4}ne^{-c d} + \mathcal{O}(n^{0.51})\right) = O(1), 
\end{equation}
and we have Eq.~\eqref{eq:low_depth_concentration}.  
For $w_x \geq \frac{n}{2}$, when we pick $d < \tfrac{0.49}{c}\log(n)$, Eq.~\eqref{eq:low_depth_condition} is satisfied, and thus the theorem follows.
\end{proof}

\section{Lack of anticoncentration at high depth}
\label{sec:high depth}
We need a different technique to analyze sufficiently deep circuits and show that they satisfy \cref{Definiton: concentration}. This is because the lightcone type arguments, in \cref{gullans paper}, break down when the lightcone sizes become too large for sufficiently deep circuits. 
Because of this, the technique we use is bounding the second moment of output probabilities and then using Chebyshev's inequality to show concentration around the mean for our desired probabilities. Just as in \cref{section: low depth}, our results are fine grained, in the spirit of \cref{newdef2}. 
\subsection{Useful properties of sufficiently deep noisy random quantum circuits}
Before stating the main results, we prove the following intermediate proposition.

\begin{propo}
\label{first thm}
Let $\mathcal{N}$ be a single qubit noise channel, and let $\mathsf{N} = \mathcal{N} \otimes \mathcal{N}$ be two copies of $\mathcal{N}$ acting on two qubits. 
Define a two-qubit operator
\begin{equation}
    \tilde{M}_{U_1, \mathsf{N}} = M_{U_1} \circ \mathsf{N} \circ M_{U_1}
\end{equation}
with 
\begin{equation}
M_{U_1}[\rho] = \underset{U_1 \sim \mathcal{U}_{\textup{Haar}}}{\mathbb{E}}\bigg[U_1^{\otimes 2} \rho U_1^{\dagger \otimes 2} \bigg]. 
\end{equation}
Suppose that 
\begin{align}
    \label{eq:noise_1}
    \tilde{M}_{U_1, \mathsf{N}}(I_4) &= (1 - a)I_4 + 2a S, \\ 
    \label{eq:noise_2}
    \tilde{M}_{U_1, \mathsf{N}}(S) &= bI_4 + (1 - 2b) S
\end{align}
with $a > 0$ and $b > 0$, 
where $I_4$ is the 2-qubit identity operator and $S$ is the 2-qubit SWAP gate. 
Then, 
\begin{equation}
\label{eq:noisy_2_random}
\begin{aligned}
   &M_{U_1} \circ \underbrace{(\mathsf{N} \circ M_{U_1}) \circ (\mathsf{N} \circ M_{U_1}) \circ \cdots \circ (\mathsf{N} \circ M_{U_1})}_{m~\text{times}} \left[ \ketbra{0}{0}^{\otimes 2} \right] \\ &= \frac{1}{10}(2I + S) + \left[\dfrac{-2a+b}{10(a+2b)} + (1-a-2b)^m \left(-\dfrac{1}{30} - \dfrac{-2a+b}{10(a+2b)}\right)\right] (I_4-2S).
  \end{aligned}
\end{equation}
\end{propo}

\begin{remark}
Without loss of generality, for any noise $\mathcal{N}$, the effect of noise against $I_4$ and $S$ in the noisy random circuit can be expressed as in Eqs.~\eqref{eq:noise_1} and \eqref{eq:noise_2}. 
A more detailed discussion is in Appendix~\ref{sec:Werner_twirl}. 
\end{remark}

\begin{proof}
First, note that by properties of the Haar measure,
\begin{equation}
\label{eq1}
    M_{U_1}[\rho] = M_{U_1} \circ M_{U_1}[\rho].
\end{equation}
Hence,
\begin{align}
    &M_{U_1} \circ \underbrace{(\mathsf{N} \circ M_{U_1}) \circ (\mathsf{N} \circ M_{U_1})\circ \cdots \circ (\mathsf{N} \circ M_{U_1}))}_{m~\text{times}} \left[ \ketbra{0}{0}^{\otimes 2} \right]  \\ &=\underbrace{(M_{U_1} \circ \mathsf{N} \circ M_{U_1}) \circ (M_{U_1} \circ \mathsf{N} \circ M_{U_1})\circ \cdots \circ (M_{U_1} \circ \mathsf{N} \circ M_{U_1})}_{m~\text{times}} \left[ \ketbra{0}{0}^{\otimes 2} \right]\\ &=\underbrace{\tilde{M}_{U_1, \mathsf{N}} \circ \tilde{M}_{U_1, \mathsf{N}}\circ \cdots \tilde{M}_{U_1, \mathsf{N}}}_{m~\text{times}} \circ M_{U_1}\left[\ketbra{0}{0}^{\otimes 2} \right] \\ 
    &= \underbrace{\tilde{M}_{U_1, \mathsf{N}} \circ \tilde{M}_{U_1, \mathsf{N}}\circ \cdots \tilde{M}_{U_1, \mathsf{N}}}_{m~\text{times}}\left[\frac{1}{6}(I_4+S) \right]
\end{align}
In the fourth line, we have used the fact that
\begin{equation}
    M_{U_1} \left[\ket{0}\bra{0}^{\otimes 2}\right] = \frac{1}{6}(I_4+S).
\end{equation}
Note that  
\begin{equation}
   \frac{1}{6}(I_4+S) = \frac{1}{10}(2I+S) - \frac{1}{30}(I_4 - 2S). 
\end{equation}
It can be verified that
\begin{align}
    \label{eq:relation1}
     \tilde{M}_{U_1, \mathsf{N}}\left[2I+S\right] 
     &= \left(2I + S\right) + (-2a+b)\left(I_4-2S\right), \\
      \label{eq:relation2}
     \tilde{M}_{U_1, \mathsf{N}}\left[I_4-2S\right] 
     &= (1-a-2b)\left(I_4 - 2S\right)
\end{align}
Thus, using Eqs.~\eqref{eq:relation1} and \eqref{eq:relation2} repeatedly, we have the following relation 
\begin{align}
    &\underbrace{\tilde{M}_{U_1, \mathsf{N}} \circ \tilde{M}_{U_1, \mathsf{N}}\circ \cdots \tilde{M}_{U_1, \mathsf{N}}}_{m~\text{times}}\left[\frac{1}{6}(I_4+S) \right] \\ 
    &= \frac{1}{10}\underbrace{\tilde{M}_{U_1, \mathsf{N}} \circ \tilde{M}_{U_1, \mathsf{N}}\circ \cdots \tilde{M}_{U_1, \mathsf{N}}}_{m~\text{times}}\left[(2I+S)\right] \\&-\frac{1}{30} \underbrace{\tilde{M}_{U_1, \mathsf{N}} \circ \tilde{M}_{U_1, \mathsf{N}}\circ \cdots \tilde{M}_{U_1, \mathsf{N}}}_{m~\text{times}}\left[(I_4-2S)\right] \notag \\ 
    &= \frac{1}{10} (2I+S) + x_m(I_4-2S), 
\end{align}
where $\{x_n\}_n$ is defined by the following recurrence relation 
\begin{align}
    x_0 &= -\frac{1}{30}, \\
    x_{n+1} &= (1-a-2b)x_n + \frac{-2a+b}{10} \quad\quad (n\geq 0).
\end{align}
Solving this relation, we have 
\begin{equation}
    x_m = \frac{-2a+b}{10(a+2b)} + (1-a-2b)^m \left(-\frac{1}{30} - \frac{-2a+b}{10(a+2b)}\right), 
\end{equation}
which leads to Eq.~\eqref{eq:noisy_2_random}. 
\end{proof}
\subsection{Second moment of output probabilities}
Now, we make a connection between the second moment of the $n$-qubit depth-$d$ noisy random circuit and the output of $2$-qubit depth-$d$ single-qubit-Haar-random circuit. 
We give the proof in Appendix~\ref{appendix:proof_of_CR}, using a stat--mech model.  
\clearpage
\begin{thm}
\label{second moment probabilities}
Consider an ensemble $\mathcal{B}$ of depth-$d$ noisy random quantum circuits characterized by a noisy channel $\mathcal{N} =  
\mathcal{N}^{(\textup{amp})}_q\circ \mathcal{N}^{(\textup{dep})}_p$ or $\mathcal{N}^{(\textup{dep})}_p\circ \mathcal{N}^{(\textup{amp})}_q$. 
Suppose that noise parameters $(p,q)$ satisfies $(p,q) \neq (0,0)$; that is, we will not consider the noiseless case. Then, for $x \in \{0, 1\}^n$ with the Hamming weight $w_x \geq \tfrac{n}{2}$,
\begin{equation}
    ~\label{eq:bound_1}
    \underset{\mathcal{B}}{\mathbb{E}}[p_x^{2}] \leq \mu^n\eta^n \exp\left[n\dfrac{\nu}{\mu} \mathrm{e}^{-c(d-1)}\right], 
\end{equation}
where, 
\begin{align}
    \mu &\coloneqq \frac{1}{4} + \frac{r^2}{12c} ~~~~~~(\geq 0), \\ 
    \nu &\coloneqq \dfrac{1}{12} - \dfrac{r^2}{12c} ~~~~~~(\geq 0), \\ 
    \eta &= 1 - r^2 ~~~~~~(\geq 0), 
\end{align}
with  
\begin{align}
    c &\coloneqq 1 - (1-p)^2(1-q)\left(1 - \frac{q}{3}\right) \\ 
    r &\coloneqq \Bigg\{ 
    \begin{array}{ll}
        q, & \mathcal{N} = 
        \mathcal{N}^{(\textup{amp})}_q \circ \mathcal{N}^{(\textup{dep})}_p, \\
        q(1-p), & \mathcal{N} = \mathcal{N}^{(\textup{dep})}_p \circ \mathcal{N}^{(\textup{amp})}_q.
    \end{array}
\end{align}
\end{thm}
\begin{remark}
    In the statement, since $(p,q) \neq (0,0)$, $c$ is always larger than $0$ by definition. 
    Thus, $\mu$ and $\nu$ are always well-defined. 
\end{remark}
\subsection{Facts about concentration inequalities}

\noindent Before stating our results, let us state some useful facts about concentration inequalities, in the context of random quantum circuits. We consider $x$ with $w_x \geq \tfrac{n}{2}$.
By Chebyshev's inequality, 
\begin{equation}
\Pr[|X - \mathbb{E}(X)| \geq k] \leq \frac{\mathsf{Var}(X)}{k^2}.
\end{equation}
For $\alpha \in (0,1)$, letting 
$k = \tfrac{\alpha}{2^{n}}$, 
we have 
\begin{equation}
\begin{aligned}
    \Pr\left[\left|p_x - \underset{\mathcal{B}}{\mathbb{E}}(p_x)\right| \geq \frac{\alpha}{2^{n}}\right] 
    \leq \frac{\underset{\mathcal{B}}{\mathsf{Var}}(p_x)}{(\alpha/2^n)^2}
    \leq \frac{\underset{\mathcal{B}}{\mathbb{E}}[p_x^2]}{(\alpha/2^n)^2}. 
\end{aligned}
\end{equation}
Since 
\begin{equation}
    \left|p_x - \underset{\mathcal{B}}{\mathbb{E}}(p_x)\right| < \dfrac{\alpha}{2^n} \Rightarrow p_x <  \underset{\mathcal{B}}{\mathbb{E}}(p_x) + \dfrac{\alpha}{2^n}, 
\end{equation} 
we have the following bound: 
\begin{equation}
\label{eq:prob_bound}
\underset{\mathcal{B}}{\Pr}\left[p_x < \frac{\alpha}{2^n} + \underset{\mathcal{B}}{\mathbb{E}}(p_x) \right] \geq 1 - \frac{4^n\underset{\mathcal{B}}{\mathbb{E}}[p_x^2]}{\alpha^2}. 
\end{equation}

\noindent From \cref{first moment}, we can write 
\begin{equation}
    \underset{\mathcal{B}}{\mathbb{E}}[p_x] = \dfrac{(1-r)^{w_x}(1+r)^{n - w_x}}{2^n}. 
\end{equation}
If $w_x \geq \tfrac{n}{2}$, 
    \begin{equation}
    \label{citeeq}
        \underset{\mathcal{B}}{\mathbb{E}}[p_x] = \dfrac{(1-r)^{w_x}(1+r)^{n - w_x}}{2^n} \leq \dfrac{(1-r^2)^{\tfrac{n}{2}}}{2^n}.  
    \end{equation}

\subsection{Our results}
We are now ready to state our main theorem.

\begin{thm}
\label{concentration2}
Let $\mathcal{B}$ be an ensemble of noisy random quantum circuits of depth $d$. Let the noise channel be $\mathcal{N}$ and let $d = \Omega(\log n)$. Then, the following statements hold:
\begin{enumerate}
\item When $\mathcal{N} = \mathcal{N}^{(\textup{amp})}_q\circ \mathcal{N}^{(\textup{dep})}_p$, then For every $x \in \{0, 1\}^n$ with $w_x \geq \frac{n}{2}$ and $\alpha \in (0, 1]$,
\begin{equation}
\label{eq: theorem 7_1}
\underset{n \rightarrow \infty}{\lim} \underset{\mathcal{B}}{\Pr}\left[p_x < \frac{\alpha}{2^n} \right] = 1,
\end{equation}
as long as 
\begin{equation}
\label{aeq}
    \frac{q^2 +2}{(q-3)(q-1)} > (1-p)^2. 
\end{equation}
\item When $\mathcal{N} = \mathcal{N}^{(\textup{dep})}_p \circ \mathcal{N}^{(\textup{amp})}_q$, then or every $x$ with $w_x \geq \frac{n}{2}$, and $\alpha \in (0, 1]$,
\begin{equation}
\label{eq: theorem 7_2}
\underset{n \rightarrow \infty}{\lim} \underset{\mathcal{B}}{\Pr}\left[p_x < \frac{\alpha}{2^n} \right] = 1,
\end{equation}
as long as 
\begin{equation}
\label{aeq2}
     q > \frac{3}{4} - \frac{1}{2(1-p)^2}. 
\end{equation}
\end{enumerate}
\end{thm}

\begin{proof}
From \eqref{citeeq}, for any $\beta \in (0,1)$, there exists $\alpha \in (0,1)$ and sufficiently large $n$ such that 
    \begin{align}
         \underset{\mathcal{B}}{\mathbb{E}}[p_x] +  \dfrac{\alpha}{2^n} < \dfrac{\beta}{2^n},
    \end{align}
    unless $r = 0$. 
    With such choice of $\alpha$ and $n$, by \cref{eq:prob_bound} and \cref{second moment probabilities},
    \begin{equation}
        \underset{\mathcal{B}}{\Pr}\left[p_x < \frac{\beta}{2^n}  \right] \geq 1 - \frac{\left(4\mu\eta\right)^{n}}{\alpha^2} \exp\left[n\dfrac{\nu}{\mu} \mathrm{e}^{-c(d-1)}\right]. 
    \end{equation}
    Since $c > 0$, when 
    \begin{equation}
        d \geq \dfrac{\log(n)}{c}, 
    \end{equation} 
    we have 
    \begin{equation}
        \exp\left[n\dfrac{\nu}{\mu} \mathrm{e}^{-c(d-1)}\right] = \exp\left[\dfrac{\nu}{\mu} \mathcal{O}(1) \right] = \mathcal{O}(1).
    \end{equation}
    Therefore, for such depth $d$, if
    \begin{equation}
        0\leq 4\mu\eta <  1, 
    \end{equation} 
    \cref{eq: theorem 7_1} is satisfied, because
    \begin{equation}
    \begin{aligned}
        \underset{\mathcal{B}}{\Pr}\left[p_x < \frac{\beta}{2^n}  \right] 
        \geq 1 - \frac{\left(4\mu\eta\right)^{n}}{\alpha^2}\mathcal{O}(1) 
        \xrightarrow{n\to\infty} 1.
    \end{aligned}
\end{equation}
Thus, to satisfy \cref{eq: theorem 7_1}, we have to make sure that $r \neq 0$ and $0 \leq 4\mu\eta <  1$. 

\begin{enumerate}
 \item Consider $\mathcal{N} = \mathcal{N}^{(\textup{amp})}_q \circ \mathcal{N}^{(\textup{dep})}_p$. 
First, since $r\neq 0$, we must have $q \neq 0$.
Next, by definition, we have 
\begin{equation}
    4\mu\eta = (1-q^2)\frac{q^2 + 3 - (1-p)^2(3-q)(1-q)}{3 - (1-p)^2(3-q)(1-q)}. 
\end{equation}
Obviously, $0\leq 4\mu\eta$. 
Thus, given parameters $0\leq p,q\leq 1$, 
we only have to check if 
\begin{equation}
    (1-q^2)\frac{q^2 + 3 - (1-p)^2(3-q)(1-q)}{3 - (1-p)^2(3-q)(1-q)} < 1. 
\end{equation}
This is equivalent to 
\begin{equation}
    0 < \left[1 - (1-p)^2\right]q^2 + 4(1-p)^2q + (2-3(1-p)^2). 
\end{equation}
If this condition is satisfied, we satisfy \cref{eq: theorem 7_2}. Simplifying this, we get \cref{aeq}.

\item Consider $\mathcal{N} = \mathcal{N}^{(\textup{dep})}_p \circ \mathcal{N}^{(\textup{amp})}_q$. 
First, since $r\neq 0$, it must follow that $p \neq 1$ and $q \neq 0$. 
Next, we have 
\begin{equation}
    4\mu\eta = (1-q^2(1-p)^2)\frac{q^2(1-p)^2 + 3 - (1-p)^2(3-q)(1-q)}{3 - (1-p)^2(3-q)(1-q)}. 
\end{equation}
By definition, $0\leq 4\mu\eta$. 
Thus, given parameters $0\leq p,q\leq 1$, 
we only have to check if 
\begin{equation}
    (1-q^2(1-p)^2)\frac{q^2(1-p)^2 + 3 - (1-p)^2(3-q)(1-q)}{3 - (1-p)^2(3-q)(1-q)} < 1. 
\end{equation}
Indeed, we can solve this inequality with respect to $q$ as 
\begin{equation}
    q > \frac{3}{4} - \frac{1}{2(1-p)^2}. 
\end{equation}
If this condition is satisfied, we to satisfy \cref{eq: theorem 7_2}. 
\end{enumerate}
\end{proof}
\begin{remark}
Note that the restrictions on $p$ and $q$, as given by \cref{aeq} and \cref{aeq2} are limitations of the proof technique, and do not necessarily mean that the output distribution behaves any differently for values of $p$ and $q$ that do not satisfy these constraints. 
\end{remark}

\begin{remark}
\label{remark7}
Our results in \cref{section: low depth} and \cref{sec:high depth} mean that our setup, provided the constraints in \cref{section: low depth} and \cref{sec:high depth} are satisfied, is never $2^{n-1}$--anticoncentrated, according to \cref{newdef2}.
\end{remark}

\section{Generalizing to arbitrary noise channels}
\label{section: general noise1}
In this section, 
we consider a general case, where the noise map $\mathcal{N}$ is characterized using parameters $t_{ij}$ with $0\leq i \leq 3$ and $1\leq j \leq 3$ as 
\begin{align}
    I_2 &\to I_2 + t_{01}\sigma_x + t_{02}\sigma_y + t_{03} \sigma_z \\ 
    \sigma_x &\to t_{11}\sigma_x + t_{12}\sigma_y + t_{13} \sigma_z \\
    \sigma_y &\to t_{21}\sigma_x + t_{22}\sigma_y + t_{23} \sigma_z \\
    \sigma_z &\to t_{31}\sigma_x + t_{32}\sigma_y + t_{33} \sigma_z. 
\end{align}
Since a set $\{I_2,\sigma_x,\sigma_y,\sigma_z\}$ forms a basis for the space of single-qubit operators, 
an arbitrary single-qubit quantum channel can be expressed in this form. 
Note that, conversely, a map expressed in this form is not necessarily a quantum channel. 

\subsection{Lack of anticoncentration using collision probability}
We show an extension of Theorem~\ref{first_theorem} and prove that the ensemble $\mathcal{B}$ fails to anticoncentrate. 
\begin{thm}
\label{first_theorem_extension}
Let $\mathcal{B}$ be an ensemble of noisy random quantum circuits with the general noise channel $\mathcal{N}$. Then,
\begin{equation}
\mathcal{Z} \geq (1+t_{03}^2)^n - 1. 
\end{equation}
\end{thm}
\begin{proof}

\noindent For simplicity, let us first consider just the single qubit case as in Theorem~\ref{first_theorem}.
Let 
\begin{equation}
\rho = \mathcal{N} \left(U_1 (\tilde{\rho}) U_1^\dagger \right)
\end{equation}
and let
\begin{equation}
\rho' = U_1 (\tilde{\rho}) U_1^{\dagger}
\end{equation}
with $\tilde{\rho}$ being the state just before the last block. 
\noindent For a single qubit, by the definition of the adjoint map, 
\begin{align}
\mathcal{Z} &= \underset{U_1}{\mathbb{E}} \left[ \Tr(\sigma_z \otimes \sigma_z ~\rho \otimes \rho) \right] \\
&= \underset{U_1}{\mathbb{E}} \left[ \Tr(\sigma_z \otimes \sigma_z ~\mathcal{N}(\rho') \otimes \mathcal{N}(\rho')) \right] \\
&= \underset{U_1}{\mathbb{E}} \left[ \Tr(\mathcal{N}^{\dagger}(\sigma_z) \otimes \mathcal{N}^{\dagger}(\sigma_z) ~ \rho' \otimes \rho') \right] \\
&= t_{03}^2 \underset{U_1}{\mathbb{E}} \left[ \Tr(I_2 \otimes I_2 ~\rho' \otimes \rho') \right] + t_{13}^2 \underset{U_1}{\mathbb{E}} \left[ \Tr(\sigma_x \otimes \sigma_x ~\rho' \otimes \rho') \right]  \\ &+ t_{23}^2 \underset{U_1}{\mathbb{E}} \left[ \Tr(\sigma_y \otimes \sigma_y ~\rho' \otimes \rho') \right] + t_{33}^2 \underset{U_1}{\mathbb{E}} \left[ \Tr(\sigma_z \otimes \sigma_z ~\rho' \otimes \rho') \right] + \sum_{i\neq j} t_{i3}t_{j3} \underset{U_1}{\mathbb{E}} \left[ \Tr(\sigma_i \otimes \sigma_j ~\rho' \otimes \rho') \right]. 
\end{align}
By using \cref{eq:vanish_cross_term}, 
\begin{equation}
    \underset{U_1}{\mathbb{E}} \left[ \Tr(\sigma_i \otimes \sigma_j ~\rho' \otimes \rho') \right] = 0 
\end{equation}
for all $i\neq j$. 
In addition, for $p = x,y,z$, 
\begin{equation}
    \underset{U_1}{\mathbb{E}} \left[ \Tr(\sigma_p \otimes \sigma_p ~\rho' \otimes \rho') \right] 
    = \underset{U_1}{\mathbb{E}} \left[ \Tr(\sigma_p \rho')^2 \right] \geq 0.  
\end{equation}
Therefore, 
\begin{align}
   \mathcal{Z} \geq t_{03}^2 \underset{U_1}{\mathbb{E}} \left[ \Tr(I_2 \otimes I_2 ~\rho' \otimes \rho') \right] = t_{03}^2. 
\end{align}
With this observation, 
by a similar discussion as in \cref{first_theorem}, 
we analyze the general $n$-qubit case;  
\begin{align}
        \mathcal{Z} 
        &= \underset{\mathcal{B}}{\mathbb{E}} \left[ \sum_{p \in \{0, 3\}^n, p \neq 0^n} \Tr(\sigma_p \otimes \sigma_p ~\mathcal{C}(\ketbra{0}{0}) \otimes \mathcal{C}(\ketbra{0}{0})) \right] \\ 
        &\geq \sum_{p \in \{0, 3\}^n, p \neq 0^n} t_{03}^{2w_p} \\ 
        \label{bound}
        &= (1 + t_{03}^2)^{n} - 1, 
    \end{align}
which completes the proof. 
\end{proof}

\begin{remark}
\cref{first_theorem_extension} implies that for any noise channel such that $t_{03}$ is a non-zero constant, the ensemble $\mathcal{B}$ fails to anticoncentrate. 
\end{remark}

\begin{remark}
By a similar argument to \cref{first_theorem_extension}, one can show that for any noise channel such that $t_{01}$ is a non-zero constant, the ensemble $\mathcal{B}$ fails to anticoncentrate when the collision probabilities are defined with respect to the Hadamard basis.
\end{remark}
\begin{remark}
For any unital channel, $t_{01} = t_{02} = t_{03} = 0$, so, by plugging into \eqref{bound}, $\mathcal{Z} \geq 0$ gives a vacuous bound.
\end{remark}

\subsection{Lack of anticoncentration using typical probabilities}

We can use similar arguments to what we did in \cref{section: low depth} and \cref{sec:high depth} to argue about the nature of the distribution. Just as there, here too the distribution has a lot of strings with very low probability weight. First, define the following quantities:
\begin{align}
    a &= \frac{t_{01}^2 + t_{02}^2 + t_{03}^2}{3}, \\ 
    b &= \frac{1}{2} - \frac{t_{01}^2 + t_{02}^2 + t_{03}^2 + t_{11}^2 + t_{12}^2 + t_{13}^2 + t_{21}^2 + t_{22}^2 + t_{23}^2 + t_{31}^2 + t_{32}^2 + t_{33}^2}{6},\\
    c &= a + 2b = 1 - \frac{t_{11}^2 + t_{12}^2 + t_{13}^2 + t_{21}^2 + t_{22}^2 + t_{23}^2 + t_{31}^2 + t_{32}^2 + t_{33}^2}{3},
\end{align}

\begin{align}
    \mu &= \frac{-t_{01}^2 - t_{02}^2 - t_{03}^2 + t_{11}^2 + t_{12}^2 + t_{13}^2 + t_{21}^2 + t_{22}^2 + t_{23}^2 + t_{31}^2 + t_{32}^2 + t_{33}^2 - 3}{4(t_{11}^2 + t_{12}^2 + t_{13}^2 + t_{21}^2 + t_{22}^2 + t_{23}^2 + t_{31}^2 + t_{32}^2 + t_{33}^2 - 3)}, \\ 
    \nu &= \frac{3(t_{01}^2 + t_{02}^2 + t_{03}^2) + t_{11}^2 + t_{12}^2 + t_{13}^2 + t_{21}^2 + t_{22}^2 + t_{23}^2 + t_{31}^2 + t_{32}^2 + t_{33}^2 - 3}{12(t_{11}^2 + t_{12}^2 + t_{13}^2 + t_{21}^2 + t_{22}^2 + t_{23}^2 + t_{31}^2 + t_{32}^2 + t_{33}^2 - 3)},\\
    \eta &= \sqrt{\max\left\{\left(1 + t_{03}\right)^2, \frac{t_{13}^2}{2} + \frac{t_{23}^2}{2} +
    \frac{t_{33}^2}{2} +\frac{\left(1 + t_{03}\right)^2}{2} \right\}}.
\end{align}

\subsubsection{Lack of anticoncentration at low depth}
We have an analogue of Theorem~\ref{concentration}, 
which shows a lack of anticoncentration for low-depth circuits. 

\begin{thm}
\label{lowdepth}
Let $\mathcal{B}$ be an ensemble of noisy random quantum circuits of depth $d$, with general noise channel $\mathcal{N}$. 
Let $x \in \{0, 1\}^n$ and $w_x$ be the Hamming weight of $x$. 
Suppose further that 
\begin{equation}
\left\langle\ketbra{0}{0},\mathcal{N}^d(\ketbra{0}{0})\right\rangle = \kappa + \tau\lambda^d
\end{equation}
with some $\kappa > \tfrac{1}{2}$, $\tau > 0$, and $\lambda\geq 0$. 
Then, if $t_{03} > 0$, there exists a constant $t$ for every $x$ with $w_x \geq \frac{n}{2}$, and $\alpha \in (0, 1]$, such that,
\begin{equation}
\underset{n \rightarrow \infty}{\lim} \underset{\mathcal{B}}{\Pr}\left[p_x < \frac{\alpha}{2^n} \right] = 1,
\end{equation}
when $d < t \log (n)$.
\end{thm}

\noindent In general, if we write 
\begin{equation}
    \mathcal{N}^n(\ketbra{0}{0}) = \frac{1}{2}\left(I_2 + x_n \sigma_x + y_n\sigma_y + z_n\sigma_z \right), 
\end{equation}
$x_n$, $y_n$, and $z_n$ are defined recursively as 
\begin{align}
    x_{n+1} &= t_{01} + t_{11}x_n + t_{21}y_n + t_{31}z_n \\ 
    y_{n+1} &= t_{02} + t_{12}x_n + t_{22}y_n + t_{32}z_n\\ 
    z_{n+1} &= t_{03} + t_{13}x_n + t_{23}y_n + t_{33}z_n
\end{align}
for $n\geq 0$ with 
\begin{align}
    x_0 &= 0 \\ 
    y_0 &= 0 \\ 
    z_0 &= 1. 
\end{align}
Thus, given noise $\mathcal{N}$ with parameters $\{t_{ij}:i=0,1,2,3; j = 1,2,3\}$, 
we may check if the conditions $\kappa > \tfrac{1}{2}$, $\tau > 0$, and $\lambda\geq 0$ are satisfied by explicitly solving the recurrence relation introduced above. 

\subsubsection{Lack of anticoncentration at high depth}
\begin{thm}
\label{highdepth}
Let $\mathcal{B}$ be an ensemble of noisy random quantum circuits of depth $d$, where each single--qubit noisy channel $\mathcal{N}$ is modeled as an arbitrary CPTP map with parameters $\{t_{ij}\}$. Let $d = \Omega(\log n)$. Then, for any $\alpha \in (0, 1]$
\begin{equation}
\underset{n \rightarrow \infty}{\lim} \underset{\mathcal{B}}{\Pr}\left[p_x < \frac{\alpha}{2^n} \right] = 1,
\end{equation}
as long as $w_x \geq n/2$ and 
\begin{align}
    &\mu \geq 0, \\
    &\nu \geq 0, \\ 
    &0 < c \leq 1, \\ 
    &0\leq 4\mu\eta < 1. 
\end{align}
\end{thm}
\begin{proof}
Recalling that 
\begin{equation}
    S = \tfrac{1}{2}\left(I_2\otimes I_2 + \sigma_x\otimes\sigma_x + \sigma_y\otimes\sigma_y + \sigma_z\otimes\sigma_z \right), 
\end{equation} 
we have
\begin{align}
    \tilde{M}_{U_1, \mathsf{N}}(I_4) &= (1 - a)I_4 + 2a S, \\ 
    \tilde{M}_{U_1, \mathsf{N}}(S) &= bI + (1 - 2b) S,
\end{align}
where the operators are as defined in Section \ref{sec:high depth}. With a similar discussion to Section \ref{sec:high depth}, we have 
\begin{equation}
    \underset{\mathcal{B}}{\mathbb{E}}[p_x^{2}] \leq \mu^n\eta^n \exp\left[n\dfrac{\nu}{\mu} \mathrm{e}^{-c(d-1)}\right]
\end{equation}
if $\mu,\nu \geq 0$ and $0< c \leq 1$. Then, following the argument in Section \ref{sec:high depth}, with the constraints being
\begin{align}
    &\mu \geq 0, \\
    &\nu \geq 0, \\ 
    &0 < c \leq 1, \\ 
    &0\leq 4\mu\eta < 1. 
\end{align}
\end{proof}

\begin{remark}
Just as in \cref{remark7}, our calculations in \cref{section: general noise1} indicate that the setup in \cref{section: general noise1} is not $2^{n-1}$--anticoncentrated, provided the constraints in \cref{lowdepth} and \cref{highdepth} are satisfied.
\end{remark}

\section{Effect of the last layer of noise}
\label{measurement}
Note that noises like the amplitude damping noise---our emblematic non--unital noise---try to push the output distribution towards a fixed state. However, a layer of random gates tries to "scramble" the distribution. In this sense, there are two opposite effects at play. This might lead one to conjecture that the behavior of the final distribution depends on which layer we end with: if ending with a layer of noiseless random gates causes anticoncentration, and if ending with a layer of amplitude damping noise causes lack of anticoncentration. 

However, we give strong evidence that this is not the case and that lack of anticoncentration occurs even if we terminate with a last layer of noiseless gates. Our results in this section are not as general as those of the other sections: they are only meant to justify our intuition. Furthermore, note that terminating with a last layer of gates is not a realistic assumption, as all known hardware has measurement noise immediately before the measurement operators, which can also be modelled as an amplitude damping noise channel; for instance, see  \cite{kandala2017hardware, Arute2019}. Because of this, the circuit model that follows is just a toy model for analysis. What we will show is that if we "fix" a last layer of noiseless single-qubit gates, then for all choices of this layer, apart from a set of choices with measure zero, we get provable lack of anticoncentration, according to \cref{Definiton: concentration}.

\subsection{Conventions}
Note that a single qubit gate $U$ is parametrized as
\begin{equation}
  U(\theta, \phi) = 
\begin{pmatrix}
&\cos \theta \cdot e^{i \phi} && \sin \theta \\
&-\sin \theta && \cos \theta \cdot e^{-i \phi}& \\
\end{pmatrix}.
\end{equation}
Let $U_i(\theta_i, \phi_i)$ be the unitary applied to the $i^{\text{th}}$ qubit in the last layer. 

\subsection{Proving lack of anticoncentration}
We will consider a fixed last layer of single qubit gates, as shown in \cref{fig101111}, and show that for almost all choices of this layer, the output distribution exhibits a lack of anticoncentration.

\begin{figure}[H]
 \begin{centering}
    \includegraphics[width=0.25\textwidth]{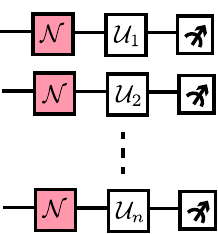}
    \caption{$U_1$, $U_2$, \ldots, $U_n$ are single qubit gates.
    }
    \label{fig101111}
     \end{centering}
\end{figure} 

As we discuss later in Appendix \ref{noiseless}, which strings have low probability weight and which ones have higher probability weight are now determined not by the Hamming weight of the strings but by which gates were applied in the last layer.

\begin{corll}
\label{last layer}
    Let $\mathcal{B}$ be an ensemble of amplitude-damped random quantum circuits, with noise strength $q$. Additionally, before measurement, for every $i \in [n]$, let $U_i(\theta_i, \phi_i)$---a single qubit, noiseless gate---be applied to qubit $i$. 
    Then,
    \begin{equation}
        \mathcal{Z} \geq (1 + q^2\cos^2 2\theta)^n - 1, 
    \end{equation}
    where 
    \begin{equation}
        \theta \coloneqq \argmin_{\theta_j: j\in[n]} |\cos 2\theta_j|. 
    \end{equation}
\end{corll}
\begin{proof}
    Note that by the action of $U_i \mathcal{N}(\cdot) U_i^\dagger$, the single-qubit identity operator $I_2$ will evolve as 
    \begin{equation}
        I_2 \to I_2 -q\cos\phi \sin2\theta_i \sigma_x +q \sin\phi\sin2\theta_i\sigma_y + q\cos2\theta_i \sigma_z. 
    \end{equation}
    Therefore, this last layer can be regarded as a noise map with $t_{03} = q\cos2\theta_i$, from which the statement follows directly using Theorem~\ref{first_theorem_extension}. 
\end{proof}

\begin{remark}
From \cref{last layer}, it holds that for any value of $\theta$, apart from those where $\cos 2\theta = 0$, the output distribution exhibits a lack of anticoncentration. The set of points for which this happens is a set of measure zero. 
\end{remark}

\begin{remark}
The utility of the last layer of gates is that it, in some sense, determines which strings have suppressed probability weights, and which ones have higher weights, in the output distribution of the circuit. This is illustrated in Appendix \ref{noiseless}. 

We can also characterize the nature of the output distribution, for certain parameter regimes, and argue about lack of anticoncentration with respect to \cref{newdef: concentration}. This is done in Appendix \ref{last2}. 
\end{remark}

\begin{remark}
Qualitatively, we believe that if a layer of amplitude damping noise is followed by a sufficiently shallow, geometrically local, random circuit, then the overall circuit still exhibits lack of anticoncentration. This is because amplitude damping "unscrambles" the output distribution, and a shallow depth geometrically local random circuit is not enough to counterbalance that and "scramble" it again, because shallow depth random circuits themselves show lack of anticoncentration \cite{Deshpande_2022}.
\end{remark}

 \section{Open problems}
 \label{open problems}
 Our paper motivates a number of open problems regarding the behavior of random circuits under non--unital noise.
 \begin{itemize}
 \item The most pertinent open question is whether the output distribution of random quantum circuits, with the non--unital noise models that we studied, are classically hard to sample from. 

 To answer this question, one potential approach is to figure out whether anticoncentration is a necessary feature in the classical sampling procedure devised in \cite{Bremner_2017, duan, aharonov2022polynomialtime} or whether it just comes up as a proof artefact during analysis of the sampler, and a different technique of analysis can potentially extend their result to regimes for which there is no anticoncentration. 

  If proof of classical hardness of sampling can be found, it might help in harnessing our results to design quantum advantage demonstrations with ensembles that have non--unital noise, which will complement existing quantum advantage demonstrations where the focus is on depolarizing noise \cite{Arute2019}. 

  Note that, to the best of our knowledge, no trivial sampling algorithm, for example those that sample from the fixed point of the noise channel, works for our non--unital noise models. While circuits with the depolarizing noise channel after every gate are at least inverse quasipolynomially close, in trace distance, to the maximally mixed state --- the fixed point of the depolarizing channel --- at sufficiently large depths, noise models like amplitude damping are not known to show such behavior. Certain standard techniques to show this closeness, like the data processing inequality of quantum relative entropy \cite{aharonov1996limitations, Wang_2021, franca}, do not hold for the amplitude damping channel.

  \begin{remark}Note that if the non--unital noise present is only in the last layer, and the rest of the circuit only has depolarizing noise, then techniques from \cite{aharonov2022polynomialtime} apply to classically sample from this circuit in polynomial time. One just stores an efficient truncated Fourier basis representation of the state, until the last layer of noise is encountered, and then just brute--force simulates the last layer of noise. But this trick does not work when every gate is followed by a noise channel that has a non--unital component.
 \end{remark}
 
\item Even though we get lack of anticoncentration, our results are different from those in \cite{Deshpande_2022}, in the sense that the lack of anticoncentration, for our case, is not "catastrophic enough" to ensure easiness of computing output probabilities to additive precision $2^{-n}$. It remains open whether that is a classically hard task.

\item  Moreover, is there any dependence on classical simulation complexity and the rate of the noise? Is there a "percolation threshold"? That is, if the amplitude damping noise strength is above a sufficiently high enough constant, do we get any phase transition in classical simulation complexity?

\item Even though the global distribution does not show anticoncentration, could it still be "locally" anticoncentrated --- that is, could there be collections of bitstrings such that the distribution looks flat "locally" when we consider the probability mass of only those bitstrings? Depending on how much locally anticoncentrated the distribution is, one could either design new hardness conjectures or modify the existing classical samplers to work in this regime.

\item It remains open whether, in the low noise regime, any amplitude damping present can be approximated by a global depolarizing noise: this is also known as the "white noise approximation." While techniques from \cite{dalzell2021random} indicate that this might be true for circuits which only have a last layer of amplitude damping noise, it remains open whether their techniques could be generalized to circuits with the middle layer of amplitude damping noise.

\item For a circuit ensemble, if the output distribution anticoncentrates, or is primarily "flat", it renders the ensemble useless for any gradient-descent based optimization tasks, because the optimization subroutine runs into the barren plateau problem \cite{McClean_2018, Wang_2021, napp2022quantifying}. 

More concretely, the gradient of the cost function of the optimization task vanishes in the ``flat" landscape and the optimization gets stuck. A concentrated distribution may potentially escape this phenomenon. This is because the second moment of a distribution diverges with respect to the number of qubits, and convergence of the second moment is necessary for barren plateaus for certain optimization setups, as was rigorously shown in \cite{napp2022quantifying}. So, for a wide range of cost functions, the "barren plateau" phenomenon is potentially escaped. 

Thus, an interesting avenue of future exploration is if our ensembles can be utilized for optimization tasks, or whether other subtleties hinder their practical use. 

\item Even though anticoncentration is a key ingredient of existing easiness of sampling results from random quantum circuits, easiness of computing the expectation value of certain observables may not require anticoncentration. 

In a recent work, a polynomial time classical estimator was proposed by \cite{shao2023simulating} to compute such expectation values for random quantum circuits with only depolarizing noise, by exploiting a special property of the noise --- the fact that only polynomially many Pauli paths have non--trival path weights. This is sketched in \cref{Expectation_Values} and this technique does not require anticoncentration. 

However, this property is very special to depolarizing noise and it remains open whether similar techniques could be extended to other noise channels, like the non--unital channels considered here. 

\item For our circuits, with respect to restricted parameter regimes, we showed how certain strings of the output distribution look like, in \cref{sec:high depth},  \cref{last2}, and \cref{section: general noise1}. We believe that the fact that we have to restrict our parameter regimes is just an artefact of the proof technique. It remains open whether we could extend our results to a wider set of parameters. 

\item It remains open how well a metric like linear cross entropy tracks circuit fidelity for our setup. More specifically, can we argue about the existence of a sharp crossover region, from low noise to high noise as was observed in \cite{morvan2023phase, ware2023sharp}, with the low noise regime being where linear cross entropy accurately tracks fidelity, and the high noise regime being where the connection breaks?

\item Note that our results in \cref{lack_anticoncentration} and \cref{section: low depth} are agnostic to the choice of architecture, as long as the circuits are parallel and geometrically local. However, our calculations in \cref{sec:high depth} make use of stat--mech models, which are known to work for $1$--D geometrically local circuits, but the techniques do not generalize to $2$--D. An open question is whether architecture agnostic techniques help us in proving the bound in \cref{sec:high depth}.
\end{itemize}

\section{Acknowledgements}
S.G. thanks Changhun Oh, Yunchao Liu, Yinchen Liu, Jordan Docter, Abhinav Deshpande, and Alexander Dalzell for helpful comments on a draft of the manuscript. K.S. thanks Christa Zoufal for helpful comments on a draft of the manuscript, and Sergey Bravyi, Bryan A. O'Gorman, Oles Shtanko, and Ryan Sweke for insightful discussions. B.F. and S.G. acknowledge support from AFOSR (award number FA9550-21-1-0008). 
K.K. was supported by a Mike and Ophelia Lazaridis Fellowship, the Funai Foundation, and a Perimeter Residency Doctoral Award.
This material is based upon work partially supported by the National Science Foundation under Grant CCF-2044923 (CAREER) and by the U.S. Department of Energy, Office of Science, National Quantum Information Science Research Centers as well as by DOE QuantISED grant DE-SC0020360. 
Any opinions, findings, and conclusions or recommendations expressed in this material are those of the author(s) and do not necessarily reflect the views of the National Science Foundation.

\newpage

\bibliographystyle{alphaurl}
\bibliography{MasterBib.bib}
\newpage 
\appendix
\section{Effect of twirling}
\label{sec: twirling}
In this section, we explain how twirling can be used to estimate the first moment of certain expressions related to random quantum circuits. 
\subsection{Preliminaries}
Let $\mathcal{N}$ be an arbitrary single--qubit noise channel. For a density matrix $\rho$, and a single qubit Haar random gate $U$, consider the following identity \cite{Emerson2005}:
\begin{equation}
\label{eqn: twirling}
\begin{aligned}
\Phi(\rho) &= \underset{U}{\mathbb{E}}\big[U^{\dagger} ~\mathcal{N}(U \rho U^{\dagger})~U \big] \\ &= (1-\lambda) \rho + \frac{\lambda}{2} I_2,
\end{aligned}
\end{equation}
for an appropriate choice of a constant $\lambda$. Note that the expression on the $\mathsf{LHS}$ is the expression of a depolarizing channel with noise strength $\lambda$. This operation, of averaging out an arbitrary error channel to convert it into a depolarizing channel, is known in literature as "twirling," and finds use in randomized benchmarking, randomized compiling, error mitigation etc \cite{Emerson2005, Wallman_2016}.

Twirling is a useful tool to gain insight into the first moments of quantities of interest for random quantum circuits. However, it is not useful to analyze second or higher-order moments. Thus, it is not a valid tool to analyze collision probabilities.

We will show this with four examples. A brief comment about notations: all gates shown in the figures below are Haar random gates, either single qubit ones or two-qubit ones depending on the context, and all noise channels shown are arbitrary single qubit  CPTP maps.

\subsection{Computing the first moment with a last layer of noiseless gates}
\label{section: last layer of noise}
Consider a noisy random quantum circuit ensemble $\mathcal{B}$ and let $\mathcal{N}$ be the noise after every gate, as shown in \cref{fig1011}. 
\begin{figure}[H]
 \begin{centering}
    \includegraphics[width=0.35\textwidth]{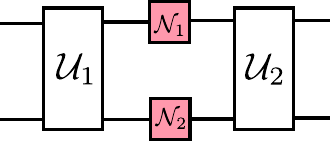}
    \caption{A portion of the circuit drawn from the ensemble $\mathcal{B}$. $U_1$ and $U_2$ are two qubit Haar random gates and $\mathcal{N}_1$ and $\mathcal{N}_2$ are arbitrary noise channels.
    }
    \label{fig1011}
     \end{centering}
\end{figure}

Since the Haar measure is left and right invariant with respect to composition by a unitary, one could rewrite the entire circuit, without loss of generality, by doing the following before and after every noise channel:

\begin{figure}[H]
 \begin{centering}
    \includegraphics[width=0.45\textwidth]{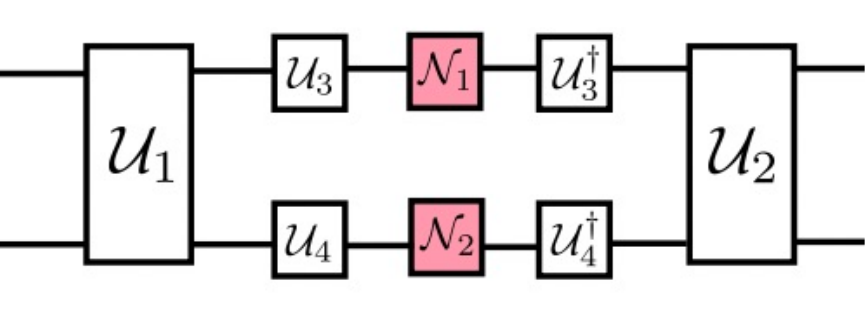}
    \caption{$U_3$ and $U_4$ are single qubit Haar random gates.
    }
    \label{fig10111}
     \end{centering}
\end{figure} 

\noindent Let $\rho^{(1)}$ be the marginal density matrix of the first qubit of \cref{fig10111} immediately before applying $U_3$ and let $\rho^{(1)}_f$ be the same immediately after applying $U_3^{\dagger}$. As is evident, 
\begin{equation}
\rho^{(1)}_{f} = \underset{U_3}{\mathbb{E}}\big[U_3^{\dagger} ~\mathcal{N}_1(U_3 ~\rho^{(1)} U_3^{\dagger})~U_3 \big],
\end{equation}
which, from \eqref{eqn: twirling}, is a depolarizing channel. Every noise channel in the circuit can be equivalently modeled as a depolarizing channel, in this way and calculating the first moment of the equivalent circuit suffices to calculate the first moment of the original circuit.

\subsection{Computing the first moment with a last layer of noisy gates}

If the circuit terminates with a last layer of noise, before the measurement layer, as illustrated in \cref{fig101}, then there is no way to twirl the last layer of noise, using the technique in \cref{section: last layer of noise}. Hence, these circuits cannot be twirled.

\begin{figure}[H]
 \begin{centering}
    \includegraphics[width=0.25\textwidth]{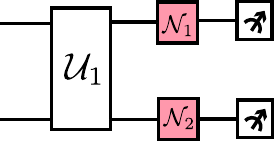}
    \caption{If the circuit terminates with a last layer of noiseless gates, then there is no way to "sandwich" the last layer of noise between a Haar random single qubit unitary and its adjoint. So, twirling does not work. 
    }
    \label{fig101}
     \end{centering}
\end{figure}

\subsection{Computing the expected linear cross--entropy score}
\label{cross-entropy}
The linear cross entropy is given by the following quantity:
\begin{equation}
\label{xeb}
\mathsf{XEB} = \underset{\mathcal{B}}{\mathbb{E}}\left[ \sum_{x \in \{0, 1\}^n} p_{\text{ideal}}(x) p_{\text{noisy}}(x) \right] = 2^n  \cdot \underset{\mathcal{B}}{\mathbb{E}}\left[ p_{\text{ideal}}(0^n) p_{\text{noisy}}(0^n) \right] ,
\end{equation}
where $p_{\text{ideal}}(x)$ is the probability of seeing bitstring $x$ in the output distribution of a circuit drawn from $\mathcal{B}$ with all the noise channels removed, and $p_{\text{noisy}}(x)$ is the probability of seeing bitstring $x$ in the noisy output distribution. We have assumed that there is no last layer of noise. Note that \eqref{xeb} can be written as
\begin{equation}
\mathsf{XEB} = 2^{n} \cdot \underset{\mathcal{B}}{\mathbb{E}}\bigg[\Tr\left(|0^n \rangle \langle 0^n| \otimes |0^n \rangle \langle 0^n| ~\rho \otimes \tilde{\rho} \right)\bigg],
\end{equation}

\noindent where $\rho$ is the density matrix corresponding to the final state of the circuit with all the noise channels removed, and $\tilde{\rho}$ is the density matrix corresponding to the circuit with noise. Twirling can be used to estimate this quantity, as is shown in \cref{fig10}.

\begin{figure}[H]
 \begin{centering}
    \includegraphics[width=0.8\textwidth]{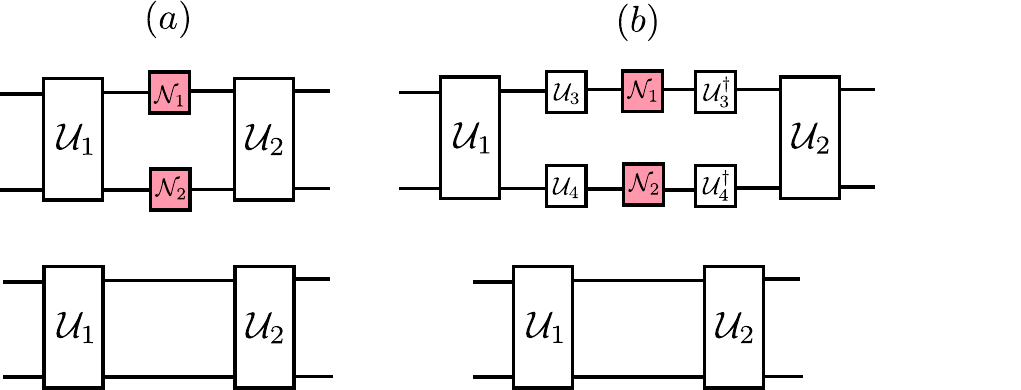}
    \caption{(a) A portion of the circuit that prepares $\tilde{\rho}$ and a portion of the circuit that prepares $\rho$ is shown. 
    (b) Note that $U U^{*} = \mathbb{I}$, where $\mathbb{I}$ is the single qubit identity operator, and $U$ is any unitary operator. This means there is no dependence of $U_3$ and $U_4$ in the "noiseless copy" of the circuit. Because of that, the quantity we are evaluating can be simplified with \eqref{eqn: twirling}, just like we did for first moment quantities previously.
    }
    \label{fig10}
     \end{centering}
\end{figure}

\subsection{Computing the second moment}
For a noisy ensemble $\mathcal{B}$, let the task be to compute
\begin{equation}
\label{cprob}
\underset{\mathcal{B}}{\mathbb{E}}\left[\sum_{x \in \{0, 1\}^n} p_x^2\right] = 2^n \cdot \underset{\mathcal{B}}{\mathbb{E}}\left[p_{0^n}^2\right],
\end{equation}
where we have assumed that there is no last layer of noise. Note that \cref{cprob} can be written as
\begin{equation}
2^n \cdot \underset{\mathcal{B}}{\mathbb{E}}\left[p_{0^n}^2\right] = 2^{n} \cdot \underset{\mathcal{B}}{\mathbb{E}}\bigg[\Tr\left(|0^n \rangle \langle 0^n| \otimes |0^n \rangle \langle 0^n| ~\tilde{\rho} \otimes \tilde{\rho} \right)\bigg],
\end{equation}

\noindent where $\tilde{\rho}$ is the density matrix corresponding to the final state of the noisy circuit. Here, we cannot hope to "twirl" the noise because both copies of the state are noisy, as illustrated in \cref{fig1}. So, checking whether the distribution anticoncentrates or not cannot be done by twirling.

\begin{figure}[H]
 \begin{centering}
    \includegraphics[width=0.25\textwidth]{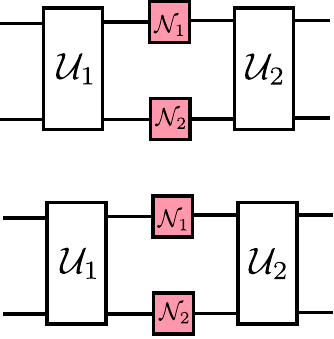}
    \caption{Two copies of a portion of the circuit that prepares $\tilde{\rho}$ is shown.
    }
    \label{fig1}
     \end{centering}
\end{figure} 
\noindent In \cref{fig1}, if we were to twirl $\mathcal{N}_1$ in a similar way as \cref{cross-entropy}, we would need to evaluate
\begin{equation}
\underset{U}{\mathbb{E}}\big[U^{\dagger} \otimes U^{\dagger} ~\mathcal{N}_1^{\otimes 2} (U \otimes U ~\rho \otimes \rho ~U^{\dagger} \otimes U^{\dagger})~U \otimes U \big],
\end{equation}
\noindent where $U$ is a single qubit Haar random unitary, which no longer simplifes to the depolarizing channel.

\section{Computation of the second moment of a special observable for amplitude damped random quantum circuits}
\label{appendixA}
We state the following Lemma. 
\begin{lemma}
\label{lemma13}
Consider the single qubit channel $\mathcal{N}$ applied $d$ times to $|0 \rangle \langle 0|$. Denote the resultant channel by $\mathcal{N}^d \coloneqq \underbrace{\mathcal{N}\circ\mathcal{N}\circ\cdots \circ\mathcal{N}}_{d~\text{times}}$. 
Suppose that 
\begin{equation}
\left\langle\ketbra{0}{0},\mathcal{N}^d(\ketbra{0}{0})\right\rangle = \kappa + \tau\lambda^d
\end{equation}
with some $\kappa \geq \tfrac{1}{2}$, $\tau,\lambda\geq 0$. 
Then,
\begin{equation}
\underset{\mathcal{B}}{\mathbb{E}}\left[\langle Z_i \rangle^2 \right] \geq \dfrac{4\left(\kappa - \tfrac{1}{2} + \tau\lambda^d\right)^2}{30^d},
\end{equation}
where $\mathcal{B}$ is an ensemble of noisy random quantum circuits of depth $d$ where the noise is modelled by $\mathcal{N}$, and $\langle Z_i \rangle \coloneqq 2 p_i - 1$ with $p_i$ being the marginal probability of getting outcome $i$ for a single qubit.
\end{lemma}
\noindent We may apply the above lemma to the noise models we consider. 
\begin{enumerate}
    \item When $\mathcal{N} = \mathcal{N}^{(\textup{amp})}_q\circ\mathcal{N}^{(\textup{dep})}_p$, 
   a simple calculation leads to  
   \begin{align}
   \kappa &= \frac{q + \tfrac{p}{2}(1-q)}{1 - (1-p)(1-q)} \\
   \tau &= 1 - \frac{q + \tfrac{p}{2}(1-q)}{1 - (1-p)(1-q)}\\
    \lambda &= (1-p)(1-q).    
   \end{align} 
   It is obvious that $\lambda \geq 0$. 
   To show $\kappa \geq \tfrac{1}{2}$ and $\tau \geq 0$, we show that 
   \begin{equation}
       \frac{1}{2} \leq \frac{q + \tfrac{p}{2}(1-q)}{1 - (1-p)(1-q)} \leq 1. 
   \end{equation}
   The left inequality can be shown as 
   \begin{align}
           \frac{q + \tfrac{p}{2}(1-q)}{1 - (1-p)(1-q)} 
           &= \frac{q + \tfrac{p}{2}(1-q)}{q + p(1-q)} \\ 
           &\geq \frac{\tfrac{q}{2} + \tfrac{p}{2}(1-q)}{q + p(1-q)} \\ 
           &= \frac{1}{2}. 
    \end{align}
   Similarly, the right inequality can be shown as  
   \begin{align}
           \frac{q + \tfrac{p}{2}(1-q)}{1 - (1-p)(1-q)} 
           &= \frac{q + \tfrac{p}{2}(1-q)}{q + p(1-q)} \\ 
           &\leq \frac{q + p(1-q)}{q + p(1-q)} \\ 
           &= 1. 
    \end{align}
   
   \item When $\mathcal{N} = \mathcal{N}^{(\textup{dep})}_p\circ\mathcal{N}^{(\textup{amp})}_q$, 
   we have 
   \begin{align}
   \kappa &= \dfrac{\tfrac{p}{2} + (1-p)q}{1 - (1-p)(1-q)} \\
   \tau &= 1 - \dfrac{\tfrac{p}{2} + (1-p)q}{1 - (1-p)(1-q)}\\
    \lambda &= (1-p)(1-q).    
   \end{align} 
\noindent   With a similar argument, we have 
   \begin{equation}
       \frac{1}{2} \leq \dfrac{\tfrac{p}{2} + (1-p)q}{1 - (1-p)(1-q)} \leq 1, 
   \end{equation}
   so $\kappa \geq \tfrac{1}{2}$ and $\tau,\lambda\geq 0$. 
\end{enumerate}

\noindent This lemma implies that if a given noise channel $\mathcal{N}$ satisfies 
\begin{equation}
    \kappa -\frac{1}{2} + \tau\lambda^d > 0, 
\end{equation}
then there exists positive numbers $a,b >0$ such that 
\begin{equation}
    \underset{\mathcal{B}}{\mathbb{E}}\left[\langle Z_i \rangle^2 \right] \geq b e^{-an}. 
\end{equation}
\begin{remark}
A special case of this lemma appears in \cite{Deshpande_2022}, where $\kappa = \frac{1}{2}, \tau = \frac{1}{2}$, and $\lambda = (1 - p)$, where $p$ is the strength of the depolarizing channel.
\end{remark}
\section{Discussion on conditional probabilities}
\label{conditional}
In this section, we will evaluate the first moment of conditional probability. 
\begin{lemma}
Let $\mathcal{B}$ be an ensemble of noisy random quantum circuits with noise channel $\mathcal{N}$. 
Let $i \in \{1,2,\ldots,n\}$, and let $J_i$ be a subset of $\{1,2,\ldots,n\}$ that does not contain $i$. 
Then, 
\begin{equation}
\underset{\mathcal{B}}{\mathbb{E}}\left[\Pr\big(X_i = x_i | \{X_j = x_j\}_{j \in J_i}\big)\right] = \dfrac{1}{2}\left[\bra{x_i}\mathcal{N}\left(I_2\right)\ket{x_i}\right]. 
\end{equation}
\end{lemma}
\begin{proof}
We only give proof for 
\begin{equation}
\underset{\mathcal{B}}{\mathbb{E}}\left[\Pr\big(X_i = x_i | \{X_j = x_j\}_{j \in \{1,2,\ldots,i-1\}}\big)\right] = \dfrac{1}{2}\left[\bra{x_i}\mathcal{N}\left(I_2\right)\ket{x_i}\right]. 
\end{equation}
The other cases follow similarly. 

Let $\mathcal{C}$ be the given quantum circuit. 
We may assume $\mathcal{C}$ can be written in the following form: 
\begin{equation}
    \mathcal{C} = \mathcal{N}^{\otimes n}\circ\left(\bigotimes_{i=1}^n \mathcal{U}_i\right)\circ\tilde{\mathcal{C}}, 
\end{equation}
where $\mathcal{N}$ is the noise channel, $\mathcal{U}_i$ is a single-qubit Haar random unitary channel, and $\tilde{\mathcal{C}}$ is the circuit without the last layer of noise. 
\clearpage
\noindent Then, the expectation over $\mathcal{B}$ is decomposed as 
\begin{equation}
    \underset{\mathcal{B}}{\mathbb{E}} =  \underset{\substack{\mathcal{U}_1,\ldots,\mathcal{U}_n \\ \sim \textup{Haar}}}{\mathbb{E}}\underset{\mathcal{B}'}{\mathbb{E}}, 
\end{equation}
where $\textup{Haar}$ represents the set of single-qubit Haar unitaries and $\mathcal{B}'$ is the set of noisy random quantum circuits without the last layer of noise. 
Then, 
\begin{align}
        &\Pr\big(X_i = x_i | \{X_j = x_j\}_{j \in \{1,2,\ldots,i-1\}}\big)\\
        &= \dfrac{\Pr\big(X_1 = x_1, \ldots, X_i = x_i\big)}{\Pr\big(X_1 = x_1, \ldots, X_{i-1} = x_{i-1}\big)\big)} \\ 
        &= \dfrac{\bra{x_1\cdots x_{i-1}x_i}\mathcal{N}^{\otimes i}\circ\left(\bigotimes_{j=1}^i \mathcal{U}_j\right)\left(\rho^{(i)}_{\tilde{\mathcal{C}}}\right)\ket{x_1\cdots x_{i-1}x_i}}{\sum_{y\in\{0,1\}}\bra{x_1\cdots x_{i-1}y}\mathcal{N}^{\otimes i}\circ\left(\bigotimes_{j=1}^{i} \mathcal{U}_j\right)\left(\rho^{(i)}_{\tilde{\mathcal{C}}}\right)\ket{x_1\cdots x_{i-1}y}}, 
\end{align}
where 
\begin{equation}
    \rho^{(i)}_{\tilde{\mathcal{C}}} = \Tr_{i+1\cdots n}\left[\left(\mathcal{I}^{\otimes i} \otimes \mathcal{N}^{\otimes (n-i)}\right)\circ \left(\mathcal{I}^{\otimes i} \otimes \bigotimes_{j=i+1}^n \mathcal{U}_j\right)\circ\tilde{\mathcal{C}}\left(\ketbra{0^n}{0^n}\right)\right].
\end{equation}
Now, let us write 
\begin{equation}
\begin{aligned}
     &\left(\mathcal{N}^{\otimes (i-1)}\otimes \mathcal{I}\right)\circ\left(\bigotimes_{j=1}^{i-1} \mathcal{U}_j \otimes \mathcal{I}\right)\left(\rho^{(i)}_{\tilde{\mathcal{C}}}\right) \\ 
     &= \sum_{y_1\cdots y_{i-1}, z_1\cdots z_{i-1} \in \{0,1\}^{i-1}} \eta^{(y_1\cdots y_{i-1})}_{(z_1\cdots z_{i-1})}\ketbra{y_1\cdots y_{i-1}}{z_1\cdots z_{i-1}}\otimes \sigma^{(y_1\cdots y_{i-1})}_{(z_1\cdots z_{i-1})},
\end{aligned} 
\end{equation}
where $\eta^{(y_1\cdots y_{i-1})}_{(z_1\cdots z_{i-1})}$ are appropriately defined coefficients so that $\sigma^{(y_1\cdots y_{i-1})}_{(z_1\cdots z_{i-1})}$ will be density operators when $y_1\cdots y_{i-1} = z_1\cdots z_{i-1}$. 
Then, we may write 
\begin{align}
        &\Pr\big(X_i = x_i | \{X_j = x_j\}_{j \in \{1,2,\ldots,i-1\}}\big)\\ &= \dfrac{\eta^{(x_1\cdots x_{i-1})}_{(x_1\cdots x_{i-1})}\bra{x_i}\left(\mathcal{N}\circ\mathcal{U}_i\right)\left(\sigma^{(x_1\cdots x_{i-1})}_{(x_1\cdots x_{i-1})}\right)\ket{x_i}}{\eta^{(x_1\cdots x_{i-1})}_{(x_1\cdots x_{i-1})}\left(\bra{0}\left(\mathcal{N}\circ\mathcal{U}_i\right)\left(\sigma^{(x_1\cdots x_{i-1})}_{(x_1\cdots x_{i-1})}\right)\ket{0} + \bra{1}\left(\mathcal{N}\circ\mathcal{U}_i\right)\left(\sigma^{(x_1\cdots x_{i-1})}_{(x_1\cdots x_{i-1})}\right)\ket{1}\right)} \\ 
        &= \dfrac{\eta^{(x_1\cdots x_{i-1})}_{(x_1\cdots x_{i-1})}\bra{x_i}\left(\mathcal{N}\circ\mathcal{U}_i\right)\left(\sigma^{(x_1\cdots x_{i-1})}_{(x_1\cdots x_{i-1})}\right)\ket{x_i}}{\eta^{(x_1\cdots x_{i-1})}_{(x_1\cdots x_{i-1})}\Tr\left[\left(\mathcal{N}\circ\mathcal{U}_i\right)\left(\sigma^{(x_1\cdots x_{i-1})}_{(x_1\cdots x_{i-1})}\right)\right]} \\ 
        &= \dfrac{\eta^{(x_1\cdots x_{i-1})}_{(x_1\cdots x_{i-1})}\bra{x_i}\left(\mathcal{N}\circ\mathcal{U}_i\right)\left(\sigma^{(x_1\cdots x_{i-1})}_{(x_1\cdots x_{i-1})}\right)\ket{x_i}}{\eta^{(x_1\cdots x_{i-1})}_{(x_1\cdots x_{i-1})}} \\ 
        &= \bra{x_i}\left(\mathcal{N}\circ\mathcal{U}_i\right)\left(\sigma^{(x_1\cdots x_{i-1})}_{(x_1\cdots x_{i-1})}\right)\ket{x_i}. 
\end{align}
Here, we used the fact that $\sigma^{(x_1\cdots x_{i-1})}_{(x_1\cdots x_{i-1})}$ is a single-qubit density operator. 
Therefore, 
\begin{align}
        &\underset{\mathcal{B}}{\mathbb{E}}\left[\Pr\big(X_i = x_i | \{X_j = x_j\}_{j \in \{1,2,\ldots,i-1\}}\big)\right] \\ 
        &= \underset{\substack{\mathcal{U}_1,\ldots,\mathcal{U}_n \\ \sim \textup{Haar}}}{\mathbb{E}}\underset{\mathcal{B}'}{\mathbb{E}} \left[\Pr\big(X_i = x_i | \{X_j = x_j\}_{j \in \{1,2,\ldots,i-1\}}\big)\right] \\ 
        &= \underset{\substack{\mathcal{U}_1,\ldots,\mathcal{U}_n \\ \sim \textup{Haar}}}{\mathbb{E}}\underset{\mathcal{B}'}{\mathbb{E}} \left[\bra{x_i}\left(\mathcal{N}\circ\mathcal{U}_i\right)\left(\sigma^{(x_1\cdots x_{i-1})}_{(x_1\cdots x_{i-1})}\right)\ket{x_i}\right] \\ 
        &= \underset{\substack{\mathcal{U}_1,\ldots,\mathcal{U}_{i-1} \\ \sim \textup{Haar}}}{\mathbb{E}}\underset{\substack{\mathcal{U}_{i+1},\ldots,\mathcal{U}_{n} \\ \sim \textup{Haar}}}{\mathbb{E}}\underset{\mathcal{B}'}{\mathbb{E}} \left[\bra{x_i}\underset{\mathcal{U}_{i} \sim \textup{Haar}}{\mathbb{E}}\left[\left(\mathcal{N}\circ\mathcal{U}_i\right)\left(\sigma^{(x_1\cdots x_{i-1})}_{(x_1\cdots x_{i-1})}\right)\right]\ket{x_i}\right] \\ 
        &= \underset{\substack{\mathcal{U}_1,\ldots,\mathcal{U}_{i-1} \\ \sim \textup{Haar}}}{\mathbb{E}}\underset{\substack{\mathcal{U}_{i+1},\ldots,\mathcal{U}_{n} \\ \sim \textup{Haar}}}{\mathbb{E}}\underset{\mathcal{B}'}{\mathbb{E}} \left[\bra{x_i}\mathcal{N}\left(\dfrac{I_2}{2}\right)\ket{x_i}\right] \\ 
        &= \frac{1}{2} \left[\bra{x_i}\mathcal{N}\left(I_2\right)\ket{x_i}\right], 
\end{align}
which we aimed to show.  
\end{proof}

\section{Discussion on the effect of noise in noisy random circuits}
\label{sec:Werner_twirl}
\noindent In this section, we show that for any given noise $\mathcal{N}$, we may express the action of $\tilde{M}_{U_1, \mathsf{N}}$ on $I_4$ and $S$ as  
\begin{align}
    \tilde{M}_{U_1, \mathsf{N}}(I_4) &= (1 - a)I_4 + 2a S, \\ 
    \tilde{M}_{U_1, \mathsf{N}}(S) &= bI + (1 - 2b) S
\end{align}
using some $a$ and $b$. The action of $M_{U_1}$ can be given as (see Example 7.25 of \cite{Watrous2018}) 
\begin{equation}
    M_{U_1}(X) = \dfrac{\langle X,\Pi_{\mathrm{sym}}\rangle}{3}\Pi_{\mathrm{sym}} + \langle X,\Pi_{\mathrm{antisym}}\rangle\Pi_{\mathrm{antisym}}, 
\end{equation}
where 
\begin{align}
    \Pi_{\mathrm{sym}} &\coloneqq \frac{I_4 + S}{2} \\ 
    \Pi_{\mathrm{antisym}} &\coloneqq \frac{I_4 - S}{2}
\end{align}
are the projection operators onto the symmetric and anti-symmetric subspaces, respectively. 
Suppose that $X$ is given as 
\begin{equation}
    X = 
    \sum_{i,j,k,l = 0,1} x_{ijkl} \ketbra{ij}{kl}. 
\end{equation}
A simple calculation leads to 
\begin{align}
    \langle X,\Pi_{\mathrm{sym}}\rangle 
    &= x_{0000} + x_{1111} + \frac{x_{0101} + x_{1010} + x_{0110} + x_{1001}}{2}, \\ 
    \langle X,\Pi_{\mathrm{antisym}}\rangle 
    &= \frac{x_{0101} + x_{1010} - x_{0110} - x_{1001}}{2}. 
\end{align}
Thus, 
\begin{equation}
    \label{eq:werner_twirl_calc}
    \begin{aligned}
    M_{U_1}(X) 
    &= \left(\frac{x_{0000}}{6} + \frac{x_{1111}}{6} + \frac{x_{0101} + x_{1010}}{3} - \frac{x_{0110} + x_{1001}}{6}\right)I_4 \\ 
    &+ \left(\frac{x_{0000}}{6} + \frac{x_{1111}}{6} - \frac{x_{0101} + x_{1010}}{6} + \frac{x_{0110} + x_{1001}}{3}\right)S. 
    \end{aligned}
\end{equation}
When 
\begin{equation}
    X = \mathsf{N}\circ M_{U_1}\left[I_4\right] = \mathsf{N}\left[I_4\right],  
\end{equation}
since $\Tr[X] = x_{0000} + x_{0101} + x_{1010} + x_{1111} = \Tr[I_4] = 4$, by setting 
\begin{equation}
    a = \frac{x_{0000}}{12} + \frac{x_{1111}}{12} - \frac{x_{0101} + x_{1010}}{12} + \frac{x_{0110} + x_{1001}}{6}, 
\end{equation}
we have 
\begin{equation}
    \tilde{M}_{U_1,\mathsf{N}}\left[I_4\right] = (1-a)I_4 + 2a S. 
\end{equation}
Similarly, 
when \begin{equation}
    X = \mathsf{N}\circ M_{U_1}\left[S\right] = \mathsf{N}\left[S\right],  
\end{equation}
since $\Tr[X] = x_{0000} + x_{0101} + x_{1010} + x_{1111} = \Tr[S] = 2$, by setting 
\begin{equation}
    b = \frac{x_{0000}}{6} + \frac{x_{1111}}{6} + \frac{x_{0101} + x_{1010}}{3} - \frac{x_{0110} + x_{1001}}{6}, 
\end{equation}
we have 
\begin{equation}
\tilde{M}_{U_1,\mathsf{N}}\left[S\right] = bI + (1-2b) S. 
\end{equation}

\noindent We see the cases of $\mathcal{N} = \mathcal{N}_q^{(\mathrm{amp}} \circ \mathcal{N}_p^{(\mathrm{dep}}$ and $\mathcal{N} = \mathcal{N}_p^{(\mathrm{dep}} \circ \mathcal{N}_p^{(\mathrm{amp}}$ as illustrative examples.  
When $\mathcal{N} = \mathcal{N}_q^{(\mathrm{amp}} \circ \mathcal{N}_p^{(\mathrm{dep}}$, we have 
\begin{align}
    \mathsf{N}[I_4] &= 
    \left(
    \begin{array}{cccc}
        (1+q)^2 & 0 & 0 & 0 \\
        0 & 1 - q^2 & 0 & 0 \\ 
        0 & 0 & 1 - q^2 & 0 \\ 
        0 & 0 & 0 & (1-q)^2
    \end{array}
    \right)\\ 
    \mathsf{N}[S] &= 
    \tiny
    \left(
    \begin{array}{cccc}
        (1-q)^2\left(\frac{p^2}{2} - p + 1\right) + 2q & 0 & 0 & 0 \\
        0 & \frac{1-q^2 - (1-q)^2(1-p)^2}{2}& (1-q)(1-p)^2 & 0 \\ 
        0 & (1-q)(1-p)^2 & \frac{1-q^2 - (1-q)^2(1-p)^2}{2} & 0 \\ 
        0 & 0 & 0 & (1-q)^2\left(\frac{p^2}{2} - p + 1\right)
    \end{array}
    \right). 
\end{align}
Therefore, we have
\begin{align}
    a = \frac{q^2}{3},\quad\quad b =  \frac{1}{2} - \frac{q^2}{6} -\frac{1}{6}(1-p)^2(1-q)(3-q). 
\end{align}
Similarly, when $\mathcal{N} = \mathcal{N}_p^{(\mathrm{dep}} \circ \mathcal{N}_p^{(\mathrm{amp}}$, 
\begin{align}
    \mathsf{N}[I_4] &= 
    \left(
    \begin{array}{cccc}
        (1+(1-p)q)^2 & 0 & 0 & 0 \\
        0 & 1 - (1-p)^2q^2 & 0 & 0 \\ 
        0 & 0 & 1 - (1-p)^2q^2 & 0 \\ 
        0 & 0 & 0 & (1-(1-p)q)^2
    \end{array}
    \right) \\ 
    \mathsf{N}[S]  &= 
    \tiny
    \left(
    \begin{array}{cccc}
        q^2(1-p)^2+1 + \frac{p^2}{2} - p(1 + (1-p)q) & 0 & 0 & 0 \\
        0 & q(1-q)(1-p)^2 - \frac{p^2}{2} + p& (1-q)(1-p)^2 & 0 \\ 
        0 & (1-q)(1-p)^2 &q(1-q)(1-p)^2 - \frac{p^2}{2} + p & 0 \\ 
        0 & 0 & 0 & (q(1-p) - 1)^2 + \frac{p^2}{2} - p(1 + (1-p)q)
    \end{array}
    \right)
    \tiny. 
\end{align}
Therefore, 
\begin{align}
    a = \frac{q^2(1-p)^2}{3},\quad\quad b =  \frac{1}{2} - \frac{q^2(1-p)^2}{6} -\frac{1}{6}(1-p)^2(1-q)(3-q). 
\end{align}

\section{Proof of Theorem~\ref{second moment probabilities}}
\label{appendix:proof_of_CR}
In the proof, we use the following fact, which we will prove in the next section.
\begin{lemma}
\label{lem:modified-circuit}
Consider an ensemble $\mathcal{B}$ of noisy random quantum circuits with noise channel $\mathcal{N}$ satisfying Eqs.~\eqref{eq:noise_1} and \eqref{eq:noise_2}. 
Let $\mathcal{B}'$ denote the ensemble of circuits obtained by removing the last layer of noise from the noisy random circuits in $\mathcal{B}$. 
Consider another ensemble $\tilde{\mathcal{B}}'$ of circuits that can be obtained by replacing each 2-qubit Haar random gate $U_{2}$ in the circuits of $\mathcal{B}'$ as 
\begin{equation}
    U_2 \to \left(U_1\otimes U'_1\right), 
\end{equation}
where $U_1,U'_1$ are independent single-qubit Haar random gates. 
In this setup, if $0 \leq 1-a-2b \leq 1$, then $\mathbb{E}_{\mathcal{B}'}[p_{x}^2] \leq \mathbb{E}_{\tilde{\mathcal{B}}'}[p_{x}^2]$ for any $x\in\{0,1\}^n$.
\end{lemma}
\begin{remark}
    In the statement of Lemma~\ref{lem:modified-circuit}, the circuits in $\tilde{\mathcal{B}}'$ are composed solely of single-qubit Haar random unitary gates and noise channel $\mathcal{N}$, without the last layer of noise. 
\end{remark}
\noindent Now, we continue on with our proof of Theorem~ \ref{second moment probabilities}. 
In the noise models we consider, we have 
\begin{align}
    a &= \Bigg\{ 
    \begin{array}{ll}
        \frac{q^2}{3}  & \mathcal{N} = 
        \mathcal{N}^{(\textup{amp})}_q \circ \mathcal{N}^{(\textup{dep})}_p, \\
         \frac{q^2(1-p)^2}{3} &  \mathcal{N} = \mathcal{N}^{(\textup{dep})}_p \circ \mathcal{N}^{(\textup{amp})}_q,
    \end{array}\\
    b&= \Bigg\{ 
    \begin{array}{ll}
       \frac{1}{2} - \frac{q^2}{6} -\frac{1}{6}(1-p)^2(1-q)(3-q)  & \mathcal{N} = 
        \mathcal{N}^{(\textup{amp})}_q \circ \mathcal{N}^{(\textup{dep})}_p, \\
        \frac{1}{2} - \frac{q^2(1-p)^2}{6} -\frac{1}{6}(1-p)^2(1-q)(3-q) &  \mathcal{N} = \mathcal{N}^{(\textup{dep})}_p \circ \mathcal{N}^{(\textup{amp})}_q.
    \end{array}
\end{align}
We can easily verify that
\begin{equation}
    1- a -2b = \left(1-p\right)^2\left(1-q\right)\left(1 - \frac{q}{3}\right)
\end{equation}
in both cases. 
Hence, $0 \leq 1-a-2b \leq 1$. 
Now, we rewrite the second-moment probability by using the weighted trajectories of the strings in $\{I_4,S\}^n$ (see also the proof of Lemma~6 in \cite{Deshpande_2022}). 
Suppose that the circuits in $\mathcal{B}'$ contain $s$ Haar random gates. 
A trajectory $\gamma \coloneqq (\gamma^{1}\gamma^{2}\cdots \gamma^{s}\gamma^{s+1}) \in \{I_4,S\}^{n\times (s+1)}$ is a sequence of $n$-bit string $\gamma^{i} \in \{I_4,S\}^n$. 
For $i = 1,2\cdots,s$, $\gamma^{i} \in \{I_4,S\}^n$ represents the $n$-bit string right before the $s$-th Haar random gate in the trajectory $\gamma$. 
In particular, $\gamma^{1}$ is the initial string of the trajectory $\gamma$. 
$\gamma^{s+1} \in \{I_4,S\}^n$ represents the final bit string of the trajectory $\gamma$ at the end of the circuit. 
For each trajectory $\gamma \in\{I_4,S\}^{n\times (s+1)}$, we define the weight $\mathrm{wt}_{\mathcal{B}'}(\gamma)$ as 
\begin{equation}
    \mathrm{wt}_{\mathcal{B}'}(\gamma) \coloneqq c_1c_2c_3\cdots c_s, 
\end{equation}
where for $i= 1,2,\ldots,s$, 
\begin{equation}
    c_i \coloneqq \Bigg\{
    \begin{array}{ll}
        \mathrm{Coefficient\,\,of\,\,transformation\,\,}\gamma^{i}\to\gamma^{i+1}, & \gamma^{i}\to\gamma^{i+1}~~\mathrm{is\,\,possible,} \\
        0, & \gamma^{i}\to\gamma^{i+1}~~\mathrm{is\,\,impossible.}
    \end{array}
\end{equation}
\clearpage
\noindent Thus, the second-moment probability with respect to $\mathcal{B}'$ can be expressed as   
\begin{equation}
    \mathbb{E}_{\mathcal{B}'}[p_{x}^2] =\sum_{\gamma_1,\gamma_2,\ldots, \gamma_n \in \{I_4,S\}}\left[\left(\sum_{\substack{\gamma \in \{I_4,S\}^{n\times (s+1)} \\ \gamma^{s+1} = \gamma_1\gamma_2\cdots \gamma_n}} \mathrm{wt}_{\mathrm{B}'}(\gamma)\right)\bra{x}\bra{x}(\gamma_1\gamma_2\cdots \gamma_n)\ket{x}\ket{x}\right],
\end{equation}
where in the right-hand side, we divide the cases based on the final bit string $\gamma^{s+1}=\gamma_1\gamma_2\cdots \gamma_n$. 
We may further rewrite as 
\begin{equation}
\label{eq:second_moment_no_noise}
\mathbb{E}_{\mathcal{B}'}[p_{x}^2] =\sum_{\gamma_1,\gamma_2,\ldots, \gamma_n \in \{I_4,S\}}\left[\left(\sum_{\substack{\gamma \in \{I_4,S\}^{n\times (s+1)} \\ \gamma^{s+1} = \gamma_1\gamma_2\cdots \gamma_n}} \mathrm{wt}_{\mathrm{B}'}(\gamma)\right)\left(\prod_{j=1}^n\bra{x_j}\bra{x_j}\gamma_j\ket{x_j}\ket{x_j}\right)\right]. 
\end{equation}
For $\gamma_j \in \{I_4,S\}$, $\bra{0}\bra{0}\gamma_i\ket{0}\ket{0} = \bra{1}\bra{1}\gamma_i\ket{1}\ket{1} = 1$. 
Therefore, 
\begin{equation}
\label{eq:second_moment_no_noise_simple}
\mathbb{E}_{\mathcal{B}'}[p_{x}^2] =\sum_{\gamma_1,\gamma_2,\ldots, \gamma_n \in \{I_4,S\}}\left[\left(\sum_{\substack{\gamma \in \{I_4,S\}^{n\times (s+1)} \\ \gamma^{s+1} = \gamma_1\gamma_2\cdots \gamma_n}} \mathrm{wt}_{\mathrm{B}'}(\gamma)\right)\right]. 
\end{equation}
On the other hand, the second-moment probability with respect to the original ensemble $\mathcal{B}$ can be obtained by considering the last layer of noise. 
Thus, in this case, instead of Eq.~\eqref{eq:second_moment_no_noise}, we have 
\begin{equation}
\mathbb{E}_{\mathcal{B}}[p_{x}^2] =\sum_{\gamma_1,\gamma_2,\ldots, \gamma_n \in \{I_4,S\}}\left[\left(\sum_{\substack{\gamma \in \{I_4,S\}^{n\times (s+1)} \\ \gamma^{s+1} = \gamma_1\gamma_2\cdots \gamma_n}} \mathrm{wt}_{\mathrm{B}'}(\gamma)\right)\left(\prod_{j=1}^n\bra{x_j}\bra{x_j}\mathsf{N}(\gamma_j)\ket{x_j}\ket{x_j}\right)\right]. 
\end{equation}
Here, observe that 
\begin{equation}
    \bra{x_j}\bra{x_j}\mathsf{N}(\gamma_j)\ket{x_j}\ket{x_j} \leq \max\{\bra{x_j}\bra{x_j}\mathsf{N}(I_4)\ket{x_j}\ket{x_j},\bra{x_j}\bra{x_j}\mathsf{N}(S)\ket{x_j}\ket{x_j}\} \eqqcolon e_j
\end{equation}
for each $j$.
Hence, 
\begin{equation}
    \mathbb{E}_{\mathcal{B}}[p_{x}^2] \leq \left[\sum_{\gamma_1,\gamma_2,\ldots, \gamma_n \in \{I_4,S\}}\left(\sum_{\substack{\gamma \in \{I_4,S\}^{n\times (s+1)} \\ \gamma^{s+1} = \gamma_1\gamma_2\cdots \gamma_n}} \mathrm{wt}_{\mathrm{B}'}(\gamma)\right)\right]\left(e_1e_2\cdots e_n\right). 
\end{equation}
Using Eq.~\eqref{eq:second_moment_no_noise_simple}, 
\begin{equation}
    \mathbb{E}_{\mathcal{B}}[p_{x}^2] \leq \mathbb{E}_{\mathcal{B}'}[p_{x}^2]\left(e_1e_2\cdots e_n\right).
\end{equation}
By Lemma~\ref{lem:modified-circuit}, since $\mathbb{E}_{\mathcal{B}'}[p_{x}^2] \leq \mathbb{E}_{\tilde{\mathcal{B}}'}[p_{x}^2]$, 
\begin{equation}
    \label{eq:second_moment_relation1}
    \mathbb{E}_{\mathcal{B}}[p_{x}^2] \leq \mathbb{E}_{\tilde{\mathcal{B}}'}[p_{x}^2]\left(e_1e_2\cdots e_n\right). 
\end{equation}
Now, by the construction of $\tilde{\mathcal{B}}'$ and Proposition~\ref{first thm}, 
\begin{equation}   
    \begin{aligned}
\mathbb{E}_{\tilde{\mathcal{B}}'}[p_{x}^2] 
&= \left(\frac{3}{10} - \left[\dfrac{-2a+b}{10(a+2b)} + (1-a-2b)^{d-1} \left(-\dfrac{1}{30} - \dfrac{-2a+b}{10(a+2b)}\right)\right]\right)^n. 
    \end{aligned}
\end{equation}
Here, we used the fact that $\bra{0}\bra{0}(2I + S)\ket{0}\ket{0} = \bra{1}\bra{1}(2I + S)\ket{1}\ket{1} = 3$ and $\bra{0}\bra{0}(I_4 - 2S)\ket{0}\ket{0} = \bra{1}\bra{1}(I_4 - 2S)\ket{1}\ket{1} = -1$. 
By definition of $a$ and $b$, 
letting 
\begin{align}
    \label{eq:c}
    c &\coloneqq 1 - (1-p)^2(1-q)\left(1 - \frac{q}{3}\right), \\
    \label{eq:r}
    r &\coloneqq \Bigg\{ 
    \begin{array}{ll}
        q, & \mathcal{N} = 
        \mathcal{N}^{(\textup{amp})}_q \circ \mathcal{N}^{(\textup{dep})}_p, \\
        q(1-p), & \mathcal{N} = \mathcal{N}^{(\textup{dep})}_p \circ \mathcal{N}^{(\textup{amp})}_q,
    \end{array}
\end{align}
we have 
\begin{equation}
    \mathbb{E}_{\tilde{\mathcal{B}}'}[p_{x}^2] = \left(\left(\frac{1}{4} + \frac{r^2}{12c}\right) +  
    (1 - c)^{d-1} \left(\dfrac{1}{12} - \dfrac{r^2}{12c}\right) \right)^n. 
\end{equation}
By taking 
\begin{align}
    \mu &= \frac{1}{4} + \frac{r^2}{12c}, \\ 
    \nu &= \dfrac{1}{12} - \dfrac{r^2}{12c}, 
\end{align}
\begin{equation}
    \mathbb{E}_{\tilde{\mathcal{B}}'}[p_{x}^2] = \left(\mu +  
    (1 - c)^{d-1} \nu \right)^n. 
\end{equation}
Here, we will check $\mu,\nu \geq 0$.
It trivially follows that $\mu \geq 0$ from the definition. 
Next, we evaluate $\nu$. 
Since $r \leq q$ by definition, 
\begin{align}
    \nu 
    &= \dfrac{1}{12} - \dfrac{r^2}{12c}\\
    &\geq \dfrac{1}{12} - \dfrac{q^2}{12c} \\
    &= \dfrac{c - q^2}{12c}.  
\end{align}
By substituting \cref{eq:c} and \cref{eq:r}, 
\begin{align}
    \dfrac{c - q^2}{12c}
    &= \dfrac{1 - (1-p)^2(1-q)\left(1 - \frac{q}{3}\right) - q^2}{12c}. 
\end{align}
Since $(1-p)^2 \leq 1$, 
\begin{align}
    \dfrac{1 - (1-p)^2(1-q)\left(1 - \frac{q}{3}\right) - q^2}{12c}
    &\geq \dfrac{1 - (1-q)\left(1 - \frac{q}{3}\right) - q^2}{12c} \\ 
    &= \dfrac{\frac{4q}{3}(1-q)}{12c} \\ 
    &= \dfrac{q(1-q)}{9c}.  
\end{align}    
Since $q\geq 0$ and $1-p \geq 0$, we have 
\begin{align}
    \nu \geq 0. 
\end{align}
Since $1 + x \leq \mathrm{e}^{x}$ for all $x\in\mathbb{R}$, 
\begin{equation}
\label{eq:second_moment_rekation2}\mathbb{E}_{\tilde{\mathcal{B}}'}[p_{x}^2] \leq \mu^n \exp\left[n\dfrac{\nu}{\mu} \mathrm{e}^{-c(d-1)}\right]. 
\end{equation}
Using Eqs.~\eqref{eq:second_moment_relation1} and \eqref{eq:second_moment_rekation2}, 
\begin{equation}
\label{eq:second_moment_rekation3}
    \mathbb{E}_{\mathcal{B}}[p_{x}^2] \leq  \mu^n \exp\left[n\dfrac{\nu}{\mu} \mathrm{e}^{-c(d-1)}\right]\left(e_1e_2\cdots e_n\right). 
\end{equation}
Now, we evaluate $e_j$. 
\begin{equation}
    e_j \coloneqq \max\left\{\bra{x_j}\bra{x_j}\mathsf{N}(I_4)\ket{x_j}\ket{x_j},\bra{x_j}\bra{x_j}\mathsf{N}(S)\ket{x_j}\ket{x_j}\right\}  = \Bigg\{
    \begin{array}{cc}
        (1+r)^2, & x_j = 0, \\
        (1-r)^2, & x_j = 1.
    \end{array}
\end{equation}
That is, if $w_x \geq \tfrac{n}{2}$, 
\begin{equation}
\label{eq:second_moment_rekation4}
    e_1e_2\cdots e_n = (1+r)^{2(n-w_x)}(1-r)^{2w_x} \leq (1+r)^{n}(1-r)^{n} = (1-r^2)^{n} = \eta^n. 
\end{equation}
Combining Eqs.~\eqref{eq:second_moment_rekation3} and \eqref{eq:second_moment_rekation4}, 
\begin{equation}
    \mathbb{E}_{\mathcal{B}}[p_{x}^2] \leq  \mu^n\eta^n \exp\left[n\dfrac{\nu}{\mu} \mathrm{e}^{-c(d-1)}\right], 
\end{equation}
which completes the proof. 

\section{Proof of Lemma~\ref{lem:modified-circuit}}
Any circuit $\mathcal{C}$ in $\mathcal{B}'$ can be written in the following form. 

\begin{align*}
\begin{tikzcd}[transparent,row sep=0.7em]
& \gate[2]{U_2}& \gate[1]{\mathcal{N}} & \gate[1]{U_1} &\gate[1]{\mathcal{N}} & \gate[2]{U_2} & \qw\\
& \qw & \gate[1]{\mathcal{N}} &  \gate[2]{U_2} &\gate[1]{\mathcal{N}} &  \qw &\qw\\
& \gate[2]{U_2} & \gate[1]{\mathcal{N}} & \qw &\gate[1]{\mathcal{N}} & \gate[2]{U_2} & \qw\\
&\qw & \gate[1]{\mathcal{N}} & \gate[2]{U_2} & \gate[1]{\mathcal{N}} & \qw & \qw\\
& \gate[1]{U_1} & \gate[1]{\mathcal{N}} &  \qw & \gate[1]{\mathcal{N}} & \gate[1]{U_1}& \qw
\end{tikzcd}
\end{align*}
Let $s$ be the number of 2-qubit Haar random gate in given $\mathcal{C}\in \mathcal{B}'$. 
Let us introduce a numbering of the 2-qubit Haar random gates in $\mathcal{C}$: $U_2^{(k)}$ refers to the $k^{\text{th}}$ 2-qubit Haar random gate, where the counting starts from the top gate at the last layer. 
Rewiring the circuit $\mathcal{C}$, we have 
\begin{align}
\label{eq:circuit_0}
\begin{tikzcd}[transparent,row sep=0.7em]
& \gate[2]{U_2^{(s-1)}}& \gate[1]{U_1} & \gate[1]{\mathcal{N}}  & \gate[1]{U_1} & \qw & \gate[1]{U_1}& \gate[1]{\mathcal{N}}  & \gate[1]{U_1}& \gate[2]{U_2^{(1)}} & \qw \\
& \qw & \gate[1]{U_1}& \gate[1]{\mathcal{N}}  & \gate[1]{U_1}& \gate[2]{U_2^{(3)}} & \gate[1]{U_1}& \gate[1]{\mathcal{N}}  & \gate[1]{U_1}&\qw& \qw \\
& \gate[2]{U_2^{(s)}} & \gate[1]{U_1} & \gate[1]{\mathcal{N}}  & \gate[1]{U_1}& \qw & \gate[1]{U_1}& \gate[1]{\mathcal{N}}  & \gate[1]{U_1}& \gate[2]{U_2^{(2)}}& \qw \\
&\qw & \gate[1]{U_1} & \gate[1]{\mathcal{N}}  & \gate[1]{U_1}&\gate[2]{U_2^{(4)}} & \gate[1]{U_1}& \gate[1]{\mathcal{N}}  & \gate[1]{U_1}&\qw& \qw \\
& \qw & \gate[1]{U_1} & \gate[1]{\mathcal{N}}  & \gate[1]{U_1}& \qw & \gate[1]{U_1}& \gate[1]{\mathcal{N}}  & \gate[1]{U_1}& \qw  &\qw
\end{tikzcd}
\end{align}
\clearpage
\noindent Construct a sequence of circuits $\{\mathcal{C}_j: j=0,1,2,\ldots,s\}$ in the following recursive way. 
\begin{itemize}
    \item $\mathcal{C}_0 = \mathcal{C} \in \mathcal{B}'$.
    \item For $j = 1,2,\ldots,s$, $\mathcal{C}_j$ is a circuit obtained by replacing the leading 2-qubit Haar random gate in $\mathcal{C}_{j-1}$, \textit{i.e.,} $U_2^{(j)}$, by two parallel independent single-qubit Haar random gates. 
\end{itemize}
By construction, $\mathcal{C}_{s}$ consists only of single-qubit Haar random gates and noise channels, and $\mathcal{C}_{s} \in \tilde{\mathcal{B}}'$. 
Let $x \in \{0,1\}^n$ be any $n$-bit string. 
Here, we aim to show 
\begin{equation}
    \underset{\mathcal{C}_0 \sim \mathcal{B}'}{\mathbb{E}}[p_x^2] \leq \underset{\mathcal{C}_s \sim \tilde{\mathcal{B}}'}{\mathbb{E}}[p_x^2], 
\end{equation}
and to this end, it suffices to show 
\begin{equation}
    \underset{\mathcal{C}_j}{\mathbb{E}}[p_x^2] \leq \underset{\mathcal{C}_{j+1}}{\mathbb{E}}[p_x^2] 
\end{equation}
for all $j = 0,1,2,\ldots, s-1$. 
For this purpose, fix $j$, and let us look at $\mathcal{C}_j$ and $\mathcal{C}_{j+1}$. 
We will show 
\begin{equation}
    \underset{\mathcal{C}_j}{\mathbb{E}}[p_x^2] \leq \underset{\mathcal{C}_{j+1}}{\mathbb{E}}[p_x^2] 
\end{equation}
for this $j$. 
Suppose that $U_2^{(j)}$ acts on the $k^{\text{th}}$ and $l^{\text{th}}$ qubits. 
Recall the characterization shown in Eq.~\eqref{eq:circuit_0}.
\clearpage
\noindent Then, by absorbing the action of the single-qubit Haar random gates and noise channels acting on qubits other than $k^{\text{th}}$ or $l^{\text{th}}$ one to the preceding gates, we may write 
\begin{equation}
\mathcal{C}_j = 
\begin{tikzcd}[transparent,row sep=0.7em]
    & \gate[5]{\tilde{\mathcal{C}}_j}\slice{} & \qw & \qw &\qw\\
    & \qw & \qw & \qw &\qw\\
    & \qw& \gate[2]{U_2^{(j+1)}} & \gate[1]{\tilde{\mathcal{U}}_{U_1,\mathcal{N}}^{(m)}}&\qw \\
    &\qw & \qw & \gate[1]{\tilde{\mathcal{U}}_{U_1,\mathcal{N}}^{(m)}} &\qw\\
    & \qw & \qw & \qw &\qw
\end{tikzcd}
\end{equation}
and 
\begin{equation}
\mathcal{C}_{j+1} = 
\begin{tikzcd}[transparent,row sep=0.7em]
    & \gate[5]{\tilde{\mathcal{C}}_j}\slice{} & \qw & \qw &\qw\\
    & \qw & \qw & \qw &\qw\\
    & \qw& \gate[1]{U_1} & \gate[1]{\tilde{\mathcal{U}}_{U_1,\mathcal{N}}^{(m)}}&\qw \\
    &\qw & \gate[1]{U_1} & \gate[1]{\tilde{\mathcal{U}}_{U_1,\mathcal{N}}^{(m)}} &\qw\\
    & \qw & \qw & \qw &\qw
\end{tikzcd}
\end{equation}
where 
\begin{equation}
    \tilde{\mathcal{U}}_{U_1,\mathcal{N}}^{(m)} \coloneqq \underbrace{(\mathcal{U}_1 \circ \mathcal{N}\circ \mathcal{U}_1) \circ \cdots \circ (\mathcal{U}_1 \circ \mathcal{N}\circ \mathcal{U}_1) }_{m~\mathrm{times}}
\end{equation}
with some non-negative integer $m$.
Here, we consider $\{I_4,S\}^n$ bitstring representation.
Suppose that we have $\gamma \in \{I_4,S\}^n$ just before the red dashed line. 
Suppose that $\gamma$ will change as 
\begin{align}
    &\gamma \xrightarrow{\mathcal{C}_j:\mathrm{After\,\,red\,\,line}} \sum_{\gamma' \in \{I_4,S\}^n} c^{(\gamma')}_j \gamma' \\ 
    &\gamma \xrightarrow{\mathcal{C}_{j+1}:\mathrm{After\,\,red\,\,line}} \sum_{\gamma' \in \{I_4,S\}^n} c^{(\gamma')}_{j+1} \gamma'. 
\end{align}
Define 
\begin{align}
     a^{(\gamma)}_j &\coloneqq \sum_{\gamma' \in \{I_4,S\}^n} c^{(\gamma')}_j \\ 
     a^{(\gamma)}_{j+1} &\coloneqq \sum_{\gamma' \in \{I_4,S\}^n} c^{(\gamma')}_{j+1}.  
\end{align}
It suffices to show that 
\begin{equation}
    a^{(\gamma)}_j \leq a^{(\gamma)}_{j+1}
\end{equation}
for all $\gamma$. Observe that after the red line, the gates only act on the $k^{\text{th}}$ and the $l^{\text{th}}$ qubits. 
So, let us focus on these two qubits. 
\begin{enumerate}
    \item If $\gamma_{k,l} = II,SS$, since 
    \begin{align}
        &II \xrightarrow{U_2^{(j+1)}} II\\
        &II \xrightarrow{U_1\otimes U_1} II
    \end{align}
    and 
     \begin{align}
        &SS \xrightarrow{U_2^{(j+1)}} SS\\
        &SS \xrightarrow{U_1\otimes U_1} SS,
    \end{align}
    we have 
    \begin{equation}
    a^{(\gamma)}_j =  a^{(\gamma)}_{j+1}
    \end{equation}
    in this case. 

    \item If $\gamma_{k,l} = IS,SI$, without loss of generality, we may assume $\gamma_{k,l} = IS$. The proof for $\gamma_{k,l} = SI$ similarly follows. 
    We have 
    \begin{align}
        &IS \xrightarrow{U_2^{(j+1)}} \dfrac{2}{5}\left(II + SS\right)\\
        &IS \xrightarrow{U_1\otimes U_1} IS. 
    \end{align}
    By Lemma~\ref{lem:sequence_noise} (shown below this section),  
    \begin{align}
    &IS \xrightarrow{\mathcal{C}_j:\mathrm{After\,\,red\,\,line}} \frac{2}{5}\left((x_m I_4 + y_m S)(x_m I_4 + y_m S) + (z_m I_4 + w_m S)(z_m I_4 + w_m S)\right) \\ 
    &IS \xrightarrow{\mathcal{C}_{j+1}:\mathrm{After\,\,red\,\,line}} (x_m I_4 + y_m S)(z_m I_4 + w_m S), 
\end{align}
where 
\begin{align}
    x_m &= 1 - \dfrac{1 - (1-a-2b)^m}{1 + \tfrac{2b}{a}},\\ 
    y_m &= \dfrac{1 - (1-a-2b)^m}{\tfrac{1}{2} + \tfrac{b}{a}}, \\ 
    z_m &= \dfrac{1}{2} - \dfrac{\tfrac{1}{2} + \tfrac{b}{a}(1-a-2b)^m}{1 + \tfrac{2b}{a}}, \\ 
    w_m &=  \dfrac{\tfrac{1}{2} + \tfrac{b}{a}(1-a-2b)^m}{\tfrac{1}{2} + \tfrac{b}{a}}. 
\end{align}
Therefore, 
\begin{align}
    a^{(\gamma)}_j &= \dfrac{2}{5}\left(u_m^2 + v_m^2\right), \\ 
    a^{(\gamma)}_{j+1} &= u_mv_m, 
\end{align}
where 
\begin{align}
    u_m &= x_m + y_m = 1 + \dfrac{1 - (1-a-2b)^m}{1 + \tfrac{2b}{a}}, \\ 
    v_m &= z_m + w_m = \dfrac{1}{2} + \dfrac{\tfrac{1}{2} + \tfrac{b}{a}(1-a-2b)^m}{1 + \tfrac{2b}{a}}. 
\end{align}
Now, 
\begin{equation}
    \begin{aligned}
        &a^{(\gamma)}_j \leq a^{(\gamma)}_{j+1} \\ 
        &\Leftrightarrow \dfrac{2}{5}\left(u_m^2 + v_m^2\right) \leq u_mv_m \\ 
        &\Leftrightarrow 0 \leq (2v_m - u_m)(2u_m - v_m). 
    \end{aligned}
\end{equation}
Here, if $0\leq 1 - a -2b$, 
\begin{equation}
    2v_m - u_m = \dfrac{\left(\tfrac{b}{a} + 2\right)(1-a-2b)^m}{1 + \tfrac{2b}{a}} \geq 0, 
\end{equation}
and if $0\leq 1 - a -2b \leq 1$, 
\begin{equation}
    2u_m - v_m = \left(2 + \dfrac{2 - 2(1-a-2b)^m}{1 + \tfrac{2b}{a}}\right) - \left(\dfrac{1}{2} + \dfrac{\tfrac{1}{2} + \tfrac{b}{a}(1-a-2b)^m}{1 + \tfrac{2b}{a}}\right) \geq 2 - 1 = 1 \geq 0.  
\end{equation}
Hence, if $0\leq 1 - a -2b \leq 1$, $a^{(\gamma)}_j \leq a^{(\gamma)}_{j+1}$, which completes the proof. 
\end{enumerate}

\noindent In the proof, we used the following lemma. 
\begin{lemma}
    \label{lem:sequence_noise}
    Let $\mathcal{N}$ be a single qubit noise channel, and let $\mathsf{N} = \mathcal{N} \otimes \mathcal{N}$ be two copies of $\mathcal{N}$ acting on two qubits. 
Define a two-qubit operator
\begin{equation}
    \tilde{M}_{U_1, \mathsf{N}} = M_{U_1} \circ \mathsf{N} \circ M_{U_1}
\end{equation}
with 
\begin{equation}
M_{U_1}[\rho] = \underset{U_1 \sim \mathcal{U}_{\textup{Haar}}}{\mathbb{E}}\bigg[U_1^{\otimes 2} \rho U_1^{\dagger \otimes 2} \bigg]. 
\end{equation}
Suppose that 
\begin{align}
    \label{eq:noise1_appendix}
    \tilde{M}_{U_1, \mathsf{N}}(I_4) &= (1 - a)I_4 + 2a S, \\
     \label{eq:noise2_appendix}
    \tilde{M}_{U_1, \mathsf{N}}(S) &= bI_4 + (1 - 2b) S
\end{align}
with $a > 0$ and $b > 0$, 
where $I_4$ is the 2-qubit identity operator and $S$ is the 2-qubit SWAP gate. 
Then, 
\begin{align}
   \underbrace{\tilde{M}_{U_1, \mathsf{N}} \circ \tilde{M}_{U_1, \mathsf{N}} \circ \cdots \circ \tilde{M}_{U_1, \mathsf{N}}}_{m~\text{times}} \left(I_4\right) &= \left(1 - \dfrac{1 - (1-a-2b)^m}{1 + \tfrac{2b}{a}}\right)I_4 + \left(\dfrac{1 - (1-a-2b)^m}{\tfrac{1}{2} + \tfrac{b}{a}}\right)S, \\
   \underbrace{\tilde{M}_{U_1, \mathsf{N}} \circ \tilde{M}_{U_1, \mathsf{N}} \circ \cdots \circ \tilde{M}_{U_1, \mathsf{N}}}_{m~\text{times}} \left(S\right) &= \left(\dfrac{1}{2} - \dfrac{\tfrac{1}{2} + \tfrac{b}{a}(1-a-2b)^m}{1 + \tfrac{2b}{a}}\right)I_4 + \left(\dfrac{\tfrac{1}{2} + \tfrac{b}{a}(1-a-2b)^m}{\tfrac{1}{2} + \tfrac{b}{a}}\right)S.
\end{align}
\end{lemma}
\begin{proof}
    Let us write 
    \begin{align}
        \underbrace{\tilde{M}_{U_1, \mathsf{N}} \circ \tilde{M}_{U_1, \mathsf{N}} \circ \cdots \circ \tilde{M}_{U_1, \mathsf{N}}}_{m~\text{times}} \left(I_4\right) &= x_m I_4 + y_m S. 
    \end{align}
    Using Eqs.~\eqref{eq:noise1_appendix} and \eqref{eq:noise2_appendix}, 
    \begin{align}
        x_{m+1} &= (1-a)x_m + by_m, \\ 
        y_{m+1} &= 2a x_m + (1-2b)y_m. 
    \end{align}
    This is equivalent to 
    \begin{align}
        x_{m+1} + \dfrac{1}{2} y_{m+1} &= x_{m} + \dfrac{1}{2} y_{m}, \\ 
        x_{m+1} - \dfrac{b}{a} y_{m+1} &= (1-a-2b) \left(x_{m} - \dfrac{b}{a} y_{m}\right). 
    \end{align}
    Therefore, we have 
    \begin{align}
        x_{m} + \dfrac{1}{2} y_{m} &= x_{0} + \dfrac{1}{2} y_{0} = 1, \\ 
        x_{m} - \dfrac{b}{a} y_{m} &= (1-a-2b)^m \left(x_{0} - \dfrac{b}{a} y_{0}\right) = (1-a-2b)^m. 
    \end{align}
    Hence, 
    \begin{align}
    x_m &= 1 - \dfrac{1 - (1-a-2b)^m}{1 + \tfrac{2b}{a}},\\ 
    y_m &= \dfrac{1 - (1-a-2b)^m}{\tfrac{1}{2} + \tfrac{b}{a}}. 
\end{align}
\noindent With a very similar argument, letting us write 
    \begin{align}
        \underbrace{\tilde{M}_{U_1, \mathsf{N}} \circ \tilde{M}_{U_1, \mathsf{N}} \circ \cdots \circ \tilde{M}_{U_1, \mathsf{N}}}_{m~\text{times}} \left(S\right) &= z_m I_4 + w_m S,  
    \end{align}
we have 
\begin{align}
    z_m &= \dfrac{1}{2} - \dfrac{\tfrac{1}{2} + \tfrac{b}{a}(1-a-2b)^m}{1 + \tfrac{2b}{a}}, \\ 
    w_m &=  \dfrac{\tfrac{1}{2} + \tfrac{b}{a}(1-a-2b)^m}{\tfrac{1}{2} + \tfrac{b}{a}}. 
\end{align}
\end{proof}

\section{Effect of noiseless single qubit gates in the last layer}
\label{noiseless}
Let us consider a layer of single-qubit gates that immediately follows the last layer of noise. This is equivalent to arbitrarily locally rotating the measurement basis. Moreover, assume that these gates are noiseless. 
A single qubit gate $U$ is parametrized as
\begin{equation}
  U(\theta, \phi) = 
\begin{pmatrix}
&\cos(\theta) \cdot e^{i \phi} && \sin(\theta) \\
&-\sin(\theta) && \cos(\theta) \cdot e^{-i \phi}& \\
\end{pmatrix}.
\end{equation}
Let $U_i(\theta_i, \phi_i)$ be the unitary applied to the $i^{\text{th}}$ qubit in the last layer. 
\begin{thm}
\label{ampdamp1}
Let $\mathcal{B}$ be an ensemble of amplitude-damped random quantum circuits, with noise strength $q$. Additionally, before measurement, for every $i \in [n]$, let $U_i(\theta_i, \phi_i)$---a single qubit, noiseless gate---be applied to qubit $i$. Then,
\begin{equation}
\underset{\mathcal{B}}{\mathbb{E}}[p_{i, b}] = \frac{1}{2} + \frac{(-1)^{b} \cdot q\cos2\theta_i}{2},
\end{equation}
where $p_{i, b}$ is marginal probabilities of getting $b \in \{0, 1\}$, in the $i^{\text{th}}$ qubit.
\end{thm}
\begin{proof}
As per earlier convention, let $\mathcal{N}$ be the single-qubit noise channel, and let $K_0$ and $K_1$ be the corresponding Kraus operators. 
Let $\mathcal{C} \in \mathcal{B}$ be a noisy random circuit.  
Let $\tilde{\mathcal{C}}$ be the quantum circuit \emph{without the last layer}; that is, 
\begin{equation}
\mathcal{C} = \mathcal{N}^{\otimes n} \circ \tilde{\mathcal{C}}.
\end{equation}
For $i \in [n]$, define $\rho_i$ as follows:
\begin{equation}
\begin{aligned}
\rho_i = \underset{\mathcal{B}}{\mathbb{E}}[\Tr_{1, \ldots, i-1, i+1, \ldots n}(\tilde{\mathcal{C}}(\ketbra{0^n}{0^n}))],
\end{aligned}
\end{equation}
where $I_2$ is the single-qubit identity matrix. 
In other words, $\rho_i$ is the expected reduced density matrix on just the $i^{\text{th}}$ qubit.
By definition, we have 
\begin{equation}
    \rho_i = \frac{I_2}{2}. 
\end{equation}

\noindent In the last noiseless layer, since the parameters are implicit, we will drop the subscripts and refer to the $i^{\text{th}}$ single qubit gate as just $U_i$. 
\noindent Then, by the definition of the adjoint map, 
\begin{align}
\underset{\mathcal{B}}{\mathbb{E}}[p_{i, 0}] &=  \underset{\mathcal{B}}{\mathbb{E}}[ \Tr(\ket{0}\bra{0} U_i \mathcal{N}(\rho_i) U_i^{\dagger}) ] \\
&=  \underset{\mathcal{B}}{\mathbb{E}}[\Tr( U_i^{\dagger} \ket{0}\bra{0} U_i \mathcal{N}(\rho_i)) ]  \\
&= \underset{\mathcal{B}}{\mathbb{E}}[\Tr( U_i^{\dagger} \ket{0}\bra{0} U_i \mathcal{N}(\rho_i)) ] \\
\label{lasteq}
& =  \underset{\mathcal{B}}{\mathbb{E}}[\Tr( K_0^{\dagger}U_i^{\dagger} \ket{0}\bra{0} U_i K_0 ~\rho_i) ] + \underset{\mathcal{B}}{\mathbb{E}}[\Tr( K_1^{\dagger}U_i^{\dagger} \ket{0}\bra{0} U_i K_1 ~\rho_i) ],  
\end{align}
where we have repeatedly used the cyclic property of trace. 
In \cref{lasteq}, we used the Kraus operators $K_0$ and $K_1$ of the amplitude damping noise channel. 
Since $\rho_i = \tfrac{I_2}{2}$, \cref{lasteq} is equal to 
\begin{align}
\frac{1}{2} \left(\Tr( K_0^{\dagger}U_i^{\dagger} \ket{0}\bra{0} U_i K_0)  + \Tr( K_1^{\dagger}U_i^{\dagger} \ket{0}\bra{0} U_i K_1)\right).  
\end{align}
By explicitly computing each term, we have 
\begin{align}
\underset{\mathcal{B}}{\mathbb{E}}[p_{i, 0}] = \frac{1}{2}\left((\cos^2\theta_i + (1-q)\sin^2\theta_i) + q\cos^2\theta_i\right)
= \frac{1 + q(\cos^2\theta_i - \sin^2\theta_i)}{2} 
= \frac{1}{2} + \frac{q\cos2\theta_i}{2}. 
\end{align}

\end{proof}
\noindent Let us define the "bias" $\beta_i$, of qubit $i$ as
\begin{equation}
\beta_i = q\cos2\theta_i.
\end{equation}
Note that $\beta_i \in [-q, q]$. Hence,
\begin{equation}
\underset{\mathcal{B}}{\mathbb{E}}[p_{i, b}] = \frac{1}{2} + \frac{\beta_i}{2},~~~~~\underset{\mathcal{B}}{\mathbb{E}}[p_{i, b}] = \frac{1}{2} - \frac{\beta_i}{2}.
\end{equation}
\begin{remark}
One way to interpret \cref{ampdamp1} is to observe that there is an ``effective'' amplitude damping channel acting on each qubit, but each such channel has a "tunable" noise strength, which can now also be negative. 

A negative noise strength means that the outcome $0$ is suppressed and the outcome $1$ is assigned higher weight. However, \emph{which} strings are suppressed, and by \emph{how much} depend on the "bias" of each qubit. By changing the single qubit gates, we can control the bias.
\end{remark}

\section{Output distribution with a last layer of noiseless gates}
\label{last2}
\begin{thm}
    Let $\mathcal{B}$ be an ensemble of amplitude-damped random quantum circuits, with noise strength $q$. Additionally, before measurement, for every $i \in [n]$, let $U_i(\theta_i, \phi_i)$---a single qubit, noiseless gate---be applied to qubit $i$. 
    If 
    \begin{equation}
        4 - \sqrt{15} < |\cos 2\theta_i| 
    \end{equation}
    for all $i \in [n]$, 
    then there exists a string $x \in \{0,1\}^n$ such that for any $q > 0$ and $d= \Omega(\log n)$, 
\begin{equation}
\underset{n \rightarrow \infty}{\lim} \underset{\mathcal{B}}{\Pr}\left[p_x < \frac{\alpha}{2^n} \right] = 1,
\end{equation}
for any $\alpha \in (0, 1]$.
\end{thm}
\begin{proof}
    Following the same argument as in the proof of Theorem~\ref{second moment probabilities} in Appendix~\ref{appendix:proof_of_CR}, for sufficiently large depth and sufficiently large $n$, we have 
    \begin{equation}
        \underset{\mathcal{B}}{\mathbb{E}}[p_x^2] \leq \left(\left(\frac{1}{4-q}\right)^n (e_1e_2\cdots e_n)\right)\times \mathcal{O}(1)
    \end{equation}
    where 
    \begin{equation}
        \begin{aligned}
         e_j 
         &= \max\left\{\bra{x_j}\bra{x_j}(\mathcal{U}_j\otimes \mathcal{U}_j)\mathsf{N}(I_4)(\mathcal{U}_j\otimes \mathcal{U}_j)^\dagger\ket{x_j}\ket{x_j},\bra{x_j}\bra{x_j}(\mathcal{U}_j\otimes \mathcal{U}_j)\mathsf{N}(S)(\mathcal{U}_j\otimes \mathcal{U}_j)^\dagger\ket{x_j}\ket{x_j}\right\} \\ 
         &=\Bigg\{ 
        \begin{array}{cc}
            (1 + q\cos 2\theta_j)^2 & x_j = 0 \\
            (1 - q\cos 2\theta_j)^2 & x_j = 1
        \end{array}   
        \end{aligned}
    \end{equation}
    for $j = 1,2,\dots,n$. 
    Choose $x$ so that 
    \begin{equation}
        e_j =  (1 - q |\cos 2\theta_j|)^2. 
    \end{equation}
    In this case, 
    \begin{equation}
        \underset{\mathcal{B}}{\mathbb{E}}[p_x^2] \leq  \left(\frac{(1 - q |\cos 2\theta|)^2}{4-q}\right)^n \times \mathcal{O}(1), 
    \end{equation}
    where 
    \begin{equation}
        \theta \coloneqq \argmin_{\theta_j : j\in[n]}  |\cos2\theta_j|. 
    \end{equation}
    To see the concentration, it suffices to check if 
    \begin{align}
        4\frac{(1 - q |\cos 2\theta|)^2}{4-q} < 1, 
    \end{align}
    which is equivalent to 
    \begin{equation}
        q\left(4|\cos 2\theta|^2q -(8|\cos 2\theta| - 1) \right) < 0. 
    \end{equation}
    When $q = 0$, this inequality cannot be satisfied, so $q > 0$. 
    Under this condition, we have 
    \begin{equation}
        4|\cos 2\theta|^2q -(8|\cos 2\theta| - 1) < 0. 
    \end{equation}
    When 
    \begin{equation}
        4 - \sqrt{15} < |\cos 2\theta|, 
    \end{equation}
    all $q > 0$ satisfy this condition, which completes the proof. 
\end{proof}

\section{A note on the easiness of sampling results}
\label{Easiness_sampling}
In this section, we will sketch the easiness results of \cite{aharonov2022polynomialtime} and why anticoncentration is believed to be an important criterion for the current analysis techniques to go through. This is just a sketch, so we will not be too formal with the definitions and proofs.

Without loss of generality, let $\mathcal{C}$ be a unitary quantum circuit circuit. By $p_x(\mathcal{C})$, for $x \in \{0, 1\}^n$, let us denote the probability of getting string $x$ in the output distribution of $\mathcal{C} |0^n \rangle$. Using techniques from \cite{aharonov2022polynomialtime}, we can rewrite $p_x(\mathcal{C})$ in terms of Pauli paths as
\begin{equation}
p_x(\mathcal{C}) = \sum_{s} f(\mathcal{C}, x, s),
\end{equation}
where $f(\mathcal{C}, x, s)$ are as defined in \cite{aharonov2022polynomialtime} --- they are path weights for each Pauli path --- and $s$ is the number of non--identity Pauli terms for each Pauli path. Similarly, let
\begin{equation}
\tilde{p}_x(\mathcal{C}) = \sum_{s} \tilde{f}(\tilde{\mathcal{C}}, x, s)
\end{equation}
be the output probability, written as Pauli paths, for the noisy version of $\mathcal{C}$, denoted by $\tilde{\mathcal{C}}$. 
\noindent Now, to sample from the output distribution of $\tilde{\mathcal{C}} |0^n \rangle$, we consider a new distribution $\tilde{\tilde{q}}$ given by
\begin{equation}
\tilde{\tilde{q}}_x(\tilde{\mathcal{C}}) = \underset{s, |s| \leq l}{\sum} \tilde{f}(\tilde{\mathcal{C}}, x, s),
\end{equation}
where $l$ is some threshold that we choose. Consider an ensemble $\mathcal{B}$ of such noisy circuits. Now, let us calculate the total variation distance between $\tilde{p}$ and $\tilde{\tilde{q}}$.
\begin{align}
\underset{\mathcal{B}}{\mathbb{E}}\big[|\tilde{p}-\tilde{\tilde{q}}|_1^{2}\big] &\leq 2^n \cdot \underset{\mathcal{B}}{\mathbb{E}} \left[ \sum_{x \in \{0, 1\}^n} \big(\tilde{p}_x(\mathcal{C}) - \tilde{\tilde{q}}_x(\tilde{\mathcal{C}})\big)^2  \right] \\
&\leq 2^n \cdot \underset{\mathcal{B}}{\mathbb{E}}\left[ \sum_{x \in \{0, 1\}^n} \sum_{s, |s| > l} \tilde{f}(\tilde{\mathcal{C}}, x, s)^2 \right],
\end{align}
where the first line follows from the Cauchy-Schwarz inequality, and the second line follows from definitions. Note that for any choice of the cutoff $l$,
\begin{align}
\label{firsteq2}
& 2^n \cdot \underset{\mathcal{B}}{\mathbb{E}}\left[ \sum_{x \in \{0, 1\}^n} \sum_{s, |s| > l} \tilde{f}(\tilde{\mathcal{C}}, x, s)^2 \right] \\
& \leq 2^n \cdot \underset{\mathcal{B}}{\mathbb{E}}\left[  \sum_{x \in \{0, 1\}^n}  \sum_{s} \tilde{f}(\tilde{\mathcal{C}}, x, s)^2 \right] \\
& \leq 2^n \cdot  \underset{\mathcal{B}}{\mathbb{E}}\left[\sum_{x \in \{0, 1\}^n} \tilde{p}_x^2\right],
\end{align}
\noindent where the last line follows from orthogonality of Pauli paths in a random circuit.

\subsection{Special case of the depolarizing channel}
\label{depolarizing2}
For the special case of the depolarizing channel,
\begin{equation}
\label{Depolarizing}
\tilde{f}(\tilde{\mathcal{C}}, x, s) = (1 - q)^{|s|} f({\mathcal{C}}, x, s).
\end{equation}
So, for a choice of cutoff $l$, using similar steps as in the previous section, \cref{firsteq2} can be upper bounded with the quantity 
\begin{equation}
(1 - q)^{2 l} \cdot 2^{n} \cdot  \sum_{x \in \{0, 1\}^n} p_x^{2}.
\end{equation}
It is known that for sufficiently deep circuits \cite{dalzell}, the scaled noiseless collision probability,
\begin{equation}
2^{n} \cdot \underset{\mathcal{B}}{\mathbb{E}}\left[\sum_{x \in \{0, 1\}^n} p_x^{2}\right]
\end{equation}
is $\mathcal{O}(1)$. So, by appropriately choosing $l$, one can make the total variation distance an inverse polynomial or less. 

\begin{remark}
Note that \cref{Depolarizing} is true for the special case of the depolarizing channel, but is \emph{not} true in general. So, this analysis, where it suffices to look at the convergence of the noiseless collision probability because that can, essentially, be "factored out" of the actual expression and dealt with separately, \emph{does not} work in general.
\end{remark}

\begin{remark}
Note that one could also have directly upper bounded the noisy collision probability 
\begin{equation}
\label{cp_remark}
2^n \cdot  \sum_{x \in \{0, 1\}^n} \tilde{p}_x^2,
\end{equation}
for the depolarizing channel. Indeed, this is done in \cite{Deshpande_2022}. However, in \cite{Deshpande_2022}, the bound is $\mathcal{O}(1)$, and not as tight as the analysis in \cref{depolarizing2}. This is necessary for inverse polynomial closeness in total variation distance, but is not sufficient. So, we still need the techniques from \cref{depolarizing2} to prove our bound.
\end{remark}

\begin{remark}
In general, since noisy ensembles do not satisfy the form in \cref{Depolarizing}, directly bounding the noisy collision probability is the best that we can hope for, which may not always give us tight bounds. 
\end{remark}

\subsection{Lack of anticoncentration implies failure of proof technique}

\noindent When the noisy ensemble $\mathcal{B}$ is anticoncentrated, with respect to \cref{Definiton: anticoncentration}, then it means that
\begin{equation}
2^n \cdot  \underset{\mathcal{B}}{\mathbb{E}}\left[\sum_{x \in \{0, 1\}^n} \tilde{p}_x^2 \right] = \mathcal{O}(1),
\end{equation}
which implies,
\begin{equation}
\label{tvd}
\underset{\mathcal{B}}{\mathbb{E}}\big[|\tilde{p}-\tilde{\tilde{q}}|_1^{2}\big] \leq 2^n \cdot \underset{\mathcal{B}}{\mathbb{E}}\left[ \sum_{x \in \{0, 1\}^n} \sum_{s, |s| > l} \tilde{f}(\tilde{\mathcal{C}}, x, s)^2 \right] \leq 2^n \cdot  \underset{\mathcal{B}}{\mathbb{E}}\left[\sum_{x \in \{0, 1\}^n} \tilde{p}_x^2 \right] = \mathcal{O}(1).
\end{equation}
\clearpage
This alone does not guarantee that the classical sampler, from \cref{Easiness_sampling}, samples from a distribution that is inverse polynomially close, in total variation distance, to the actual distribution. However, the satisfaction of \cref{Definiton: anticoncentration} is a \emph{necessary condition} for present proof techniques to go through, in the following sense: if an ensemble were to exhibit lack of anticoncentration, according to \cref{Definiton: concentration}, then the quantity
\begin{equation}
2^n \cdot  \underset{\mathcal{B}}{\mathbb{E}}\left[\sum_{x \in \{0, 1\}^n} \tilde{p}_x^2 \right]
\end{equation}
diverges with $n$, and the chain of inequalities in \cref{tvd} does not hold, \emph{no matter where we choose the cutoff $l$}. This does not mean that there couldn't be better analysis techniques which do not need anticoncentration: however, to the best of our knowledge, no such technique is known.

\section{Easiness of computing expectation values}
\label{Expectation_Values}
In this section, we will, very broadly, sketch the argument of \cite{shao2023simulating} about computing the expectation value of certain observables, for random quantum circuits with the depolarizing noise. The ideas are extremely similar to that of \cref{Easiness_sampling}, but the argument does not require anticoncentration. To start with, note that for the depolarizing channel, 
\begin{equation}
\label{impeq}
\tilde{f}(\tilde{\mathcal{C}}, x, s) = (1 - q)^{|s|} f({\mathcal{C}}, x, s),
\end{equation}
where $f$ is a Pauli path of the noiseless circuit $\mathcal{C}$, $\tilde{f}$ is a Pauli path of the noisy circuit $\tilde{C}$, $x \in \{0, 1\}^n$, and $s \in \{0, 1\}^{2n}$ is the number of non--identity Pauli terms in the path. In \cite{shao2023simulating}, it is shown how the expectation value of any observable can be written as
\begin{equation}
\sum_{s} \tilde{f}(\tilde{\mathcal{C}}, x, s).
\end{equation}
Now, using \cref{impeq}, any path such that $s = \omega(\log n)$ is at least inverse--superpolynomially suppressed, and only polynomially many paths --- those for which $s = \mathcal{O}(\log n)$ --- have at least $\frac{1}{\text{poly}(n)}$ weight. Then, just by classically estimating those paths, we get an estimate of the expectation value that is $\approx \frac{1}{\text{poly}(n)}$ close to the original expectation value. Note that each Pauli path can be estimated in classical polynomial time. Note that this technique does not need anticoncentration.

\begin{remark}
Note that \cref{impeq} is a special property of the depolarizing channel: it is not evident whether the strategy described works for other noise channels.
\end{remark} 

\section{Anticoncentration and closeness to the uniform distribution}
\label{closeness to uniform}
In this section, we show the relation between anticoncentration, according to \cref{Definiton: anticoncentration}, and closeness to the uniform distribution. Let $\mathcal{B}$ be an ensemble of noisy random quantum circuits, and let $\mathcal{C}$ be the random variable corresponding to each circuit. Let
\begin{equation}
\rho = \mathcal{C}(\ket{0^n}\bra{0^n}).
\end{equation}
For $x \in \{0, 1\}^n$, let
\begin{equation}
\underset{\mathcal{B}}{\mathbb{E}}\left[p_x\right] = \underset{\mathcal{B}}{\mathbb{E}}\left[\Tr(|x\rangle \langle x| ~\rho)\right].
\end{equation}
\noindent Now, note that
\begin{align}
\underset{\mathcal{B}}{\mathbb{E}}\left[\sum_{x \in \{0, 1\}^n} \bigg| p_x - \frac{1}{2^n} \bigg|^2 \right]
&= \underset{\mathcal{B}}{\mathbb{E}}\left[\sum_{x \in \{0, 1\}^n} p_x^2 \right] - 2 \cdot \frac{1}{2^n} + \frac{1}{2^n}
\\&= \underset{\mathcal{B}}{\mathbb{E}}\left[\sum_{x \in \{0, 1\}^n} p_x^2 \right] - \frac{1}{2^n}.
\end{align}
\noindent Hence, if
\begin{equation}~\label{eq:flatness_condition}
\underset{\mathcal{B}}{\mathbb{E}}\left[\sum_{x \in \{0, 1\}^n} \bigg| p_x - \frac{1}{2^n} \bigg|^2 \right] = \mathcal{O}(2^{-n}),
\end{equation}
then,
\begin{equation}~\label{eq:anticoncentration_condition}
\underset{\mathcal{B}}{\mathbb{E}}\left[\sum_{x \in \{0, 1\}^n} p_x^2 \right] = \mathcal{O}(2^{-n}),
\end{equation}
which means \cref{Definiton: anticoncentration} is satisfied. The vice versa also holds.

\begin{remark}
Let $\Delta_n$ be the $n$-qubit completely dephasing channel with respect to the computational basis; that is, for any $n$-qubit state $\sigma$, 
\begin{equation}
    \Delta_n(\sigma) \coloneqq \sum_{x\in\{0,1\}^n} \bra{x}\sigma\ket{x} \ketbra{x}{x}. 
\end{equation}
\noindent Now, let us assume the following equation holds:
\begin{equation}~\label{eq:1norm_condition_original2}
    \underset{\mathcal{B}}{\mathbb{E}} \left\|\rho -  \frac{I_{2^n}}{2^n}\right\|_1^2 = \mathcal{O}(2^{-n}).
\end{equation}
\noindent Then, from monotonicity of trace distance, we get that
\begin{equation}~\label{eq:1norm_condition2}
    \underset{\mathcal{B}}{\mathbb{E}} \left\|\Delta_n(\rho) -  \frac{I_{2^n}}{2^n}\right\|_1^2 = \mathcal{O}(2^{-n}). 
\end{equation}
We used the fact that $\Delta_n(I_{2^n}) = I_{2^n}$. Then, since the 2-norm is smaller than the trace norm, from \cref{eq:1norm_condition2} we get 

\begin{equation}~\label{eq:2norm_condition}
    \underset{\mathcal{B}}{\mathbb{E}} \left\|\Delta_n(\rho) -  \frac{I_{2^n}}{2^n}\right\|_2^2 = \mathcal{O}(2^{-n}). 
\end{equation}
\noindent Hence, combining this observation with the calculations in \cref{closeness to uniform}, \cref{eq:2norm_condition} implies \cref{eq:anticoncentration_condition}.
\end{remark}
\end{document}